\documentclass[11pt,a4paper,reqno]{amsart}%
\usepackage{amsthm,amsmath,amsfonts,amssymb,xcolor,amsxtra,appendix,bookmark,dsfont,bm,mathrsfs}
\usepackage[latin1]{inputenc}
\usepackage{color} 
\usepackage{amssymb}\usepackage{graphicx}
\usepackage{tikz}
\usepackage{hyperref}
\theoremstyle{plain}
\newtheorem{theorem}{Theorem}
\newtheorem{lemma}[theorem]{Lemma}

\newtheorem{proposition}[theorem]{Proposition}

\newtheorem{definition}{Definition}

\theoremstyle{remark}
\newtheorem{remark}[theorem]{Remark}

\allowdisplaybreaks

\title[Evolution of finite temperature Bose-Einstein Condensates]{Evolution of finite temperature Bose-Einstein Condensates: Some rigorous studies on condensate growth}

\author[G. Staffilani]{Gigliola Staffilani
}
\address{Department of Mathematics, Massachusetts Institute of Technology, Cambridge, MA 02139, USA}
\email{gigliola@math.mit.edu} 
\thanks{G.S. is  funded in part by  the NSF grants DMS-2052651, DMS-2306378 and the Simons Foundation through the Simons Collaboration on Wave Turbulence.}

\author[M.-B. Tran]{Minh-Binh Tran}
\address{Department of Mathematics, Texas A\&M University, College Station, TX 77843, USA}
\email{minhbinh@tamu.edu} 
\thanks{M.-B. T is  funded in part by  a   Humboldt Fellowship,   NSF CAREER  DMS-2303146, and NSF Grants DMS-2204795, DMS-2305523,  DMS-2306379.}

\begin{document}
\date{\today}

\begin{abstract} 
In trapped Bose-Einstein condensates (BECs), \emph{condensate growth} refers to the process in which an increasing number of quasi-particles are immediately transferred from the non-condensate state (the thermal cloud) into the condensate state following the initial formation of the BEC. Despite its physical significance, this phenomenon has not yet been studied rigorously from a mathematical standpoint.

In this work, we investigate a kinetic equation whose collision operator includes three types of wave interactions: one corresponding to a 3-wave process, and two classified as 4-wave processes. This wave kinetic equation models the evolution of the density function of the thermal cloud. We establish the immediate formation of condensation in solutions to this equation, thus providing a rigorous demonstration of the condensate growth phenomenon.
\end{abstract}

\maketitle

 \tableofcontents

\section{Introduction}\label{intro} 
\subsection{The evolution of Bose-Einstein condensations after being formed  and condensate growth}\label{Subsec:ConGrowth}

The experiments on Bose-Einstein condensation (BEC) in dilute atomic gases~\cite{WiemanCornell,Ketterle,bradley1995evidence} have initiated a period of intense theoretical and experimental activity. In these experiments, rapid evaporative cooling lowers the temperature of the Bose gas below the condensation point--the BEC transition temperature--leading to the formation of a condensate. One of the most fundamental aspects of this phenomenon concerns the nonequilibrium growth of the condensate, which occurs after the BEC has initially formed at finite temperature. \footnote{Although the system temperature is extremely low below the BEC transition temperature, absolute zero cannot be reached experimentally. Consequently, the finite-temperature regime is of primary importance.
}

\textit{Condensate growth} refers to the process by which the condensate, once formed at finite temperature, immediately begins to grow at a finite rate~\cite{anglin2002bose,miesner1998bosonic}. Figure~\ref{fig1} illustrates this phenomenon. The image is adapted from Figure 5 in~\cite{bijlsma2000condensate}. The condensate initially contains $10^2$, $10^3$, $10^4$, $10^5$, or $10^6$ particles at time $t = 0$; these values serve as the initial conditions. The figure displays the growth in the number of particles in the condensate for $t > 0$. A similar illustration can be found in Figure 3 of~\cite{miesner1998bosonic}.

\begin{figure}[h!]
	\centering
	\includegraphics[scale=0.3]{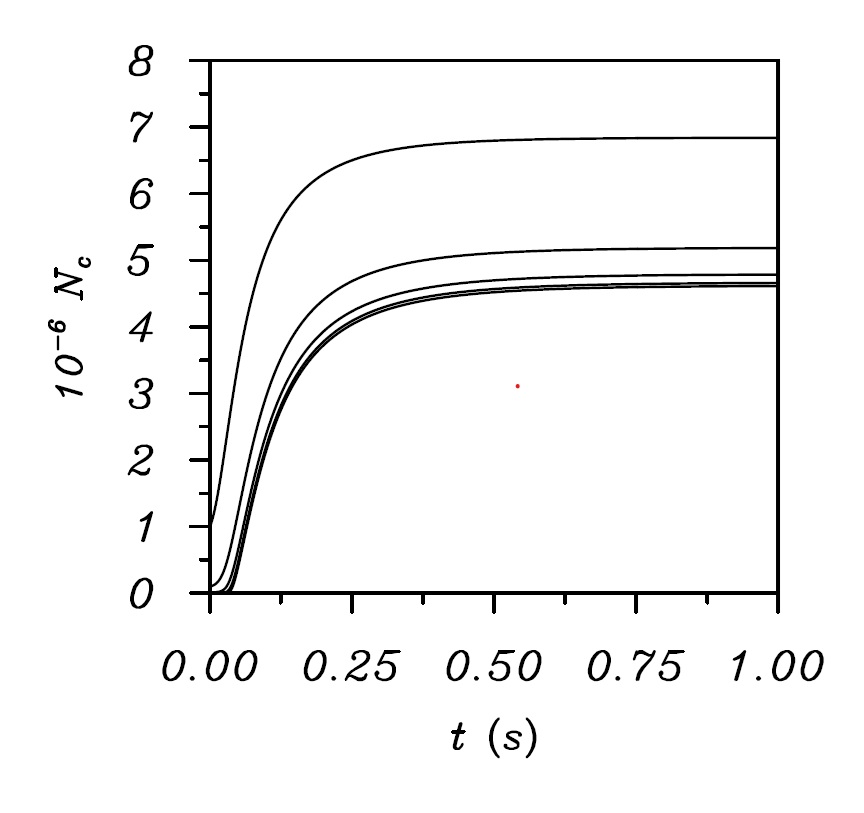}
\caption{Condensate growth curves for different initial numbers of condensed particles. At time $t = 0$, the condensate initially contains $10^2$, $10^3$, $10^4$, $10^5$, or $10^6$ particles. This is Figure 5 in~\cite{bijlsma2000condensate}.}

	\label{fig1}
\end{figure}

Understanding the theoretical foundation of condensate growth remains a central area of research \cite{anglin2002bose,miesner1998bosonic}. This  quantum kinetics problem has been studied in several perspectives, from  quantum optics to  condensed matter theory \cite{gardiner1998quantum,gardiner1997kinetics,kohl2002growth,moss2002formation}. 
 
Theoretically, the atomic states can be divided into two categories: the condensate band, which includes energy levels influenced by the ground-state condensate, and the non-condensate band, which encompasses all higher energy levels. The non-condensate band, often referred to as the thermal cloud, acts as a thermal reservoir containing the majority of atoms, and is responsible for supplying atoms that contribute to condensate growth.
At finite temperatures, the BEC coexists with a thermal cloud of non-condensed atoms. In this setting, two types of dynamical processes can be distinguished: scattering and growth. Scattering dynamics involve collisions between atoms within the same energy band, which do not lead to transitions between bands. In contrast, growth dynamics involve collisions between atoms in the thermal cloud that result in the transfer of atoms into the condensate band, thereby increasing the condensate fraction. The growth dynamics can be described by collisional kinetic theory.

While rigorous proofs of the formation of BECs have been established~\cite{erdHos2010derivation,lieb2002proof}, a mathematically rigorous proof of condensate growth  remains completely open, despite its significance and the availability of numerous numerical and experimental studies. Therefore, our main goal is to provide a first step toward establishing such a mathematical foundation.

To this end, we consider the following equation, which describes the dynamics of the density function of the thermal cloud $f(t,k)$

\begin{equation}
	\label{4wave}
	\partial_t f(t,k) = \mathscr Q[f](t,k) := C_{12}[f](t,k) + C_{22}[f](t,k) + C_{31}[f](t,k), \ \ \ \ f(0,k) \ = \ f_0(k) 
\end{equation}
where the collision operators are given by

\begin{align}
	\label{C12}
	\begin{split}
		C_{12}[f] :=\ &\mathfrak c_{12}\iint_{\mathbb{R}^3 \times \mathbb{R}^3} \mathrm{d}k_1\, \mathrm{d}k_2 \; \mathcal{W}_{12}(|k|,|k_1|,|k_2|) \, \delta(\omega - \omega_1 - \omega_2) \, \delta(k - k_1 - k_2) \left[ f_1 f_2 - (f_1 + f_2) f \right] \\
		&- 2 \iint_{\mathbb{R}^3 \times \mathbb{R}^3} \mathrm{d}k_1\, \mathrm{d}k_2 \; \mathcal{W}_{12}(|k_1|,|k|,|k_2|) \, \delta(\omega_1 - \omega - \omega_2) \, \delta(k_1 - k - k_2)  \left[ f f_2 - (f + f_2) f_1 \right],
	\end{split}
\end{align}

\begin{equation}
	\label{C22}
	\begin{aligned}
		C_{22}[f] :=\ &\mathfrak c_{22}\iiint_{\mathbb{R}^{3 \times 3}} \mathrm{d}k_1\, \mathrm{d}k_2\, \mathrm{d}k_3 \; \mathcal{W}_{22}(|k|,|k_1|,|k_2|,|k_3|) \delta(k + k_1 - k_2 - k_3) \, \delta(\omega + \omega_1 - \omega_2 - \omega_3) \\
		&\quad \times \left[ f_2 f_3 (f_1 + f) - f f_1 (f_2 + f_3) \right],
	\end{aligned}
\end{equation}

\begin{align}
	\label{C31}
	\begin{split}
		C_{31}[f] :=\ &\mathfrak c_{31}\iiint_{\mathbb{R}^{3 \times 3}} \mathrm{d}k_1\, \mathrm{d}k_2\, \mathrm{d}k_3 \; \mathcal{W}_{31}(|k|,|k_1|,|k_2|,|k_3|)  \delta(\omega - \omega_1 - \omega_2 - \omega_3) \, \delta(k - k_1 - k_2 - k_3) \\
		&\quad \times \left[ f_1 f_2 f_3 - (f_1 f_2 + f_2 f_3 + f_1 f_3) f \right] \\
		&- 3 \iiint_{\mathbb{R}^{3 \times 3}} \mathrm{d}k_1\, \mathrm{d}k_2\, \mathrm{d}k_3 \; \mathcal{W}_{31}(|k_1|,|k|,|k_2|,|k_3|) \delta(\omega_1 - \omega - \omega_2 - \omega_3) \, \delta(k_1 - k - k_2 - k_3) \\
		&\quad \times \left[ f f_2 f_3 - (f f_2 + f_2 f_3 + f f_3) f_1 \right],
	\end{split}
\end{align}
where $\mathfrak{c}_{12}, \mathfrak{c}_{22}, \mathfrak{c}_{31}$ are positive constants, $t \in \mathbb{R}_+$ denotes the time variable,  $k \in \mathbb{R}^3$ is the three-dimensional momentum variable, $\omega(k)=\omega(|k|)$ is the dispersion relation of the quasi-particles, and $f_0(k)=f_0(|k|)\ge0$ is the initial density function of the thermal cloud. We assume that both the dispersion relation $\omega$ and the solution $f$ are radial. We refer the readers to Assumption X below for the more precise assumptions on $\omega$. Thus, we can identify $f(k)$ with $f(|k|)$ and with $f(\omega)$. We have used the shorthand notations:
\begin{equation}\label{Shorthand}	\begin{aligned}
& f = f(k) = f(|k|) = f(\omega),\quad
f_1 = f(k_1) = f(|k_1|) = f(\omega_1),\quad
f_2 = f(k_2) = f(|k_2|) = f(\omega_2),\\
&\text{and likewise for } \omega = \omega(k) = \omega(|k|)	\text{ etc.}\end{aligned}
\end{equation}

Following \cite{soffer2020energy}, we set the collision kernel $\mathcal{W}_{12}(|k|,|k_1|,|k_2|)$ to be
\begin{equation}
	\label{W21}
	\mathcal{W}_{12}(|k|,|k_1|,|k_2|) = |k|\,|k_1|\,|k_2|.
\end{equation}

Following \cite{staffilani2024energy,staffilani2024condensation}, we set
\begin{equation}
	\label{W22}
	\mathcal{W}_{22}(|k|,|k_1|,|k_2|,|k_3|) = \left[ \frac{\max\{|\omega - \omega_2|,\  |\omega_1 - \omega_2|,\ |\omega - \omega_3|,\  |\omega_1 - \omega_3|\}}{\omega + \omega_1 + \omega_2 + \omega_3} \right]^\mu,
\end{equation}
for $\mu \ge 0$, where the case $\mu = 0$ corresponds to the one studied in \cite{staffilani2024energy,staffilani2024condensation}.

Finally, we define, as a simplification of the kernel computed in \cite{tran2020boltzmann}
\begin{equation}
	\label{W31}
	\mathcal{W}_{31}(|k|,|k_1|,|k_2|,|k_3|) = \left|\,|k| - |k_1| - |k_2| - |k_3|\,\right|^{-1}|k|\,|k_1|\,|k_2|\,|k_3|.
\end{equation}

We refer to Subsection~\ref{Subsec:PhysicalContext} for the physical context of this model.  The BEC can be mathematically interpreted as a Dirac mass \( n \delta_{k=0} \), where \( n \) denotes the condensate density. In the condensate growth process, the initial density \( f_0(k) \) is a regular function which does not contain any Dirac mass. However, for an immediate time \( t > 0 \), the solution \( f(t,k) \) immediately develops a Dirac mass, representing the accumulation of quasi-particles from the thermal cloud into the condensate, thereby marking the onset of condensate growth.

The main goal of this work is to rigorously prove the onset of condensate growth by demonstrating the immediate formation of a Dirac delta function at an immediate time \( t > 0 \) in the solution to the thermal cloud equation~\eqref{4wave}, starting from a regular initial condition.
In other words, we initially assume that  (see \eqref{Radon} for the precise definition of the space of solutions)
\begin{equation}\label{Ini}
	\int_{\{0\}}\mathrm{d}k\, f_0(k) = 0,
\end{equation}
and define  
\begin{equation}\label{T0}
	T_0 := \sup\left\{T ~\middle|~ \int_{\{0\}}\mathrm{d}k\, f(t,k) = 0 \quad \text{for all } t \in [0,T) \right\}.
\end{equation}

Our goal is to prove that \( T_0 = 0 \) under the following  assumptions on the parameters of the equation, and hence, establish the first rigorous study of condenstate growth.

\bigskip

\textbf{Assumption A:} \\
\textit{
	\begin{itemize}
		\item[(A1)] There exist constants \(  1 < 1/\delta', 1/\delta \le 2 \), and \( 0 < C_{\mathrm{disper}}', C_{\mathrm{disper}} \) such that the dispersion relation satisfies
		\begin{equation}\label{CondDisper}
			C_{\mathrm{disper}}' |k|^{1/\delta'} \ge \omega(k) \ge C_{\mathrm{disper}} |k|^{1/\delta}.
		\end{equation}
	The function \(\omega(|k|)\) is continuous and satisfies \(\omega'(|k|) \ge 0\) for all \( k \in \mathbb{R}^3 \).	Moreover,
		\begin{equation}\label{CondDisper1}
			\omega(|k_1| + |k_2|) \ge \omega(|k_1|) + \omega(|k_2|), \quad \forall\, k_1, k_2 \in \mathbb{R}^3.
		\end{equation}
		For each \(\omega\), there exists a unique value of \(|k|\) such that \(\omega(|k|) = \omega\). We denote this dependence explicitly by writing \(|k| = |k|(\omega)\). 		In addition, we assume $\omega(0)=0$.
		\item[(A2)] We assume that the solution \( f \) is radial, i.e., \( f(k) = f(|k|) \). Then, by a change of variables, we have
		\begin{equation}
			\int_{\mathbb{R}^3} \mathrm{d}k\, f(k) = 2\pi^2 \int_{[0,\infty)} \mathrm{d}\omega\, \frac{|k|^2}{\omega'(|k|)} f(\omega).
		\end{equation}
		Therefore, we can identify \( f(k) \) with \( f(\omega) \) and write \( f(k) = f(\omega) \).
			\item[(A3)] There exist constants \( 0 < \theta\le 1, C_\mathfrak{A} \) such that
		\begin{equation}\label{CondA}
			0 \le \mathfrak{A}(\omega) := \frac{|k|^2}{\omega'(|k|)}  \le C_\mathfrak{A} \omega^\theta, \quad \text{and} \quad \mathfrak{A}(0) = 0.
		\end{equation}
		Moreover,   \( \mathfrak{A}(|k|) \) is a non-decreasing function  \( |k| \).	
		\item[(A4)] There exist constants \( 0 \le \varrho\le 1 \), and \( 0 < C_\Theta', C_\Theta \) such that
		\begin{equation}\label{CondTheta}
			C_\Theta' \omega^\varrho \le \Theta(\omega) := \frac{|k|}{\omega'(|k|)} \le C_\Theta \omega^\varrho.
		\end{equation}
		Moreover, the function \( \Theta(\omega) \) is continuous and non-decreasing with respect to \( \omega \).	
		\item[(A5)] The constants satisfy the following constraint:
		\begin{equation}\label{Parameters}
			2\delta - \tfrac{1}{2} > \varrho.
		\end{equation}
	\end{itemize}
}

\begin{remark}
	As an example, consider the dispersion relation \( \omega(k) = |k|^2 \). Then
	\[
	\Theta(\omega) = \frac{|k|}{\omega'(|k|)} = \frac{1}{2}, \quad \text{and} \quad \mathfrak{A}(\omega) = \frac{|k|^2}{\omega'(|k|)} = \frac{|k|}{2}.
	\]
	In this case, we have \( \varrho = 0 \), \( \delta =\delta'= \tfrac{1}{2} \), and \( \theta = \frac12 \). 
\end{remark}

\subsection{Physical context of the model}\label{Subsec:PhysicalContext}
In the pioneering works \cite{KD1,KD2}, Kirkpatrick and Dorfman pioneered to attack the complex problem of formulating the kinetic equation for a gas of particles outside the condensate. This line of research was later extended by Zaremba, Nikuni, and Griffin in \cite{ZarembaNikuniGriffin:1999:DOT}, where they derived a fully coupled system consisting of a quantum Boltzmann equation for the density function of the normal fluid (thermal cloud) and a Gross-Pitaevskii equation for the wavefunction of the Bose-Einstein condensate (BEC). Independently, Pomeau, Brachet, Metens, and Rica also developed a similar model in \cite{PomeauBrachetMetensRica}.

All of these models involve two types of collision operators, whose explicit forms will be provided later.

\begin{itemize}
	\item The $\mathscr{C}_{22}$ collision operator models $2 \leftrightarrow 2$ interactions among the excited atoms themselves.
	\item The $\mathscr{C}_{12}$ collision operator models $1 \leftrightarrow 2$ interactions between the condensate and the excited atoms.
\end{itemize}
A completely new collision operator, $\mathscr{C}_{31}$, which was previously absent from existing models, was proposed by Reichl and Gust in \cite{ReichlGust:2012:CII}. This operator accounts for $1 \leftrightarrow 3$ type collisions between excitations. The existence of this new collision $\mathscr{C}_{31}$ remained an open problem for over a decade, until it was mathematically confirmed in \cite{tran2020boltzmann}. Experimental evidences of $\mathscr{C}_{31}$ have also been provided in \cite{reichl2019kinetic}. 
For further discussions on this topic, we refer the reader to \cite{GriffinNikuniZaremba:BCG:2009, PomeauBinh, tran2021thermal}.

The full kinetic equation, which includes all three quantum kinetic collision operators, reads
\begin{equation}
	\label{KineticFinal}
	\partial_t f(t,k) = \mathscr{C}_{12}[f](k) + \mathscr{C}_{22}[f](t,k) + \mathscr{C}_{31}[f](t,k), \ \ \ \ f(0,k) \ = \ f_0(k) 
\end{equation}
where the forms of the quantum kinetic collision operators $\mathscr{C}_{12}$, $\mathscr{C}_{22}$, and $\mathscr{C}_{31}$ are given below:

\begin{equation}
	\label{C12Discrete}
	\begin{aligned}
		\mathscr{C}_{12}[f](t,k) =\; & 4\pi g^2 n \iiint_{\mathbb{R}^3 \times \mathbb{R}^3 \times \mathbb{R}^3} \mathrm{d}k_1\, \mathrm{d}k_2\, \mathrm{d}k_3 \; 
		\left[\delta(k - k_1) - \delta(k - k_2) - \delta(k - k_3)\right] \\
		& \times \delta(\omega(k_1) - \omega(k_2) - \omega(k_3)) \, (K_{1,2,3}^{1,2})^2 \, \delta(k_1 - k_2 - k_3) \\
		& \times \left[ f(k_2) f(k_3) (f(k_1) + 1) - f(k_1)(f(k_2) + 1)(f(k_3) + 1) \right],
	\end{aligned}
\end{equation}

\begin{equation}
	\label{C22Discrete}
	\begin{aligned}
		\mathscr{C}_{22}[f](t,k) =\; & \pi g^2 \iiiint_{\mathbb{R}^3 \times \mathbb{R}^3 \times \mathbb{R}^3 \times \mathbb{R}^3} \mathrm{d}k_1\, \mathrm{d}k_2\, \mathrm{d}k_3\, \mathrm{d}k_4 \;
		\left[\delta(k - k_1) + \delta(k - k_2) - \delta(k - k_3) - \delta(k - k_4)\right] \\
		& \times (K_{1,2,3,4}^{2,2})^2 \, \delta(k_1 + k_2 - k_3 - k_4) \, \delta(\omega(k_1) + \omega(k_2) - \omega(k_3) - \omega(k_4)) \\
		& \times \left[ f(k_3) f(k_4) (f(k_1) + 1)(f(k_2) + 1) - f(k_1) f(k_2) (f(k_3) + 1)(f(k_4) + 1) \right],
	\end{aligned}
\end{equation}

\begin{equation}
	\label{C31Discrete}
	\begin{aligned}
		\mathscr{C}_{31}[f](t,k) =\; & \pi g^2 \iiiint_{\mathbb{R}^3 \times \mathbb{R}^3 \times \mathbb{R}^3 \times \mathbb{R}^3} \mathrm{d}k_1\, \mathrm{d}k_2\, \mathrm{d}k_3\, \mathrm{d}k_4 \;
		\left[\delta(k - k_1) - \delta(k - k_2) - \delta(k - k_3) - \delta(k - k_4)\right] \\
		& \times (K_{1,2,3,4}^{3,1})^2 \, \delta(k_1 - k_2 - k_3 - k_4) \, \delta(\omega(k_1) - \omega(k_2) - \omega(k_3) - \omega(k_4)) \\
		& \times \left[ f(k_2) f(k_3) f(k_4)(f(k_1) + 1) - f(k_1)(f(k_2) + 1)(f(k_3) + 1)(f(k_4) + 1) \right],
	\end{aligned}
\end{equation}
where $n$ is the condensate density and $g$ is the interaction constant. The quantities $(K_{1,2,3}^{1,2})^2$, $(K_{1,2,3,4}^{2,2})^2$, and $(K_{1,2,3,4}^{3,1})^2$ are explicit collision kernels.

 In general, the dynamics of the BEC can be described by the Gross-Pitaevskii equation, from which the condensate density \( n \) can be deduced as a function of time \( t \). However, since equation~\eqref{KineticFinal} is already quite complex, we assume that the condensate fraction is sufficiently large and therefore treat \( n \) as a constant. For the full dynamical coupling system, in which the full equation for $n$ is written, we refer the readers to \cite{PomeauBinh} (see also Figure \ref{fig2}). 
 In this case, we define the constants
\begin{equation}\label{Simpl1}
\mathfrak c_{12} := 4\pi g^2 n, \qquad \mathfrak c_{22} := \pi g^2, \qquad \mathfrak c_{31} := \pi g^2,
\end{equation}
to simplify notation in the collision operators.

 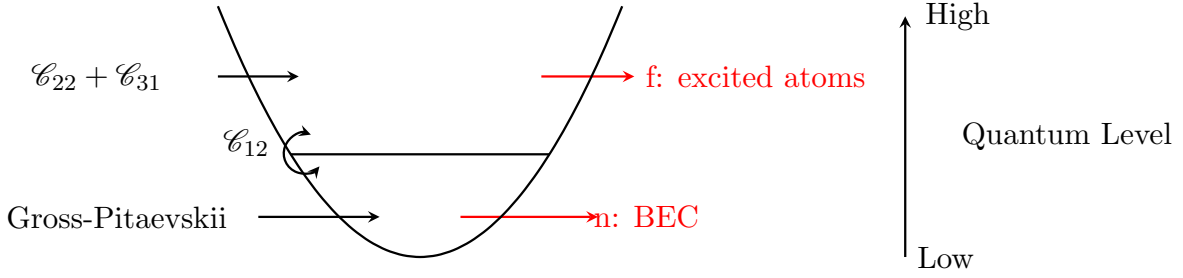
\begin{figure}
	\centering
	\resizebox{\textwidth}{!}{%
		\begin{tikzpicture}[>=stealth]
			\draw[thick, ->] (-2.5,2.25) -- ++(1,0) node[xshift=-2.5cm] {$\mathscr C_{22}+\mathscr C_{31}$};
			\draw[thick, ->, red] (1.5,2.25) -- ++(1.15,0) node[xshift=1.5cm] {f: excited atoms};
			\draw[thick, ->] (-2,.5) -- ++(1.5,0) node[xshift=-3.25cm] {Gross-Pitaevskii};
			\draw[thick, ->, red] (0.5,.5) -- ++(1.7,0) node[xshift=.6cm] {n: BEC};
			\draw[thick, ->] (6,0) node [xshift=.5cm]{Low} -- ++(0, 1.5)  node[xshift=2.0cm] {Quantum Level} -- ++(0,1.5) node[xshift=.65cm] {High};
			\draw[ thick,domain=-2.5:2.5,smooth,variable=\x,black] plot ({\x},{.5*\x*\x});
			\draw [thick, black] (-1.6, 1.28) -- (1.6, 1.28);
			
			\draw[thick, <->] (-1.35,1.525) arc (70:320:.25);
			\node at (-2.15, 1.425) {$\mathscr C_{12}$};
		\end{tikzpicture}
	}%
	\caption{The Bose--Einstein Condensate (BEC) and the excited atoms.}\label{fig2}
\end{figure}

Performing a common simplification by retaining only the higher-order terms and omitting the lower-order terms in  
\[
f(k_2) f(k_3) - f(k_1)(f(k_2) + f(k_3) + 1),
\]
\[
f(k_3) f(k_4) (f(k_1) + 1)(f(k_2) + 1) - f(k_1) f(k_2) (f(k_3) + 1)(f(k_4) + 1),
\]
\[
f(k_2) f(k_3) f(k_4)(f(k_1) + 1) - f(k_1)(f(k_2) + 1)(f(k_3) + 1)(f(k_4) + 1),
\]
we approximate them by
\begin{equation}\label{Simpl2}
f(k_2) f(k_3) - f(k_1)(f(k_2) + f(k_3)),
\end{equation}
\begin{equation}\label{Simpl3}
f(k_3) f(k_4)(f(k_1) + f(k_2)) - f(k_1) f(k_2)(f(k_3) + f(k_4)),
\end{equation}
\begin{equation}\label{Simpl4}
f(k_2) f(k_3) f(k_4) - \left( f(k_2) f(k_3) + f(k_2) f(k_4) + f(k_3) f(k_4) \right).
\end{equation}
Combining the simplifications \eqref{Simpl1}--\eqref{Simpl4}, and replacing the collision kernels $(K_{1,2,3}^{1,2})^2$, $(K_{1,2,3,4}^{2,2})^2$, and $(K_{1,2,3,4}^{3,1})^2$ with the much simpler forms $\mathcal{W}_{12}$, $\mathcal{W}_{22}$, and $\mathcal{W}_{31}$, respectively, we reduce \eqref{KineticFinal} to \eqref{4wave}.

The 4-wave collision operator \( C_{31} \) (and its full form \( \mathscr{C}_{31} \)) is highly sophisticated. Consequently, numerical simulations carried out by physicists have so far only been performed for \( C_{22} \) and \( C_{12} \), along with their full forms \( \mathscr{C}_{22} \) and \( \mathscr{C}_{12} \) (see \cite{bijlsma2000condensate,ZarembaNikuniGriffin:1999:DOT}). In contrast, the simulation of \( C_{31} \) and \( \mathscr{C}_{31} \) remains an entirely open problem. To ensure that our mathematical results apply to both the former model--comprising only $\mathscr C_{12}$ and $\mathscr C_{22}$--and the more complete model that includes all three collision operators $\mathscr C_{12}$, $\mathscr C_{22}$, and $\mathscr C_{31}$, we impose the following assumption:

\bigskip

\noindent\textbf{Assumption B.} \textit{We assume that} $\mathfrak c_{12}, \mathfrak c_{22} > 0$ \textit{while} $\mathfrak c_{31} \ge 0$.

\begin{remark}
	\textbf{Assumption B} implies that it is possible to set $\mathfrak c_{31} = 0$ in our mathematical results.
\end{remark}

\subsection{3-wave and 4-wave kinetic operators}

Wave turbulence theory, rooted in the works of Peierls \cite{Peierls:1993:BRK,Peierls:1960:QTS}, Brout and Prigogine \cite{brout1956statistical}, Zaslavskii and Sagdeev \cite{zaslavskii1967limits}, Hasselmann \cite{hasselmann1962non,hasselmann1974spectral}, Benney, Saffman, and Newell \cite{benney1969random,benney1966nonlinear}, and Zakharov \cite{zakharov2012kolmogorov}, describes the dynamics of nonlinear waves out of thermal equilibrium and has played a crucial role in both theoretical and applied physics.
In wave kinetic theory, there are two main types of kinetic collision operators: 3-wave and 4-wave interactions.
Following this definition, \( C_{12} \) is classified as a 3-wave collision operator, while \( C_{22} \) and \( C_{31} \) are defined as 4-wave  operators.

The dynamics of a dilute Bose gas at room temperatures can be described by the so-called Boltzmann-Nordheim equation~\cite{Nordheim}:
\begin{equation}
	\label{Nordheim}
	\partial_t F(t,k) = \mathscr{C}_{22}[F](t,k), \quad \text{with } \mu = 0,
\end{equation}
which simplifies to the following 4-wave kinetic equation:
\begin{equation}
	\label{4wavepre}
	\partial_t F(t,k) = C_{22}[F](t,k), \quad \text{with } \mu = 0.
\end{equation}
Note that $\mu$ is given in \eqref{W22}. 
It is well known in the physics literature~\cite{josserand2006self,PomeauBinh,Spohn:2010:KOT} that solutions to equations~\eqref{Nordheim} and~\eqref{4wavepre} can develop a singularity in finite time, which can be interpreted as the formation of a Bose-Einstein condensate in a certain sense. The formation of finite-time condensates for the 4-wave kinetic equation
in the special case $\omega(k) = |k|^2$ and $\mu = 0$, under the isotropic assumption $F(t,k) = F(t,\omega)$, has been studied in \cite{EscobedoVelazquez:2015:FTB,EscobedoVelazquez:2015:OTT}. These works rely on the fundamental convexity property:
\begin{equation}
	\label{4wavepre:convex1}
	\int_{\mathbb{R}_+} \mathrm{d}\omega\, C_{22}[F](\omega)\, \mathfrak{A}(\omega)\, \phi(\omega) \ge 0,
\end{equation}
for any convex test function $\phi(\omega)$. In particular, choosing
\[
\phi(\omega) = \left[1 - \frac{\omega}{\epsilon}\right]_+ = \max\left\{1 - \frac{\omega}{\epsilon},\, 0\right\},\ \ \ \epsilon>0,
\]
yields
\begin{equation}
	\label{4wavepre:convex2}
\partial_t \int_{\mathbb{R}_+} \mathrm{d}\omega\, F(t,\omega)\, \mathfrak{A}(\omega)\, \left[1 - \frac{\omega}{\epsilon}\right]_+ \ge 0,
\end{equation}
which indicates that the mass of \( F \) becomes increasingly concentrated near the origin as time evolves and $\epsilon$ tends to $0$.

As discussed in Subsection~\ref{Subsec:ConGrowth}, unlike the dynamics of dilute Bose gases at room temperature, which involve only $2 \leftrightarrow 2$ interactions, the dynamics of trapped Bose gases at finite temperature, governed by equation~\eqref{4wave}, must incorporate three types of interactions: $1 \leftrightarrow 2$, $2 \leftrightarrow 2$, and $3 \leftrightarrow 1$.
However, the expressions
\begin{equation}
	\label{4wavepre:convex3}
\int_{\mathbb{R}_+} \mathrm{d}\omega\, C_{12}[F](\omega)\, \mathfrak{A}(\omega)\, \phi(\omega)
\quad \text{and} \quad
\int_{\mathbb{R}_+} \mathrm{d}\omega\, C_{31}[F](\omega)\, \mathfrak{A}(\omega)\, \phi(\omega)
\end{equation}
do not have a definite sign. Even when choosing $\phi = 1$, the integrals exhibit bad signs:
\begin{equation}
	\label{4wavepre:convex4}
	\int_{\mathbb{R}_+} \mathrm{d}\omega\, C_{12}[F](\omega)\, \mathfrak{A}(\omega) \le 0,
	\quad \text{and} \quad
	\int_{\mathbb{R}_+} \mathrm{d}\omega\, C_{31}[F](\omega)\, \mathfrak{A}(\omega) \le 0.
\end{equation}

Moreover, the presence of the kernel \eqref{W22} also prevents the concentration of mass near the origin, as it regularizes the dynamics whenever any of the quantities $\omega, \omega_1, \omega_2, \omega_3$ approach zero.

 As a consequence, the strategy introduced in \cite{EscobedoVelazquez:2015:FTB, EscobedoVelazquez:2015:OTT} fails in the general case where all three types of collision operators, namely, $1 \leftrightarrow 2$, $2 \leftrightarrow 2$, and $1 \leftrightarrow 3$, and the new kernel \eqref{W22} are present.

Extensions of the results in \cite{EscobedoVelazquez:2015:FTB, EscobedoVelazquez:2015:OTT} to more general classes of dispersion relations and initial conditions have been recently developed in \cite{staffilani2024energy, staffilani2024condensation}. In these works, a novel strategy based on domain decomposition techniques \cite{halpern2009nonlinear, Lions:1989:OSA, toselli2004domain} is introduced. In the present work, we fully extend this methodology. The central idea is to partition the half-line $\mathbb{R}_+$ into small subdomains, enabling a divide-and-conquer approach. The mass concentration in each subdomain is then carefully compared, yielding refined estimates for accumulation near the origin. This framework further allows for a quantitative comparison of mass concentration under the combined influence of  collision operators $C_{12}$, $C_{31}$  and $C_{22}$, as well as the new kernel \eqref{W22}. The approach is  effective even when the rough estimate \eqref{4wavepre:convex2} fails to hold, due to the indefinite sign in \eqref{4wavepre:convex3}, and when the new kernel \eqref{W22} prevents mass concentration near the origin (see Remark \ref{RemarkMainTheo}). We note that some numerical simulations of the system concerning $C_{12}$ and $C_{22}$ have been done in \cite{das2025energy}.

In a companion paper~\cite{staffilani2025energyfinite}, we investigate the question of finite-time energy cascade for the kinetic equation~\eqref{4wave}.

We now briefly summarize what is generally known regarding the analysis of 4-wave kinetic equations of the type \eqref{4wavepre}, as well as 3-wave kinetic equations:

\begin{itemize}
	
	\item For 4-wave kinetic equations of the $2 \leftrightarrow 2$ type: convergence rates of discrete solutions and local well-posedness for the MMT model, a one-dimensional 4-wave kinetic equation, have been investigated in \cite{dolce2024convergence} and \cite{germain2023local}. Questions concerning near equilibrium stability, scattering theory, as well as the stability and cascade behavior of the Kolmogorov-Zakharov spectrum (a stationary solution), have been explored in \cite{menegaki20222, escobedo2024instability, ampatzoglou2024scattering} and \cite{collot2024stability}. Local well-posed and ill-posed results for $4$-wave kinetic equations and the hierarchy in polynomially weighted $L^\infty$ spaces for inhomogeneous 4-wave kinetic equations have also been studied in \cite{ampatzoglou2025ill,ampatzoglou2025optimal,GermainIonescuTran} and \cite{ampatzoglou2025inhomogeneous}.
	
	\item For 3-wave kinetic equations, we refer to the works \cite{GambaSmithBinh} on stratified flows in the ocean; \cite{cortes2020system, EPV, escobedo2023linearized1, escobedo2023linearized, escobedo2025local,nguyen2017quantum} on Bose-Einstein condensates; \cite{AlonsoGambaBinh, CraciunBinh, EscobedoBinh, GambaSmithBinh, tran2020reaction} on phonon interactions in anharmonic crystal lattices; \cite{das2024numerical, das2025energy,nguyen2017quantum, soffer2020energy, walton2022deep, walton2023numerical, walton2024numerical} on capillary waves; and \cite{rumpf2021wave} on beam waves. The formation of condensates for non-radial solutions to 3-wave kinetic equations has also been studied in \cite{staffilani2025formation} by the authors.
			\item The existence of classical solutions to the quantum kinetic equation, which includes the two collision operators $\mathscr{C}_{12}$ and $\mathscr{C}_{22}$, has been studied in \cite{soffer2018dynamics}.

	\item 4-wave kinetic equations on the torus have also been recently investigated in \cite{escobedo2024entropy, germain2024stability}. The 4-wave kinetic collision operator is closely connected to the Nordheim collision operator. A substantial body of work on the Nordheim equation has been developed in \cite{escobedo2003homogeneous,escobedo2007fundamental,Lu2014_RegularityCondensation,Lu2016_LongTimeBEC,Lu2018_LongTimeStrongBE}.

\end{itemize}

We now lay out the plan of the paper. In Section~\ref{Sec:Setting}, we provide the necessary definitions and state our main results. Sections~\ref{Sec:WeakFormulation} and~\ref{Sec:Estimates} are devoted to the weak formulations and the \emph{a priori} estimates. The global existence result is proved in Section~\ref{Sec:Global}. Section~\ref{Sec:Multiscale} describes the strategy for proving the condensate growth phenomenon and outlines the proofs in Sections~\ref{Sec:First}, \ref{Sec:Second}, \ref{Sec:Third}, and~\ref{Sec:CondensateGrowth}. The final section contains the proof of the main theorem.

\section{Main results}\label{Sec:Setting}

We denote by \( \mathscr{R}_+([0,\infty)) \) the set of all non-negative Radon measures \( f(\omega) \) on \( [0,\infty) \) such that
\begin{equation}\label{Radon}
	\|f\|_{\mathscr{R}_+} \ := \ \int_{[0,\infty)} \mathrm{d}\omega\, f(\omega) \ = \int_{[0,\infty)} f(\mathrm{d}\omega)  < \ \infty.
\end{equation}

We have the definition. 
\begin{definition}
	Under Assumptions A and B, we say that \( f(t,k) = f(t,\omega) \) is a \emph{global mild radial solution} of \eqref{4wave} with radial initial condition \( f_0(k) = f_0(|k|) \ge 0 \) if \( f(t,k) \ge 0 \), \( f(t,\omega)\,\mathfrak{A}(\omega) \in C^1([0,\infty), \mathscr{R}_+([0,\infty))) \), and for all \( \phi \in C^2([0,\infty)) \) such that the set \( \{ \omega \mid \phi(\omega) \ne 0 \} \) is a compact subset of \( [0,\infty) \) and $\phi'(0)=0$, we have
	\begin{equation}\label{4wavemild}
		\begin{aligned}
			\int_{[0,\infty)} \mathrm{d}\omega\, f(t,\omega)\, \mathfrak{A}(\omega)\, \phi(\omega) 
			= \ &\int_{[0,\infty)} \mathrm{d}\omega\, f(0,\omega)\, \mathfrak{A}(\omega)\, \phi(\omega) \\
			&+ \int_0^t\mathrm d s \int_{[0,\infty)} \mathrm{d}\omega\, \mathscr{Q}[f]\, \phi(\omega)\, \mathfrak{A}(\omega),
		\end{aligned}
	\end{equation}
	for all \( t \in \mathbb{R}_+ \).
\end{definition}

We have the main theorem, whose proof is given in Section \ref{Sec:Proof}.

\begin{theorem}
	\label{Theorem1} 
	We assume Assumption A and Assumption B. 
	
	Let $f_0(k) = f_0(|k|) \ge 0$ be an initial condition satisfying
	\begin{equation}
		\label{MassEnergy} 
		\int_{\mathbb{R}^3} \mathrm{d}k\, f_0(k) = \mathscr{M}, \quad 
		\int_{\mathbb{R}^3} \mathrm{d}k\, f_0(k)\, \omega(k) = \mathscr{E}, 
	\end{equation}
	for some constants $\mathscr{M},\mathscr{E}>0$.
	
	Then there exists at least one global mild radial solution $f(t,k)$ of \eqref{4wave} in the sense of \eqref{4wavemild} such that
	\begin{equation}
		\label{Theorem1:2} 
		\int_{\mathbb{R}^3} \mathrm{d}k\, f(t,k) \le \mathscr{M},
	\end{equation}
	for all $t \ge 0$. 
We define
\begin{equation}
	\label{FDefinition} 
	G(t,k) = f(t,k)\,|k|\,\Theta(k),
\end{equation}
for \( t \ge 0 \), and assume that
\begin{equation}
	\int_{\{\omega = 0\}} \mathrm{d}\omega\, G(0,\omega) = 0.
\end{equation}

	Moreover,  suppose further that there exist constants $C_{\mathrm{ini}} > 0$, $\theta> c_{\mathrm{ini}} \ge 0$, and $r_0 > 0$ such that
	
		\begin{equation}
		\label{Theorem1:4}
		\int_{0}^{r} \mathrm{d}\omega\, G(0,\omega) \ge C_{\mathrm{ini}}\, r^{c_{\mathrm{ini}}},
	\end{equation}
	for all $0 < r < r_0$.

	The following then hold:
	
	\begin{itemize}
		\item[(i)] (Immediate Condensation-Condensate Growth) When
		\begin{equation}
			\label{Theorem1:3}
			0 \le c_{\mathrm{ini}} < \min\left\{ \frac{2\delta - \frac{1}{2} - \varrho }{10(2 + \mu + \varrho)}\, , \theta \, , {2\delta - \varrho},\, \tfrac{2}{3} \delta,\, \tfrac{\frac{2\delta - \frac{1}{2} - \varrho }{5(2 + \mu + \varrho)} + \varrho}{2} \right\},
		\end{equation}
	 the quantity $T_0$ defined in \eqref{T0} is precisely zero. This confirms the onset of condensate growth.
		
		\item[(ii)] (Finite-time Condensation) If we assume the weaker assumption: \( 	  \theta> c_{\mathrm{ini}} \ge 0 \), then the quantity \( T_0 \) defined in \eqref{T0} is finite.
		
	\end{itemize}
\end{theorem}

We conclude this section by introducing notations to be used throughout the paper.

We define the non-condensation time set
\begin{equation}
	\label{NoCondensateTime}
	\Xi := \left\{ t \in [0,\infty) \,\middle|\, \int_{\{0\}} \mathrm{d}k\, f(t,k) = 0 \right\}.
\end{equation}

For $x, y, z \in \mathbb{R}$, we define
\begin{equation}
	\label{Mid}
	\mathrm{mid}\{x, y, z\} := t \in \{x, y, z\} \setminus \{\max\{x, y, z\}, \min\{x, y, z\}\}.
\end{equation}

For $\omega, \omega_1, \omega_2$, we define
\begin{equation}
	\label{Sec:DDM:6}
	\begin{aligned}
		\omega_{\mathrm{Max}}(\omega, \omega_1, \omega_2) &= \max\{\omega, \omega_1, \omega_2\}, \\
		\omega_{\mathrm{Min}}(\omega, \omega_1, \omega_2) &= \min\{\omega, \omega_1, \omega_2\}, \\
		\omega_{\mathrm{Mid}}(\omega, \omega_1, \omega_2) &= \mathrm{mid}\{\omega, \omega_1, \omega_2\}.
	\end{aligned}
\end{equation}

\begin{remark}\label{RemarkMainTheo}
	Part~(i) of the above theorem implies that, starting from a regular initial condition sufficiently concentrated near the origin, the density function of the thermal cloud immediately develops a Dirac mass at the origin. This Dirac mass signifies that additional particles are rapidly transferred into the existing BEC, causing the condensate to begin growing instantaneously. This behavior explains the initial growth observed at \( t = 0 \) in the curves shown in Figure~\ref{fig1}.
	
	Part~(ii) implies that, starting from a regular initial condition that is not sufficiently concentrated near the origin, the density function still develops a Dirac mass at the origin in finite time.

We note that, by Lemma \ref{lemma:C22}, the kernel \( \mathcal{W}_{22} \) can be replaced by  
\begin{equation*}
	\label{Remark:E1}
	\left[ \frac{\max\{|\omega - \omega_2|,\  |\omega_1 - \omega_2|\}}{2(\omega + \omega_1)} \right]^\mu.
\end{equation*}
Using the notations in \eqref{Sec:DDM:6}, we suppose \( \omega = \omega_1 = \omega_{\mathrm{Max}} = \omega_{\mathrm{Mid}} \) and \( \omega_2 = \omega_{\mathrm{Min}} \). Then,
\[
\max\{|\omega - \omega_2|,\  |\omega_1 - \omega_2|\}^\mu = |\omega_{\mathrm{Mid}} - \omega_{\mathrm{Min}}|^\mu.
\]
This implies that in the estimate \eqref{Lemma:Concave:1}, the degeneracy is of order \( |\omega_{\mathrm{Mid}} - \omega_{\mathrm{Min}}|^{\mu+2} \). However, no constraints are imposed on \( \mu \) other than \( \mu \ge 0 \).
The main reason is that in our estimates, the quantity \( \mathcal{W}_{22} \) can be bounded from below through careful choices of \( \omega, \omega_1, \omega_2 \) (see, for instance, \eqref{muestimate}). This demonstrates the flexibility of our framework in handling various challenging collisional kernels.

Moreover, the theorem above and its proof remain valid if we set \( \mathfrak{c}_{31} = \mathfrak{c}_{21} = 0 \).

\end{remark}

\begin{remark}\label{RemarkMainTheo:a}
A standard symmetry argument implies that  
\begin{equation}
	\label{RemarkMainTheo:a1}
	\partial_t \int_{\mathbb{R}^3} \mathrm{d}k\, f(t,k)\omega(k)
	= \int_{\mathbb{R}^3} \mathrm{d}k\, \mathscr{Q}[f](t,k)\omega(k)
	= 0.
\end{equation}
This formally expresses the conservation of energy.  
However, when the solution is taken in the space $\mathscr{R}_+([0,\infty))$ as defined in~\eqref{Radon},  
the integrals in $\int_{\mathbb{R}^3} \mathrm{d}k\, \mathscr{Q}[f](t,k)\omega(k)$ are not well-defined.  
Consequently, part of the energy escapes to infinity.  
This phenomenon, finite-time energy loss or cascade, will be the focus of the companion paper~\cite{staffilani2025energyfinite}.

\end{remark}
\section{Weak formulations}\label{Sec:WeakFormulation}
In this section, we introduce the weak formulations for the $3$ collision operators. The proofs of these formulas are standard (see~\cite{staffilani2024condensation, staffilani2024energy}); however, we include them here for the sake of completeness.

\begin{lemma}
	\label{lemma:C12} 	We assume Assumption A and Assumption B. 

	For any suitable test function $\phi(\omega)$, we have the following equation:
	\begin{equation}\label{Lemma:C12:1}
		\begin{aligned}
			& \int_{\mathbb{R}_+} \mathrm{d}\omega\, C_{12}[f](\omega)\, \phi(\omega)\, \mathfrak A(\omega) \\
			=\, &  c_{12} \iiint_{\mathbb{R}_+^3}  \mathrm{d}\omega_1 \mathrm{d}\omega_2 \mathrm{d}\omega \delta(\omega - \omega_1 - \omega_2)  \Theta\Theta_1\Theta_2 |k||k_1||k_2| \, [f_1 f_2 - f(f_1 + f_2)] \, [\phi(\omega) - \phi(\omega_1) - \phi(\omega_2)]
		\end{aligned}
	\end{equation}
	where $c_{12}$ is a constant independent of $f$ and $\phi$ and we have used the notations introduced in \eqref{Shorthand}.
	
\end{lemma} 

\begin{proof}
	Following a standard argument (see \cite{soffer2020energy}) by performing the change of variables $k \leftrightarrow k_1$ and $k \leftrightarrow k_2$, we find
	\begin{align*}
		\int_{\mathbb{R}^3} \mathrm{d}k\, C_{12}[f](k)\, \phi(\omega)
		=\, & \mathfrak{c}_{12} \iiint_{\mathbb{R}^3 \times \mathbb{R}^3 \times \mathbb{R}^3} \mathrm{d}k\, \mathrm{d}k_1\, \mathrm{d}k_2\, |k||k_1||k_2|\, \delta(k - k_1 - k_2)\, \delta(\omega - \omega_1 - \omega_2) \\
		& \times \Big[ f(k_1)f(k_2) - f(k_1)f(k) - f(k_2)f(k) \Big] \Big[ \phi(\omega) - \phi(\omega_1) - \phi(\omega_2) \Big].
	\end{align*}
	
Following \cite{staffilani2024condensation, staffilani2024energy}, we compute
\begin{equation}\label{Lemma:C12:E1}
	\begin{aligned}
	&	\int_{\mathbb{R}^3} \mathrm{d}k\, C_{12}[f](k)\, \phi(k)\\
		=\, & \mathfrak{c}_{12} \iiint_{\mathbb{R}_+^3} \mathrm{d}|k|\, \mathrm{d}|k_1|\, \mathrm{d}|k_2|\, |k|^3 |k_1|^3 |k_2|^3\, \delta(\omega - \omega_1 - \omega_2) \left[ f_1 f_2 - f_1 f - f_2 f \right] \\
		& \times \left[ \phi(\omega) - \phi(\omega_1) - \phi(\omega_2) \right] \iiint_{(\mathbb{S}^2)^3} \mathrm{d}\mathcal{V}_1\, \mathrm{d}\mathcal{V}_2\, \mathrm{d}\mathcal{V}\, \left[ \frac{1}{(2\pi)^3} \int_{\mathbb{R}^3} \mathrm{d}s\, e^{is \cdot (k - k_1 - k_2)} \right] \\
		=\, & \mathfrak{c}_{12} \iiint_{\mathbb{R}_+^3} \mathrm{d}|k|\, \mathrm{d}|k_1|\, \mathrm{d}|k_2|\, |k|^2 |k_1|^2 |k_2|^2\, \delta(\omega - \omega_1 - \omega_2) \left[ f_1 f_2 - f_1 f - f_2 f \right] \\
		& \times \left[ \phi(\omega) - \phi(\omega_1) - \phi(\omega_2) \right] \cdot 32\pi \int_0^\infty \mathrm{d}s\, \frac{\sin(|k_1| s)\, \sin(|k_2| s)\, \sin(|k| s)}{s}.
	\end{aligned}
\end{equation}

Since \( |k| < |k_1| + |k_2| \), we deduce
\begin{equation*}
	\int_0^\infty \mathrm{d}s\, \frac{\sin(|k_1| s)\, \sin(|k_2| s)\, \sin(|k| s)}{s} = \tfrac{\pi}{4},
\end{equation*}
which, when combined with \eqref{Lemma:C12:E1}, yields
\begin{equation}\label{Lemma:C12:E2}
	\begin{aligned}
	\int_{\mathbb{R}_+} \mathrm{d}\omega\, C_{12}[f](\omega)\, \phi(\omega)\, \mathfrak A(\omega)
		=\, & c_{12} \iiint_{\mathbb{R}_+^3} \mathrm{d}|k|\, \mathrm{d}|k_1|\, \mathrm{d}|k_2|\, |k|^2 |k_1|^2 |k_2|^2\, \delta(\omega - \omega_1 - \omega_2) \\
		& \times \left[ f_1 f_2 - f_1 f - f_2 f \right] \left[ \phi(\omega) - \phi(\omega_1) - \phi(\omega_2) \right],
	\end{aligned}
\end{equation}
for a universal constant \( c_{12} > 0 \). By changing the variables \( |k| \mapsto \omega \), \( |k_1| \mapsto \omega_1 \), and \( |k_2| \mapsto \omega_2 \), we deduce from \eqref{Lemma:C12:E2}
\begin{equation*} 
	\begin{aligned}
		\int_{\mathbb{R}_+} \mathrm{d}\omega\, C_{12}[f](\omega)\, \phi(\omega)\, \mathfrak A(\omega)
		=\, & c_{12} \iiint_{\mathbb{R}_+^3} \mathrm{d}\omega_1\, \mathrm{d}\omega_2\, \mathrm{d}\omega\, \delta(\omega - \omega_1 - \omega_2)\,\Theta\,\Theta_1\,\Theta_2\, |k|\, |k_1|\, |k_2| \\
		& \times \left[ f_1 f_2 - f(f_1 + f_2) \right] \left[ \phi(\omega) - \phi(\omega_1) - \phi(\omega_2) \right],
	\end{aligned}
\end{equation*}
which concludes our proof of the lemma.

\end{proof}

\begin{lemma}
	\label{lemma:C22} 	We assume Assumption A and Assumption B. 

For any suitable test function $\phi(\omega)$, we have the following equation:
\begin{equation}\label{Lemma:C22:1}
	\begin{aligned}
		& \int_{\mathbb{R}_+} \mathrm{d}\omega\, C_{22}[f](\omega)\, \phi(\omega)\, \mathfrak A(\omega) \\
		=\, & c_{22} \iiiint_{\mathbb{R}_+^4} \mathrm{d}\omega_1\, \mathrm{d}\omega_2\, \mathrm{d}\omega_3\, \mathrm{d}\omega\, \delta(\omega + \omega_1 - \omega_2 - \omega_3)  \Theta\,\Theta_1\,\Theta_2\,\Theta_3\, \min\{|k_1|, |k_2|, |k_3|, |k|\} 
		  \\
		& \times \left[ \frac{\max\{|\omega - \omega_2|,\  |\omega_1 - \omega_2|\}}{2(\omega + \omega_1)} \right]^\mu f_1 f_2 f \left[ -\phi(\omega) - \phi(\omega_1) + \phi(\omega_2) + \phi(\omega + \omega_1 - \omega_2) \right],
	\end{aligned}
\end{equation}
where $c_{22}$ is a constant independent of $f$ and $\phi$ and we have used the notations introduced in \eqref{Shorthand}.

	\end{lemma} 

\begin{proof}
Following \cite{staffilani2024condensation,staffilani2024energy}, we compute
\begin{equation}\label{Lemma:C22:E1}
	\begin{aligned}
		C_{22} \left[ f \right]  
		=\, & \mathfrak{c}_{22} \iiint_{\mathbb{R}^{3 \times 3}} \mathrm{d}k_1\, \mathrm{d}k_2\, \mathrm{d}k_3\, \delta(\omega + \omega_1 - \omega_2 - \omega_3)\, \delta(k + k_1 - k_2 - k_3) \\
		&\times [-f f_1(f_2 + f_3) + f_2 f_3(f_1 + f)] \left[ \frac{\max\{|\omega - \omega_2|,\  |\omega_1 - \omega_2|,\ |\omega - \omega_3|,\  |\omega_1 - \omega_3|\}}{\omega + \omega_1 + \omega_2 + \omega_3} \right]^\mu \\
		=\, & \mathfrak{c}_{22} \iiint_{\mathbb{R}_+^{3}} \mathrm{d}|k_1|\, \mathrm{d}|k_2|\, \mathrm{d}|k_3|\, |k_1|^2 |k_2|^2 |k_3|^2\, \delta(\omega + \omega_1 - \omega_2 - \omega_3) \\
		&\times \iiint_{(\mathbb{S}^2)^3} \mathrm{d}\mathcal{V}_1\, \mathrm{d}\mathcal{V}_2\, \mathrm{d}\mathcal{V}_3\, [-f f_1(f_2 + f_3) + f_2 f_3(f_1 + f)] \left[ \frac{1}{(2\pi)^3} \int_{\mathbb{R}^3} \mathrm{d}s\, e^{is \cdot (k + k_1 - k_2 - k_3)} \right] \\
		&\times \left[ \frac{\max\{|\omega - \omega_2|,\  |\omega_1 - \omega_2|,\ |\omega - \omega_3|,\  |\omega_1 - \omega_3|\}}{\omega + \omega_1 + \omega_2 + \omega_3} \right]^\mu \\
		=\, & \mathfrak{c}_{22} \iiint_{\mathbb{R}_+^{3}} \mathrm{d}|k_1|\, \mathrm{d}|k_2|\, \mathrm{d}|k_3|\, \frac{32\pi}{|k|} \int_0^\infty \mathrm{d}s\, \frac{\sin(|k_1|s)\sin(|k_2|s)\sin(|k_3|s)\sin(|k|s)}{s^2} \\
		&\times \delta(\omega + \omega_1 - \omega_2 - \omega_3)\, [-f f_1(f_2 + f_3) + f_2 f_3(f_1 + f)]\, |k_1||k_2||k_3|\\
		&\times \left[ \frac{\max\{|\omega - \omega_2|,\  |\omega_1 - \omega_2|,\ |\omega - \omega_3|,\  |\omega_1 - \omega_3|\}}{\omega + \omega_1 + \omega_2 + \omega_3} \right]^\mu .
	\end{aligned}
\end{equation}

Since \( \omega + \omega_1 = \omega_2 + \omega_3 \), we compute
\begin{equation*}
	\begin{aligned}
	&	\int_0^\infty \mathrm{d}s\, \frac{\sin(|k_1|s)\, \sin(|k_2|s)\, \sin(|k_3|s)\, \sin(|k|s)}{s^2}\\
		=  & \frac{\pi}{16} \sum_{x_1,x_2,x_3,x_4=1}^2 (-1)^{x_1 + x_2 + x_3 + x_4 + 1} \left| (-1)^{x_1} |k_1| + (-1)^{x_2} |k_2| + (-1)^{x_3} |k_3| + (-1)^{x_4} |k| \right| \\
		= &\tfrac{\pi}{4} \min\{ |k_1|, |k_2|, |k_3|, |k| \}.
	\end{aligned}
\end{equation*}

Combining this with \eqref{Lemma:C22:E1} yields:

\begin{equation}\label{Lemma:C22:E2}
	\begin{aligned}
	&	\int_{\mathbb{R}_+} \mathrm{d}\omega\, C_{22}[f](\omega)\, \phi(\omega)\, |k|\,\Theta
		=\,  c_{22} \iiiint_{\mathbb{R}_+^{4}} \mathrm{d}|k|\,\mathrm{d}|k_1|\, \mathrm{d}|k_2|\, \mathrm{d}|k_3|\, {|k||k_1||k_2||k_3|\, \min\{|k_1|, |k_2|, |k_3|, |k|\}} \\
		&\times \delta(\omega + \omega_1 - \omega_2 - \omega_3) \left[ \frac{\max\{|\omega - \omega_2|,\  |\omega_1 - \omega_2|\}}{2(\omega + \omega_1)} \right]^\mu [-f f_1(f_2 + f_3) + f_2 f_3(f_1 + f)]\phi,
	\end{aligned}
\end{equation}
for a universal constant \( c_{22} > 0 \).

After the change of variables \( |k| \mapsto \omega \), \( |k_1| \mapsto \omega_1 \), \( |k_2| \mapsto \omega_2 \), \( |k_3| \mapsto \omega_3 \), we obtain the following weak form of \eqref{Lemma:Concave:E3}:
\begin{equation}\label{Lemma:C22:E3}
	\begin{aligned}
	&	\int_{\mathbb{R}_+} \mathrm{d}\omega\, C_{22}[f](\omega)\, \phi(\omega)\, |k|\,\Theta 
		=\,  c_{22} \iiiint_{\mathbb{R}_+^4} \mathrm{d}\omega_1\, \mathrm{d}\omega_2\, \mathrm{d}\omega_3\, \mathrm{d}\omega\, \delta(\omega + \omega_1 - \omega_2 - \omega_3) \\
		&\times\Theta\,\Theta_1\,\Theta_2\,\Theta_3\, \min\{|k_1|, |k_2|, |k_3|, |k|\}   \left[ \frac{\max\{|\omega_1 - \omega_2|, |\omega_2 - \omega|\}}{2(\omega + \omega_1)} \right]^\mu \\
		&\times [-f f_1(f_2 + f_3) + f_2 f_3(f_1 + f)]\, \phi.
	\end{aligned}
\end{equation}
  Following a standard argument (see \cite{soffer2018dynamics}) by performing the change of  variables $k\leftrightarrow k_1$, $k\leftrightarrow k_2$ and $k\leftrightarrow k_3$, we find  \eqref{Lemma:C22:1}.

\end{proof}

\begin{lemma}
	\label{lemma:C31} 	We assume Assumption A and Assumption B. 

	For any suitable test function $\phi(\omega)$, we have the following equation:
	\begin{equation}\label{Lemma:C31:1}
		\begin{aligned}
			& \int_{\mathbb{R}_+} \mathrm{d}\omega\, C_{31}[f](\omega)\, \phi(\omega)\, \mathfrak A(\omega) \\
			=\, &  c_{31} \iiiint_{\mathbb{R}_+^4}\mathrm{d}\omega_1 \mathrm{d}\omega_2 \mathrm{d}\omega_3 \mathrm{d}\omega \delta(\omega - \omega_1 - \omega_2 - \omega_3) \\
			& \quad \times\Theta\Theta_1\Theta_2\Theta_3 |k||k_1||k_2||k_3| \, [f_1 f_2 f_3 - f(f_1 f_2 + f_2 f_3 + f_1 f_3)] \, [\phi(\omega) - \phi(\omega_1) - \phi(\omega_2) - \phi(\omega_3)] \, 
		\end{aligned}
	\end{equation}
	where $c_{31}$ is a constant independent of $f$ and $\phi$ and we have used the notations introduced in \eqref{Shorthand}.
	
\end{lemma} 

\begin{proof}
Similar to Lemma~\ref{lemma:C12}, by performing the changes of variables \( k \leftrightarrow k_1 \), \( k \leftrightarrow k_2 \) and \( k \leftrightarrow k_3 \), we find
\begin{align*}
	\int_{\mathbb{R}^3} \mathrm{d}k\, C_{31}[f](k)\, \phi(k)
	&= \mathfrak{c}_{31} \iiiint_{\mathbb{R}^3 \times \mathbb{R}^3 \times \mathbb{R}^3 \times \mathbb{R}^3} \mathrm{d}k\, \mathrm{d}k_1\, \mathrm{d}k_2\, \mathrm{d}k_3\, \delta(k - k_1 - k_2 - k_3) \\
	&\quad \times |k| |k_1| |k_2|\,|k_3| \left||k_1| + |k_2| + |k_3| - |k|\right|^{-1} \left[ f_1 f_2 f_3 - f(f_1 f_2 + f_2 f_3 + f_3 f_1) \right] \\
	&\quad \times \delta(\omega - \omega_1 - \omega_2 - \omega_3) \left[ \phi(\omega) - \phi(\omega_1) - \phi(\omega_2) - \phi(\omega_3) \right].
\end{align*}

In analogy with \eqref{Lemma:C12:E1}, we compute
\begin{equation}\label{Lemma:C31:E1}
	\begin{aligned}
	&	\int_{\mathbb{R}^3} \mathrm{d}k\, C_{31}[f](k)\, \phi(k)
		= \mathfrak{c}_{31} \iiiint_{\mathbb{R}_+^4} \mathrm{d}|k|\, \mathrm{d}|k_1|\, \mathrm{d}|k_2|\, \mathrm{d}|k_3|\, |k|^3 |k_1|^3 |k_2|^3 |k_3|^3 \\
		&\quad \times \delta(k - k_1 - k_2 - k_3)\, \delta(\omega - \omega_1 - \omega_2 - \omega_3)\left||k_1| + |k_2| + |k_3| - |k|\right|^{-1} \\
		&\quad \times  \left[ f_1 f_2 f_3 - f(f_1 f_2 + f_2 f_3 + f_3 f_1) \right]  \left[ \phi(\omega) - \phi(\omega_1) - \phi(\omega_2) - \phi(\omega_3) \right] \\
		&\quad \times \iiint_{(\mathbb{S}^2)^3} \mathrm{d}\mathcal{V}_1\, \mathrm{d}\mathcal{V}_2\, \mathrm{d}\mathcal{V}_3\, \left[ \frac{1}{(2\pi)^3} \int_{\mathbb{R}^3} \mathrm{d}s\, e^{is \cdot (k - k_1 - k_2 - k_3)} \right] \\
		&= \mathfrak{c}_{31}' \iiiint_{\mathbb{R}_+^4} \mathrm{d}|k|\, \mathrm{d}|k_1|\, \mathrm{d}|k_2|\, \mathrm{d}|k_3|\, |k|^2 |k_1|^2 |k_2|^2 |k_3|^2 \left||k_1| + |k_2| + |k_3| - |k|\right|^{-1}\\
		&\quad \times \delta(\omega - \omega_1 - \omega_2 - \omega_3) \left[ f_1 f_2 f_3 - f(f_1 f_2 + f_2 f_3 + f_3 f_1) \right] \\
		&\quad \times \left[ \phi(\omega) - \phi(\omega_1) - \phi(\omega_2) - \phi(\omega_3) \right]   \int_0^\infty \mathrm{d}s\, \frac{\sin(|k_1|s)\, \sin(|k_2|s)\, \sin(|k_3|s)\, \sin(|k|s)}{s^2},
	\end{aligned}
\end{equation}
where \( \mathfrak{c}_{31}' > 0 \) is a positive constant.

Since
\[
\omega = \omega_1 + \omega_2 + \omega_3,
\]
by \eqref{CondDisper1}, we deduce that
\[
|k| < |k_1| + |k_2| + |k_3|.
\]
 We now compute
\begin{equation*}
	\begin{aligned}
		\int_0^\infty \mathrm{d}s\, \frac{\sin(|k_1|s)\, \sin(|k_2|s)\, \sin(|k_3|s)\, \sin(|k|s)}{s^2}
		&= \frac{\pi}{16} \sum_{x_1,x_2,x_3,x_4=1}^2 (-1)^{x_1 + x_2 + x_3 + x_4 + 1} \\
		&\quad \times \left| (-1)^{x_1} |k_1| + (-1)^{x_2} |k_2| + (-1)^{x_3} |k_3| + (-1)^{x_4} |k| \right| \\
		&= \tfrac{\pi}{8} \left[|k_1| + |k_2| + |k_3| - |k|\right].
	\end{aligned}
\end{equation*}

Combining this with \eqref{Lemma:C31:E1}, we obtain
\begin{equation}\label{Lemma:C31:E2}
	\begin{aligned}
		\int_{\mathbb{R}_+} \mathrm{d}\omega\, C_{31}[f](\omega)\, \phi(\omega)\, \mathfrak{A}(\omega)
		&= c_{31} \iiiint_{\mathbb{R}_+^4} \mathrm{d}|k|\, \mathrm{d}|k_1|\, \mathrm{d}|k_2|\, \mathrm{d}|k_3|\, |k|^2 |k_1|^2 |k_2|^2 |k_3|^2 \\
		&\quad \times \delta(\omega - \omega_1 - \omega_2 - \omega_3) \left[ f_1 f_2 f_3 - f(f_1 f_2 + f_2 f_3 + f_3 f_1) \right] \\
		&\quad \times \left[ \phi(\omega) - \phi(\omega_1) - \phi(\omega_2) - \phi(\omega_3) \right],
	\end{aligned}
\end{equation}
for a universal constant \( c_{31} > 0 \). By changing variables \( |k| \to \omega \), \( |k_1| \to \omega_1 \), and \( |k_2| \to \omega_2 \), we deduce \eqref{Lemma:C31:1} from \eqref{Lemma:C31:E2}.

\end{proof}

\section{A priori estimates using the initial mass and energy}\label{Sec:Estimates}
In this section, we provide some \emph{a priori} estimates for the solutions.

\begin{lemma}
	\label{lemma:Concave}  	We assume Assumption A and Assumption B. 

	Let $f$ be a radial solution in the sense of \eqref{4wavemild} of the wave kinetic equation \eqref{4wave}. Then, for all $T > 0$ and for all $\alpha \in [1/2,1)$, the following estimate holds:
	\begin{equation}\label{Lemma:Concave:1}
		\begin{aligned}
			\mathscr{M} + \mathscr{E} \ \ge\  
			& \mathcal C_1 \int_0^T \mathrm{d}t \iiint_{\mathbb{R}_+^3} \mathrm{d}\omega_1\, \mathrm{d}\omega_2\, \mathrm{d}\omega\,
			f_1 f_2 f\, \frac{(\omega_{\mathrm{Mid}} - \omega_{\mathrm{Min}})^2 \alpha(1 - \alpha)}{(2\omega_{\mathrm{Mid}} - \omega_{\mathrm{Min}})^{2 - \alpha}}\, |k_{\mathrm{Min}}| \chi_{(0,1)}(2\omega_{\mathrm{Mid}} - \omega_{\mathrm{Min}})\\
			& \quad \times\Theta(\omega_{\mathrm{Max}})\,\Theta(\omega_{\mathrm{Min}})\,\Theta(\omega_{\mathrm{Mid}})\,\Theta(\omega_{\mathrm{Max}} - \omega_{\mathrm{Min}} + \omega_{\mathrm{Mid}})  \left[ \frac{\max\{|\omega_1 - \omega_2|,\ |\omega_2 - \omega|\}}{2(\omega + \omega_1)} \right]^\mu ,
		\end{aligned}
	\end{equation}
	for some constant $\mathcal C_1>0$ independent of $f$.
	Here, $k_{\text{Min}}$ is associated to $\omega_{\text{Min}}$. The other notations are consistent with those introduced in \eqref{Sec:DDM:6}.
	
\end{lemma}

\begin{proof} 
Let \( \phi(\omega) = \omega^\alpha \chi_{[0,1]}(\omega) + \chi_{(1,\infty)}(\omega) \), with \( 0 < \alpha < 1 \).
 Then, by Lemma \ref{lemma:C12}, Lemma \ref{lemma:C22} and Lemma \ref{lemma:C31} 
	
	\begin{equation}\label{Lemma:Concave:E1}
		\begin{aligned}
			& \int_{\mathbb{R}_+} \mathrm{d}\omega  \partial_t f(\omega) \, \phi(\omega) |k|\Theta \, \\
			=\ &  \int_{\mathbb{R}_+} \mathrm{d}\omega C_{12}[f](\omega) \, \phi(\omega) |k|\Theta \, +  \int_{\mathbb{R}_+} \mathrm{d}\omega C_{22}[f](\omega) \, \phi(\omega) |k|\Theta \, 
			+ \int_{\mathbb{R}_+}  \mathrm{d}\omega C_{31}[f](\omega) \, \phi(\omega) |k|\Theta \, \\
			=\ & c_{12} \iiint_{\mathbb{R}_+^3}  \mathrm{d}\omega_1 \mathrm{d}\omega_2 \mathrm{d}\omega \delta(\omega - \omega_1 - \omega_2)  \Theta\Theta_1\Theta_2 |k||k_1||k_2| \, [f_1 f_2 - f(f_1 + f_2)] \, [\phi - \phi_1 - \phi_2] \,\\ 
			& + c_{22} \iiiint_{\mathbb{R}_+^4} \mathrm{d}\omega_1 \mathrm{d}\omega_2 \mathrm{d}\omega_3 \mathrm{d}\omega\delta(\omega + \omega_1 - \omega_2 - \omega_3)		\left[ \frac{\max\{|\omega_1 - \omega_2|,\ |\omega_2 - \omega|\}}{2(\omega + \omega_1)} \right]^\mu  \\
			& \quad \times\Theta\Theta_1\Theta_2\Theta_3 \min\{|k|, |k_1|, |k_2|, |k_3|\} \, f_1 f_2 f \, [-\phi - \phi_1 + \phi_2 + \phi_3] \,  \\
			& + c_{31} \iiiint_{\mathbb{R}_+^4}\mathrm{d}\omega_1 \mathrm{d}\omega_2 \mathrm{d}\omega_3 \mathrm{d}\omega \delta(\omega - \omega_1 - \omega_2 - \omega_3) \\
			& \quad \times\Theta\Theta_1\Theta_2\Theta_3 |k||k_1||k_2||k_3| \, [f_1 f_2 f_3 - f(f_1 f_2 + f_2 f_3 + f_1 f_3)] \, [\phi - \phi_1 - \phi_2 - \phi_3] \, .
		\end{aligned}
	\end{equation}
	We now divide the remainder of the proof into several steps to clarify the structure and facilitate the analysis.

	{\it Step 1:} 
	Recalling that $\mathfrak A(\omega)=\Theta|k|$, we rewrite the first term on the right-hand side of \eqref{Lemma:Concave:E1} as
	
	\begin{equation*}\label{Lemma:Concave:E6}
		\begin{aligned}
			& \int_{\mathbb{R}_+} \mathrm{d}\omega\, C_{12}[f]\, \phi(\omega)\, \mathfrak A(\omega) \\
			=\, & c_{12} \iint_{\mathbb{R}_+^2} \mathrm{d}\omega_1\, \mathrm{d}\omega_2\, \mathfrak{A}(\omega_1 + \omega_2)\, \mathfrak{A}(\omega_1)\, \mathfrak{A}(\omega_2) \\
			& \quad \times \left[ f(\omega_1) f(\omega_2) - f(\omega_1 + \omega_2)\big(f(\omega_1) + f(\omega_2)\big) \right]  \left[ \phi(\omega_1 + \omega_2) - \phi(\omega_1) - \phi(\omega_2) \right],
		\end{aligned}
	\end{equation*}
	which, after rearranging the terms, can be re-expressed as
	\begin{equation}\label{Lemma:Concave:E7a}
		\begin{aligned}
			& \int_{\mathbb{R}_+} \mathrm{d}\omega\, C_{12}[f]\, \phi(\omega)\, \mathfrak A(\omega) \\
			=\, & c_{12} \iint_{\mathbb{R}_+^2} \mathrm{d}\omega_1\, \mathrm{d}\omega_2\, \mathfrak{A}(\omega_1)\, \mathfrak{A}(\omega_2)\, f(\omega_1) f(\omega_2)\, \mathfrak{A}(\omega_1 + \omega_2) \left[ \phi(\omega_1 + \omega_2) - \phi(\omega_1) - \phi(\omega_2) \right] \\
			& - 2c_{12} \iint_{\mathbb{R}_+^2} \mathrm{d}\omega_1\, \mathrm{d}\omega_2\, \mathfrak{A}(\omega_1 + \omega_2)\, \mathfrak{A}(\omega_1)\, \mathfrak{A}(\omega_2)\, f(\omega_1 + \omega_2) f(\omega_1) \left[ \phi(\omega_1 + \omega_2) - \phi(\omega_1) - \phi(\omega_2) \right].
		\end{aligned}
	\end{equation}
	
	Next, by performing the change of variables $(\omega_1 + \omega_2,\omega_1) \to (\omega_1,\omega_2)$, we find
	\begin{equation*}\label{Lemma:Concave:E8}
		\begin{aligned}
			& \int_{\mathbb{R}_+} \mathrm{d}\omega\, C_{12}[f]\, \phi(\omega)\, \mathfrak{A}(\omega) \\
			=\, & c_{12} \iint_{\mathbb{R}_+^2} \mathrm{d}\omega_1\, \mathrm{d}\omega_2\, \mathfrak{A}(\omega_1)\, \mathfrak{A}(\omega_2)\, f(\omega_1) f(\omega_2)\, \mathfrak{A}(\omega_1 + \omega_2) \left[ \phi(\omega_1 + \omega_2) - \phi(\omega_1) - \phi(\omega_2) \right] \\
			& - 2c_{12} \iint_{\omega_1 \ge \omega_2} \mathrm{d}\omega_1\, \mathrm{d}\omega_2\, \mathfrak{A}(\omega_1)\, \mathfrak{A}(\omega_1 - \omega_2)\, \mathfrak{A}(\omega_2)\, f(\omega_1) f(\omega_2) \left[ \phi(\omega_1) - \phi(\omega_1 - \omega_2) - \phi(\omega_2) \right],
		\end{aligned}
	\end{equation*}
	which can again be rewritten as
	\begin{equation}\label{Lemma:Concave:E8a}
		\begin{aligned}
			& \int_{\mathbb{R}_+} \mathrm{d}\omega\, C_{12}[f]\, \phi(\omega)\, \mathfrak{A}(\omega) \\
			=\, & 2c_{12} \iint_{\omega_1 > \omega_2} \mathrm{d}\omega_1\, \mathrm{d}\omega_2\, \mathfrak{A}(\omega_1)\, \mathfrak{A}(\omega_2)\, f(\omega_1) f(\omega_2) \\
			& \quad \times \left[ \mathfrak{A}(\omega_1 + \omega_2) \left( \phi(\omega_1 + \omega_2) - \phi(\omega_1) - \phi(\omega_2) \right)   -\ \mathfrak{A}(\omega_1 - \omega_2) \left( \phi(\omega_1) - \phi(\omega_1 - \omega_2) - \phi(\omega_2) \right) \right] \\
			& + c_{12} \iint_{\omega_1 = \omega_2} \mathrm{d}\omega_1\, \mathrm{d}\omega_2\, \mathfrak{A}(\omega_1)\, \mathfrak{A}(\omega_2)\, f(\omega_1) f(\omega_2)\, \mathfrak{A}(2\omega_1) \left[ \phi(2\omega_1) - 2\phi(\omega_1) \right],
		\end{aligned}
	\end{equation}
	with the observation that $\mathfrak{A}(0) = 0$. 
	
	Next, we compute, using the fact that $\phi(0)=0$
	\[
	\mathfrak{A}(\omega_1 + \omega_2) \left( \phi(\omega_1 + \omega_2) - \phi(\omega_1) - \phi(\omega_2) \right) 
	= \mathfrak{A}(\omega_1 + \omega_2) \int_0^{\omega_2} \int_0^{\omega_1} \mathrm{d}s\, \mathrm{d}s_0\, \phi''(s + s_0),
	\]
	and
	\[
	\mathfrak{A}(\omega_1 - \omega_2) \left( \phi(\omega_1) - \phi(\omega_1 - \omega_2) - \phi(\omega_2) \right) 
	= \mathfrak{A}(\omega_1 - \omega_2) \int_0^{\omega_2} \int_0^{\omega_1 - \omega_2} \mathrm{d}s\, \mathrm{d}s_0\, \phi''(s + s_0),
	\]
	yielding
	\begin{equation}\label{Lemma:Concave:E8b}
		\begin{aligned}
			& \mathfrak{A}(\omega_1 + \omega_2) \left( \phi(\omega_1 + \omega_2) - \phi(\omega_1) - \phi(\omega_2) \right) \  -\ \mathfrak{A}(\omega_1 - \omega_2) \left( \phi(\omega_1) - \phi(\omega_1 - \omega_2) - \phi(\omega_2) \right) \\
			=\ & [\mathfrak{A}(\omega_1 + \omega_2) - \mathfrak{A}(\omega_1 - \omega_2)] \int_0^{\omega_2} \int_0^{\omega_1 - \omega_2} \mathrm{d}s\, \mathrm{d}s_0\, \phi''(s + s_0) \\
			& + \ \mathfrak{A}(\omega_1 + \omega_2) \int_0^{\omega_2} \int_{\omega_1 - \omega_2}^{\omega_1} \mathrm{d}s\, \mathrm{d}s_0\, \phi''(s + s_0).
		\end{aligned}
	\end{equation}
Note that the derivatives in the above equalities are taken in the weak sense.

	Therefore,
	\begin{equation}\label{Lemma:Concave:E9}
		\begin{aligned}
			& \int_{\mathbb{R}_+} \mathrm{d}\omega\, C_{12}[f]\, \phi(\omega)\, \mathfrak{A}(\omega) \\
			=\ & 2c_{12} \iint_{\omega_1 > \omega_2} \mathrm{d}\omega_1\, \mathrm{d}\omega_2\, \mathfrak{A}(\omega_1)\, \mathfrak{A}(\omega_2)\, f(\omega_1) f(\omega_2) \\
			& \quad \times \left\{ [\mathfrak{A}(\omega_1 + \omega_2) - \mathfrak{A}(\omega_1 - \omega_2)] \int_0^{\omega_2} \int_0^{\omega_1 - \omega_2} \mathrm{d}s\, \mathrm{d}s_0\, \phi''(s + s_0) \right. \\
			& \qquad \left. + \ \mathfrak{A}(\omega_1 + \omega_2) \int_0^{\omega_2} \int_{\omega_1 - \omega_2}^{\omega_1} \mathrm{d}s\, \mathrm{d}s_0\, \phi''(s + s_0) \right\} \\
			& + c_{12} \iint_{\omega_1 = \omega_2} \mathrm{d}\omega_1\, \mathrm{d}\omega_2\, \mathfrak{A}(\omega_1)\, \mathfrak{A}(\omega_2)\, f(\omega_1) f(\omega_2)\, \mathfrak{A}(2\omega_1) \int_0^{\omega_1} \int_0^{\omega_1} \mathrm{d}s\, \mathrm{d}s_0\, \phi''(s + s_0).
		\end{aligned}
	\end{equation}
	Since \(\phi''(s) = -\frac{\alpha(1-\alpha)}{s^{2-\alpha}}\) when \(s\in(0,1)\) and \(\phi''(s) = 0\) when \(s>1\), we deduce from \eqref{Lemma:Concave:E9} that
	\begin{equation}\label{Lemma:Concave:E10}
		\begin{aligned}
			& \int_{\mathbb{R}_+} \mathrm{d}\omega\, C_{12}[f]\, \phi(\omega)\, \mathfrak{A}(\omega) \\
			\le\ & -2c_{12} \iint_{\omega_1 > \omega_2} \mathrm{d}\omega_1\, \mathrm{d}\omega_2\, \mathfrak{A}_1\, \mathfrak{A}_2\, f_1 f_2 \\
			& \quad \times \left\{ [\mathfrak{A}(\omega_1 + \omega_2) - \mathfrak{A}(\omega_1 - \omega_2)] \cdot \frac{\omega_2(\omega_1 - \omega_2)\, \alpha(1-\alpha)}{\omega_1^{2-\alpha}} \chi_{(0,1)}(\omega_1)\right.\\
			& \quad\quad \left.
			+\ \mathfrak{A}(\omega_1 + \omega_2) \cdot \frac{\omega_2 \omega_1\, \alpha(1-\alpha)}{(\omega_1 + \omega_2)^{2 - \alpha}} \chi_{(0,1)}(\omega_1+\omega_2)\right\} \\
			& - c_{12} \iint_{\omega_1 = \omega_2} \mathrm{d}\omega_1\, \mathrm{d}\omega_2\, \mathfrak{A}_1\, \mathfrak{A}_2\, f_1 f_2\, \mathfrak{A}(2\omega_1) \cdot \frac{\omega_1^2\, \alpha(1-\alpha)}{(2\omega_1)^{2 - \alpha}}\chi_{(0,1)}(2\omega_1),
		\end{aligned}
	\end{equation}
	where \(\mathfrak{A}_i := \mathfrak{A}(\omega_i)\) and \(f_i := f(\omega_i)\).
	
	{\it Step 2:} 
	We rewrite the second term on the right-hand side of \eqref{Lemma:Concave:E1} as
	\begin{equation}\label{Lemma:Concave:E2}
		\begin{aligned}
			& \int_{\mathbb{R}_+} \mathrm{d}\omega\, C_{22}[f]\, \phi(\omega)\,\mathfrak A(\omega) \\
			=\, & c_{22} \iiint_{\mathbb{R}_+^3} \mathrm{d}\omega_1\,\mathrm{d}\omega_2\,\mathrm{d}\omega\, 
			\Theta(\omega)\Theta(\omega_1)\Theta(\omega_2)\Theta(\omega+\omega_1-\omega_2)		\left[ \frac{\max\{|\omega_1 - \omega_2|,\ |\omega_2 - \omega|\}}{2(\omega + \omega_1)} \right]^\mu  \\
			& \quad \times \min\big\{ |k|(\omega), |k|(\omega_1), |k|(\omega_2), |k|(\omega+\omega_1-\omega_2) \big\} 
			f_1 f_2 f \\
			& \quad \times \big[-\phi(\omega) - \phi(\omega_1) + \phi(\omega_2) + \phi(\omega+\omega_1-\omega_2)\big] \\
			=\, & c_{22} \iiint_{\mathbb{R}_+^3} \mathrm{d}\omega_1\,\mathrm{d}\omega_2\,\mathrm{d}\omega\, f_1 f_2 f\,  		\left[ \frac{\max\{|\omega_1 - \omega_2|,\ |\omega_2 - \omega|\}}{2(\omega + \omega_1)} \right]^\mu \\
			& \quad \times \Big\{ 
			[-\phi(\omega_{\text{Max}}) - \phi(\omega_{\text{Min}}) + \phi(\omega_{\text{Mid}}) + \phi(\omega_{\text{Max}} + \omega_{\text{Min}} - \omega_{\text{Mid}})] \\
			& \qquad \times\Theta(\omega_{\text{Max}})\Theta(\omega_{\text{Min}})\Theta(\omega_{\text{Mid}})\Theta(\omega_{\text{Max}} + \omega_{\text{Min}} - \omega_{\text{Mid}}) |k_{\text{Min}}| \\
			& \quad + [-\phi(\omega_{\text{Max}}) - \phi(\omega_{\text{Mid}}) + \phi(\omega_{\text{Min}}) + \phi(\omega_{\text{Max}} + \omega_{\text{Mid}} - \omega_{\text{Min}})] \\
			& \qquad \times\Theta(\omega_{\text{Max}})\Theta(\omega_{\text{Mid}})\Theta(\omega_{\text{Min}})\Theta(\omega_{\text{Max}} + \omega_{\text{Mid}} - \omega_{\text{Min}}) |k_{\text{Min}}| \\
			& \quad + [-\phi(\omega_{\text{Min}}) - \phi(\omega_{\text{Mid}}) + \phi(\omega_{\text{Max}}) + \phi(\omega_{\text{Min}} + \omega_{\text{Mid}} - \omega_{\text{Max}})] \\
			& \qquad \times\Theta(\omega_{\text{Max}})\Theta(\omega_{\text{Mid}})\Theta(\omega_{\text{Min}})\Theta(\omega_{\text{Min}} + \omega_{\text{Mid}} - \omega_{\text{Max}})\mathbf{1}_{\omega_{\text{Min}} + \omega_{\text{Mid}} - \omega_{\text{Max}}\ge0}\\
			& \qquad \times \min\big\{ |k|(\omega_{\text{Max}}), |k|(\omega_{\text{Min}}), |k|(\omega_{\text{Mid}}), |k|(\omega_{\text{Min}} + \omega_{\text{Mid}} - \omega_{\text{Max}}) \big\} 
			\Big\},
		\end{aligned}
	\end{equation}
	where we assume that $|k_{\text{Min}}|$ is associated to $\omega_{\text{Min}}$, $|k_{\text{Max}}|$ to $\omega_{\text{Max}}$, and $|k_{\text{Mid}}|$ to $\omega_{\text{Mid}}$.

	We observe that
	\begin{equation}\label{Lemma:Concave:E3}
		\begin{aligned}
			& \left[-\phi(\omega_{\text{Min}}) - \phi(\omega_{\text{Mid}}) + \phi(\omega_{\text{Max}}) + \phi(\omega_{\text{Min}} + \omega_{\text{Mid}} - \omega_{\text{Max}})\right] \\
			=\, & \int_{0}^{\omega_{\text{Max}} - \omega_{\text{Min}}} \mathrm{d}\xi_1 \int_{0}^{\omega_{\text{Max}} - \omega_{\text{Mid}}} \mathrm{d}\xi_2 \, \phi''(\xi_1 + \xi_2 + \omega_{\text{Min}}) \ 
			\le\,   0,
		\end{aligned}
	\end{equation}
	since $\phi'' \le 0$.

	We now compute
	\begin{equation}\label{Lemma:Concave:E4}
		\begin{aligned}
			& \left[-\phi(\omega_{\text{Max}}) - \phi(\omega_{\text{Min}}) + \phi(\omega_{\text{Mid}}) + \phi(\omega_{\text{Max}} + \omega_{\text{Min}} - \omega_{\text{Mid}})\right]\Theta(\omega_{\text{Max}} + \omega_{\text{Min}} - \omega_{\text{Mid}}) \\
			& + \left[-\phi(\omega_{\text{Max}}) - \phi(\omega_{\text{Mid}}) + \phi(\omega_{\text{Min}}) + \phi(\omega_{\text{Max}} + \omega_{\text{Mid}} - \omega_{\text{Min}})\right]\Theta(\omega_{\text{Max}} + \omega_{\text{Mid}} - \omega_{\text{Min}}) \\
			=\, & -\int_{0}^{\omega_{\text{Mid}} - \omega_{\text{Min}}} \mathrm{d}s \int_{0}^{\omega_{\text{Max}} - \omega_{\text{Mid}}} \mathrm{d}s_0\,\Theta(\omega_{\text{Max}} + \omega_{\text{Min}} - \omega_{\text{Mid}})\, \phi''(\omega_{\text{Min}} + s + s_0) \\
			& + \int_{0}^{\omega_{\text{Mid}} - \omega_{\text{Min}}} \mathrm{d}s \int_{0}^{\omega_{\text{Max}} - \omega_{\text{Min}}} \mathrm{d}s_0\,\Theta(\omega_{\text{Max}} - \omega_{\text{Min}} + \omega_{\text{Mid}})\, \phi''(\omega_{\text{Min}} + s + s_0) \\
			=\, & \int_{0}^{\omega_{\text{Mid}} - \omega_{\text{Min}}} \mathrm{d}s \int_{0}^{\omega_{\text{Max}} - \omega_{\text{Mid}}} \mathrm{d}s_0\, \phi''(\omega_{\text{Min}} + s + s_0) \\
			& \quad \times \left[\Theta(\omega_{\text{Max}} - \omega_{\text{Min}} + \omega_{\text{Mid}}) -\Theta(\omega_{\text{Max}} + \omega_{\text{Min}} - \omega_{\text{Mid}}) \right] \\
			& + \int_{0}^{\omega_{\text{Mid}} - \omega_{\text{Min}}} \mathrm{d}s \int_{0}^{\omega_{\text{Mid}} - \omega_{\text{Min}}} \mathrm{d}s_0\,\Theta(\omega_{\text{Max}} - \omega_{\text{Min}} + \omega_{\text{Mid}})\, \phi''(\omega_{\text{Min}} + s + s_0) \\
			\le\, & \int_{0}^{\omega_{\text{Mid}} - \omega_{\text{Min}}} \mathrm{d}s \int_{0}^{\omega_{\text{Mid}} - \omega_{\text{Min}}} \mathrm{d}s_0\,\Theta(\omega_{\text{Max}} - \omega_{\text{Min}} + \omega_{\text{Mid}})\, \phi''(\omega_{\text{Min}} + s + s_0) \\
			\le\, & -\int_{0}^{\omega_{\text{Mid}} - \omega_{\text{Min}}} \mathrm{d}s \int_{0}^{\omega_{\text{Mid}} - \omega_{\text{Min}}} \mathrm{d}s_0\,\Theta(\omega_{\text{Max}} - \omega_{\text{Min}} + \omega_{\text{Mid}})\, \frac{\alpha(1 - \alpha)}{(\omega_{\text{Min}} + s + s_0)^{2 - \alpha}}\chi_{(0,1)}(\omega_{\text{Min}} + s + s_0) \\
			\le\, & -\Theta(\omega_{\text{Max}} - \omega_{\text{Min}} + \omega_{\text{Mid}})\, \frac{(\omega_{\text{Mid}} - \omega_{\text{Min}})^2 \alpha(1 - \alpha)}{(2\omega_{\text{Mid}} - \omega_{\text{Min}})^{2 - \alpha}}\chi_{(0,1)}(2\omega_{\text{Mid}} - \omega_{\text{Min}}).
		\end{aligned}
	\end{equation}
	
	Combining \eqref{Lemma:Concave:E2}, \eqref{Lemma:Concave:E3}, and \eqref{Lemma:Concave:E4}, we obtain
	\begin{equation}\label{Lemma:Concave:E5}
		\begin{aligned}
			\int_{\mathbb{R}_+} \mathrm{d}\omega\, C_{22}[f]\, \phi\, \mathfrak A(\omega)
			\le\ & -c_{22} \iiint_{\mathbb{R}_+^{3}} \mathrm{d}\omega_1\, \mathrm{d}\omega_2\, \mathrm{d}\omega\, f_1 f_2 f\, \frac{(\omega_{\text{Mid}} - \omega_{\text{Min}})^2 \alpha(1 - \alpha)}{(2\omega_{\text{Mid}} - \omega_{\text{Min}})^{2 - \alpha}}\, |k_{\text{Min}}| \\
			& \times\Theta(\omega_{\text{Max}})\,\Theta(\omega_{\text{Min}})\,\Theta(\omega_{\text{Mid}})\,\Theta(\omega_{\text{Max}} - \omega_{\text{Min}} + \omega_{\text{Mid}}) \\
			& \times \left[ \frac{\max\{|\omega_1 - \omega_2|,\ |\omega_2 - \omega|\}}{2(\omega + \omega_1)} \right]^\mu\chi_{(0,1)}(2\omega_{\mathrm{Mid}} - \omega_{\mathrm{Min}}).
		\end{aligned}
	\end{equation}
	
	{\it Step 3:} 
	Finally, we rewrite the last term on the right-hand side of \eqref{Lemma:Concave:E1} as
	\begin{equation*}\label{Lemma:Concave:E11}
		\begin{aligned}
			& \int_{\mathbb{R}_+} \mathrm{d}\omega\, C_{31}[f]\, \phi(\omega)\, \mathfrak{A}(\omega) \\
			=\, & c_{31} \iiint_{\mathbb{R}_+^3} \mathrm{d}\omega_1\, \mathrm{d}\omega_2\, \mathrm{d}\omega_3\, \mathfrak{A}(\omega_1 + \omega_2 + \omega_3)\, \mathfrak{A}(\omega_1)\, \mathfrak{A}(\omega_2)\, \mathfrak{A}(\omega_3) \\
			& \quad \times \left[ f(\omega_1) f(\omega_2) f(\omega_3) 
			- f(\omega_1 + \omega_2 + \omega_3) \left( f(\omega_1) f(\omega_2) + f(\omega_2) f(\omega_3) + f(\omega_3) f(\omega_1) \right) \right] \\
			& \quad \times \left[ \phi(\omega_1 + \omega_2 + \omega_3) - \phi(\omega_1) - \phi(\omega_2) - \phi(\omega_3) \right].
		\end{aligned}
	\end{equation*}
	
	After rearranging the terms, this expression can be rewritten as
	\begin{equation*}\label{Lemma:Concave:E12}
		\begin{aligned}
			& \int_{\mathbb{R}_+} \mathrm{d}\omega\, C_{31}[f]\, \phi(\omega)\, \mathfrak{A}(\omega) \\
			=\, & c_{31}  \iiint_{\mathbb{R}_+^3} \mathrm{d}\omega_1\, \mathrm{d}\omega_2\, \mathrm{d}\omega_3\, \mathfrak{A}(\omega_1 + \omega_2 + \omega_3)\, \mathfrak{A}(\omega_1)\, \mathfrak{A}(\omega_2)\, \mathfrak{A}(\omega_3)\, f(\omega_1) f(\omega_2) f(\omega_3) \\
			& \quad \times\left[ \phi(\omega_1 + \omega_2 + \omega_3) - \phi(\omega_1) - \phi(\omega_2) - \phi(\omega_3) \right] \\
			& \quad - 3c_{31} \iiint_{\mathbb{R}_+^3} \mathrm{d}\omega_1\, \mathrm{d}\omega_2\, \mathrm{d}\omega_3\, \mathfrak{A}(\omega_1 + \omega_2 + \omega_3)\, \mathfrak{A}(\omega_1)\, \mathfrak{A}(\omega_2)\, \mathfrak{A}(\omega_3) \\
			& \qquad \times f(\omega_1 + \omega_2 + \omega_3) f(\omega_2) f(\omega_3)\, \left[ \phi(\omega_1 + \omega_2 + \omega_3) - \phi(\omega_1) - \phi(\omega_2) - \phi(\omega_3) \right].
		\end{aligned}
	\end{equation*}

	Next, by performing the change of variables $\omega_1 + \omega_2 + \omega_3 \to \omega_1$, we obtain
	\begin{equation}
		\begin{aligned}\label{Lemma:Concave:E12aa}
			& \int_{\mathbb{R}_+} \mathrm{d}\omega\, C_{31}[f]\, \phi(\omega)\, \mathfrak{A}(\omega) \\
			=\, & c_{31}  \iiint_{\mathbb{R}_+^3} \mathrm{d}\omega_1\, \mathrm{d}\omega_2\, \mathrm{d}\omega_3\, \mathfrak{A}(\omega_1 + \omega_2 + \omega_3)\, \mathfrak{A}(\omega_1)\, \mathfrak{A}(\omega_2)\, \mathfrak{A}(\omega_3)\, f(\omega_1) f(\omega_2) f(\omega_3) \\
			& \quad \times \left[ \phi(\omega_1 + \omega_2 + \omega_3) - \phi(\omega_1) - \phi(\omega_2) - \phi(\omega_3) \right] \\
			& \quad - 3c_{31} \iiint_{\omega_1 - \omega_2 - \omega_3 \ge 0} \mathrm{d}\omega_1\, \mathrm{d}\omega_2\, \mathrm{d}\omega_3\, \mathfrak{A}(\omega_1 - \omega_2 - \omega_3)\, \mathfrak{A}(\omega_1)\, \mathfrak{A}(\omega_2)\, \mathfrak{A}(\omega_3) \\
			& \qquad \times f(\omega_1) f(\omega_2) f(\omega_3)\, \left[ \phi(\omega_1) - \phi(\omega_1 - \omega_2 - \omega_3) - \phi(\omega_2) - \phi(\omega_3) \right],
		\end{aligned}
	\end{equation}
	which can be further rewritten as
	\begin{equation}\label{Lemma:Concave:E12a}
		\begin{aligned}
			& \int_{\mathbb{R}_+} \mathrm{d}\omega\, C_{31}[f]\, \phi(\omega)\, \mathfrak{A}(\omega) \\
			=\, & 3c_{31} \iiint_{\omega_1 > \omega_2 + \omega_3} \mathrm{d}\omega_1\, \mathrm{d}\omega_2\, \mathrm{d}\omega_3\, \mathfrak{A}(\omega_1)\, \mathfrak{A}(\omega_2)\, \mathfrak{A}(\omega_3)\, f(\omega_1) f(\omega_2) f(\omega_3) \\
			& \quad \times \left[ \mathfrak{A}(\omega_1 + \omega_2 + \omega_3) \left( \phi(\omega_1 + \omega_2 + \omega_3) - \phi(\omega_1) - \phi(\omega_2) - \phi(\omega_3) \right) \right. \\
			& \left. \qquad -\ \mathfrak{A}(\omega_1 - \omega_2 - \omega_3) \left( \phi(\omega_1) - \phi(\omega_1 - \omega_2 - \omega_3) - \phi(\omega_2) - \phi(\omega_3) \right) \right] \\
			& \quad + c_{31} \iiint_{\omega_1 = \omega_2 + \omega_3} \mathrm{d}\omega_1\, \mathrm{d}\omega_2\, \mathrm{d}\omega_3\,\mathfrak{A}(\omega_1)\, \mathfrak{A}(\omega_2)\, \mathfrak{A}(\omega_3)\, f(\omega_1) f(\omega_2) f(\omega_3)\\
						& \quad \times \mathfrak{A}(3\omega_1)\, \left[ \phi(2\omega_1) - \phi(\omega_1)- \phi(\omega_2)- \phi(\omega_3) \right] \\
			& \quad + c_{31} \iiint_{\mathbb{R}_+^3 \setminus \left( \{\omega_1 > \omega_2 + \omega_3\} \cup \{\omega_2 > \omega_1 + \omega_3\} \cup \{\omega_3 > \omega_1 + \omega_2\} \right)} \mathrm{d}\omega_1\, \mathrm{d}\omega_2\, \mathrm{d}\omega_3\, \mathfrak{A}(\omega_1 + \omega_2 + \omega_3) \\
			& \quad \times \mathfrak{A}(\omega_1)\, \mathfrak{A}(\omega_2)\, \mathfrak{A}(\omega_3)\, f(\omega_1) f(\omega_2) f(\omega_3)\, \left[ \phi(\omega_1 + \omega_2 + \omega_3) - \phi(\omega_1) - \phi(\omega_2) - \phi(\omega_3) \right],
		\end{aligned}
	\end{equation}
	with the observation that $\mathfrak{A}(0) = 0$.

	Next, we compute, using the fact that $\phi(0)=0$
	\begin{equation*}
		\begin{aligned}
			&	\mathfrak{A}(\omega_1 + \omega_2 + \omega_3) \left( \phi(\omega_1 + \omega_2 + \omega_3) - \phi(\omega_1) - \phi(\omega_2)  - \phi(\omega_3)\right) \\
			\ = \ & \mathfrak{A}(\omega_1 + \omega_2 + \omega_3) \left(  \int_0^{\omega_1} \int_0^{\omega_2+\omega_3}\mathrm{d}s\, \mathrm{d}s_0\, \phi''(s + s_0) +  \int_0^{\omega_2} \int_0^{\omega_3} \mathrm{d}s\, \mathrm{d}s_0\, \phi''(s + s_0) \right),
		\end{aligned}
	\end{equation*}
	
	and
	\begin{equation*}
		\begin{aligned}
			&	\mathfrak{A}(\omega_1 - \omega_2 - \omega_3) \left( \phi(\omega_1) - \phi(\omega_1 - \omega_2 - \omega_3) - \phi(\omega_2)  - \phi(\omega_3)\right) \\
			\ = \ & \mathfrak{A}(\omega_1- \omega_2 - \omega_3) \left( \int_0^{\omega_1- \omega_2 - \omega_3} \int_0^{\omega_2+\omega_3} \mathrm{d}s\, \mathrm{d}s_0\, \phi''(s + s_0) +  \int_0^{\omega_2} \int_0^{\omega_3} \mathrm{d}s\, \mathrm{d}s_0\, \phi''(s + s_0) \right),
		\end{aligned}
	\end{equation*}
	yielding
	\begin{equation}\label{Lemma:Concave:E13a}
		\begin{aligned}
			&	\mathfrak{A}(\omega_1 + \omega_2 + \omega_3) \left( \phi(\omega_1 + \omega_2 + \omega_3) - \phi(\omega_1) - \phi(\omega_2)  - \phi(\omega_3)\right) \\
			&	- \mathfrak{A}(\omega_1 - \omega_2 - \omega_3) \left( \phi(\omega_1) - \phi(\omega_1 - \omega_2 - \omega_3) - \phi(\omega_2)  - \phi(\omega_3)\right) \\
			=\ & [\mathfrak{A}(\omega_1 + \omega_2 + \omega_3) - \mathfrak{A}(\omega_1 - \omega_2 - \omega_3)] \\
			&\times \left( \int_0^{\omega_1- \omega_2 - \omega_3} \int_0^{\omega_2+\omega_3} \mathrm{d}s\, \mathrm{d}s_0\, \phi''(s + s_0) +  \int_0^{\omega_2} \int_0^{\omega_3} \mathrm{d}s\, \mathrm{d}s_0\, \phi''(s + s_0) \right) \\
			& + \  \mathfrak{A}(\omega_1 + \omega_2 + \omega_3) \left(  \int_{\omega_1- \omega_2 - \omega_3}^{\omega_1} \int_0^{\omega_2+\omega_3}\mathrm{d}s\, \mathrm{d}s_0\, \phi''(s + s_0) \right).
		\end{aligned}
	\end{equation}
	
	Hence,
	\begin{equation}\label{Lemma:Concave:E13}
		\begin{aligned}
			& \int_{\mathbb{R}_+} \mathrm{d}\omega\, C_{31}[f]\, \phi(\omega)\, \mathfrak{A}(\omega) \\
			=\, & 3c_{31} \iiint_{\omega_1 > \omega_2 + \omega_3} \mathrm{d}\omega_1\, \mathrm{d}\omega_2\, \mathrm{d}\omega_3\, \mathfrak{A}(\omega_1)\, \mathfrak{A}(\omega_2)\, \mathfrak{A}(\omega_3)\, f(\omega_1) f(\omega_2) f(\omega_3) \\
			& \quad \times \Big\{ \left[ \mathfrak{A}(\omega_1 + \omega_2 + \omega_3) - \mathfrak{A}(\omega_1 - \omega_2 - \omega_3) \right] \\
			& \quad \times \left( \int_0^{\omega_1 - \omega_2 - \omega_3} \int_0^{\omega_2 + \omega_3} \mathrm{d}s\, \mathrm{d}s_0\, \phi''(s + s_0) + \int_0^{\omega_2} \int_0^{\omega_3} \mathrm{d}s\, \mathrm{d}s_0\, \phi''(s + s_0) \right) \\
			& \quad +\ \mathfrak{A}(\omega_1 + \omega_2 + \omega_3) \int_{\omega_1 - \omega_2 - \omega_3}^{\omega_1} \int_0^{\omega_2 + \omega_3} \mathrm{d}s\, \mathrm{d}s_0\, \phi''(s + s_0) \Big\} \\
			& + c_{31} \iiint_{\omega_1 = \omega_2 + \omega_3} \mathrm{d}\omega_1\, \mathrm{d}\omega_2\, \mathrm{d}\omega_3\, \mathfrak{A}(\omega_1)\, \mathfrak{A}(\omega_2)\, \mathfrak{A}(\omega_3)\, f(\omega_1) f(\omega_2) f(\omega_3)\, \mathfrak{A}(2\omega_1) \\
			& \quad \times \left( \int_0^{\omega_1} \int_0^{\omega_2 + \omega_3} \mathrm{d}s\, \mathrm{d}s_0\, \phi''(s + s_0) + \int_0^{\omega_2} \int_0^{\omega_3} \mathrm{d}s\, \mathrm{d}s_0\, \phi''(s + s_0) \right)
			\\
			&  + c_{31} \iiint_{\mathbb{R}_+^3 \setminus \left( \{\omega_1 > \omega_2 + \omega_3\} \cup \{\omega_2 > \omega_1 + \omega_3\} \cup \{\omega_3 > \omega_1 + \omega_2\} \right)} \mathrm{d}\omega_1\, \mathrm{d}\omega_2\, \mathrm{d}\omega_3\, \mathfrak{A}(\omega_1 + \omega_2 + \omega_3) \\
			&\quad  \times \mathfrak{A}(\omega_1)\, \mathfrak{A}(\omega_2)\, \mathfrak{A}(\omega_3)\, f(\omega_1) f(\omega_2) f(\omega_3)\, \\
		&\quad  \times 	\left(  \int_0^{\omega_1} \int_0^{\omega_2+\omega_3}\mathrm{d}s\, \mathrm{d}s_0\, \phi''(s + s_0) +  \int_0^{\omega_2} \int_0^{\omega_3} \mathrm{d}s\, \mathrm{d}s_0\, \phi''(s + s_0) \right).
		\end{aligned}
	\end{equation}
Using the fact that
\[
\phi''(\omega) = \frac{\alpha(\alpha - 1)}{\omega^{2 - \alpha}} \text{ for } \omega<1,
 \text{ and } \phi''(\omega) = 0 \text{ for } \omega>1\]
and the following inequalities:
\begin{align*}
&	\int_0^{\omega_1 - \omega_2 - \omega_3} \int_0^{\omega_2 + \omega_3} \mathrm{d}s\, \mathrm{d}s_0\, \phi''(s + s_0) + \int_0^{\omega_2} \int_0^{\omega_3} \mathrm{d}s\, \mathrm{d}s_0\, \phi''(s + s_0)
\\
	\quad \le &  -\frac{(\omega_2 + \omega_3)(\omega_1 - \omega_2 - \omega_3)\alpha(1 - \alpha)}{\omega_1^{2 - \alpha}}\chi_{(0,1)}(\omega_1)  - \frac{\omega_2 \omega_3 \alpha(1 - \alpha)}{(\omega_2 + \omega_3)^{2 - \alpha}}\chi_{(0,1)}(\omega_2+\omega_3), \label{ineq:phi1}
\end{align*}
and
\begin{equation*} \label{ineq:phi2}
	\int_{\omega_1 - \omega_2 - \omega_3}^{\omega_1} \int_0^{\omega_2 + \omega_3} \mathrm{d}s\, \mathrm{d}s_0\, \phi''(s + s_0)
	\le -\frac{(\omega_2 + \omega_3)\omega_1 \alpha(1 - \alpha)}{(\omega_1 + \omega_2 + \omega_3)^{2 - \alpha}}\chi_{(0,1)}(\omega_1+\omega_2+\omega_3),
\end{equation*}
we bound
	\begin{equation}\label{Lemma:Concave:E14}
		\begin{aligned}
			& \int_{\mathbb{R}_+} \mathrm{d}\omega\, C_{31}[f]\, \phi(\omega)\, \mathfrak{A}(\omega) \\
			\le\ & -3c_{31}  \iiint_{\omega_1 > \omega_2 + \omega_3} \mathrm{d}\omega_1\, \mathrm{d}\omega_2\, \mathrm{d}\omega_3\, \mathfrak{A}_1\, \mathfrak{A}_2\, \mathfrak{A}_3\, f_1 f_2 f_3 \\
			& \quad \times \left\{ \left[ \mathfrak{A}(\omega_1 + \omega_2 + \omega_3) - \mathfrak{A}(\omega_1 - \omega_2 - \omega_3) \right] 
			\frac{(\omega_2+\omega_3)(\omega_1 - \omega_2 - \omega_3)\alpha(1-\alpha)}{\omega_1^{2-\alpha}} \chi_{(0,1)}(\omega_1)\right. \\
			& \qquad + \left[ \mathfrak{A}(\omega_1 + \omega_2 + \omega_3) - \mathfrak{A}(\omega_1 - \omega_2 - \omega_3) \right] 
			\frac{\omega_2\omega_3\alpha(1-\alpha)}{(\omega_2+\omega_3)^{2-\alpha}}\chi_{(0,1)}(\omega_2+\omega_3) \\
			& \qquad \left. +\ \mathfrak{A}(\omega_1 + \omega_2 + \omega_3) 
			\frac{(\omega_2+\omega_3)\omega_1 \alpha(1-\alpha)}{(\omega_1+\omega_2 +\omega_3)^{2-\alpha}}\chi_{(0,1)}(\omega_1+\omega_2+\omega_3) \right\} \\
			& - c_{31} \iiint_{\omega_1 = \omega_2 + \omega_3} \mathrm{d}\omega_1\, \mathrm{d}\omega_2\, \mathrm{d}\omega_3\, \mathfrak{A}_1\, \mathfrak{A}_2\, \mathfrak{A}_3\, f_1 f_2 f_3\, \mathfrak{A}(2\omega_1)\, 
			\frac{\omega_1^2 \alpha(1-\alpha)}{(2\omega_1)^{2-\alpha}}\chi_{(0,1)}(2\omega_1).
		\end{aligned}
	\end{equation}
	
	Combining \eqref{Lemma:Concave:E1}, \eqref{Lemma:Concave:E10}, \eqref{Lemma:Concave:E5} and \eqref{Lemma:Concave:E14}, we obtain
	\begin{equation}\label{Lemma:Concave:E15a}
		\begin{aligned}
			& \int_{\mathbb{R}^3} \mathrm{d}k\, f(T,\omega)\, \phi(\omega)
			- \int_{\mathbb{R}^3} \mathrm{d}k\, f(0,\omega)\, \phi(\omega) \\
			\le\ & -2c_{12} \int_0^T \mathrm{d}t \iint_{\omega_1 > \omega_2} \mathrm{d}\omega_1\, \mathrm{d}\omega_2\, 
			\mathfrak{A}_1\, \mathfrak{A}_2\, f_1 f_2 \\
			& \times \left\{ 
			[\mathfrak{A}(\omega_1 + \omega_2) - \mathfrak{A}(\omega_1 - \omega_2)] 
			\frac{\omega_2 (\omega_1 - \omega_2)\, \alpha(1-\alpha)}{\omega_1^{2-\alpha}} \chi_{(0,1)}(\omega_1)\right.\\
			&\left.\quad
			+ \mathfrak{A}(\omega_1 + \omega_2) 
			\frac{\omega_2 \omega_1\, \alpha(1-\alpha)}{(\omega_1 + \omega_2)^{2-\alpha}} 
			\chi_{(0,1)}(\omega_2+\omega_1)\right\} \\
			& - c_{12} \int_0^T \mathrm{d}t \iint_{\omega_1 = \omega_2} \mathrm{d}\omega_1\, \mathrm{d}\omega_2\,
			\mathfrak{A}_1\, \mathfrak{A}_2\, f_1 f_2\, \mathfrak{A}(2\omega_1)
			\frac{\omega_1^2\, \alpha(1-\alpha)}{(2\omega_1)^{2 - \alpha}}\chi_{(0,1)}(2\omega_1)\\
			& - c_{22} \int_0^T \mathrm{d}t \iiint_{\mathbb{R}_+^3} \mathrm{d}\omega_1\, \mathrm{d}\omega_2\, \mathrm{d}\omega\,
			f_1 f_2 f\, \frac{(\omega_{\text{Mid}} - \omega_{\text{Min}})^2\, \alpha(1 - \alpha)}{(2\omega_{\text{Mid}} - \omega_{\text{Min}})^{2 - \alpha}}\, |k_{\text{Min}}|\chi_{(0,1)}(2\omega_{\text{Mid}} - \omega_{\text{Min}}) \\
			& \quad \times\Theta(\omega_{\text{Max}})\,\Theta(\omega_{\text{Min}})\,\Theta(\omega_{\text{Mid}})\,\Theta(\omega_{\text{Max}} - \omega_{\text{Min}} + \omega_{\text{Mid}}) \left[ \frac{\max\{|\omega_1 - \omega_2|,\ |\omega_2 - \omega|\}}{2(\omega + \omega_1)} \right]^\mu \\
			& - 3c_{31} \int_0^T \mathrm{d}t \iiint_{\omega_1 > \omega_2 + \omega_3} \mathrm{d}\omega_1\, \mathrm{d}\omega_2\, \mathrm{d}\omega_3\,
			\mathfrak{A}_1\, \mathfrak{A}_2\, \mathfrak{A}_3\, f_1 f_2 f_3 \\
			& \quad \times \left\{ 
			[\mathfrak{A}(\omega_1 + \omega_2 + \omega_3) - \mathfrak{A}(\omega_1 - \omega_2 - \omega_3)] 
			\frac{(\omega_2+\omega_3)(\omega_1 - \omega_2 - \omega_3)\, \alpha(1-\alpha)}{\omega_1^{2-\alpha}} \chi_{(0,1)}(\omega_1)\right. \\
			& \qquad + [\mathfrak{A}(\omega_1 + \omega_2 + \omega_3) - \mathfrak{A}(\omega_1 - \omega_2 - \omega_3)] 
			\frac{\omega_2 \omega_3\, \alpha(1-\alpha)}{(\omega_2+\omega_3)^{2-\alpha}}\chi_{(0,1)}(\omega_2+\omega_3) \\
			& \qquad \left. + \mathfrak{A}(\omega_1 + \omega_2 + \omega_3)
			\frac{(\omega_2+\omega_3)\omega_1\, \alpha(1-\alpha)}{(\omega_1 + \omega_2 + \omega_3)^{2-\alpha}} 
		\chi_{(0,1)}(\omega_1+\omega_2+\omega_3)	\right\} \\
			& - c_{31} \int_0^T \mathrm{d}t \iiint_{\omega_1 = \omega_2 + \omega_3} \mathrm{d}\omega_1\, \mathrm{d}\omega_2\, \mathrm{d}\omega_3\,
			\mathfrak{A}_1\, \mathfrak{A}_2\, \mathfrak{A}_3\, f_1 f_2 f_3\, \mathfrak{A}(2\omega_1)\,
			\frac{\omega_1^2\, \alpha(1-\alpha)}{(2\omega_1)^{2-\alpha}}\chi_{(0,1)}(2\omega_1).
		\end{aligned}
	\end{equation}
	For $\alpha \in [1/2,1)$, $\omega \in [0,1]$, we have $1 + \omega \ge \frac{\omega^\alpha}{\alpha}$. Hence,
	\[
	\int_{\mathbb{R}^3} \mathrm{d}k\, f(0,\omega)\, \phi(\omega)
	\ \le\ \int_{\mathbb{R}^3} \mathrm{d}k\, f(0,\omega)\, [1 + \omega]
	\ \le\  \mathscr{M} + \mathscr{E},
	\]
	which, in combination with \eqref{Lemma:Concave:E15a}, yields \begin{equation}\label{Lemma:Concave:E15}
		\begin{aligned}
			\mathscr{M} + \mathscr{E} \ \ge\ & 2c_{12} \int_0^T \mathrm{d}t \iint_{\omega_1 > \omega_2} \mathrm{d}\omega_1\, \mathrm{d}\omega_2\, 
			\mathfrak{A}_1\, \mathfrak{A}_2\, f_1 f_2 \\
			& \times \left\{ 
			[\mathfrak{A}(\omega_1 + \omega_2) - \mathfrak{A}(\omega_1 - \omega_2)] 
			\frac{\omega_2 (\omega_1 - \omega_2)\alpha(1-\alpha)}{\omega_1^{2-\alpha}} \chi_{(0,1)}(\omega_1)
			\right.\\
			&\left.\quad+ \mathfrak{A}(\omega_1 + \omega_2) 
			\frac{\omega_2 \omega_1 \alpha(1-\alpha)}{(\omega_1 + \omega_2)^{2-\alpha}} \chi_{(0,1)}(\omega_1+\omega_2)
			\right\} \\
			& + c_{12} \int_0^T \mathrm{d}t \iint_{\omega_1 = \omega_2} \mathrm{d}\omega_1\, \mathrm{d}\omega_2\,
			\mathfrak{A}_1\, \mathfrak{A}_2\, f_1 f_2\, \mathfrak{A}(2\omega_1)
			\frac{\omega_1^2 \alpha(1-\alpha)}{(2\omega_1)^{2 - \alpha}}\chi_{(0,1)}(2\omega_1)\end{aligned}
		\end{equation} \begin{equation*}
				\begin{aligned}
			& + c_{22} \int_0^T \mathrm{d}t \iiint_{\mathbb{R}_+^3} \mathrm{d}\omega_1\, \mathrm{d}\omega_2\, \mathrm{d}\omega\,
			f_1 f_2 f\, \frac{(\omega_{\mathrm{Mid}} - \omega_{\text{Min}})^2\alpha (1 - \alpha)}{(2\omega_{\mathrm{Mid}} - \omega_{\text{Min}})^{2 - \alpha}}\, |k_{\mathrm{Min}}| \chi_{(0,1)}(2\omega_{\text{Mid}} - \omega_{\text{Min}})\\
			& \quad \times\Theta(\omega_{\mathrm{Max}})\,\Theta(\omega_{\mathrm{Min}})\,\Theta(\omega_{\mathrm{Mid}})\,\Theta(\omega_{\mathrm{Max}} - \omega_{\mathrm{Min}} + \omega_{\mathrm{Mid}})  \left[ \frac{\max\{|\omega_1 - \omega_2|,\ |\omega_2 - \omega|\}}{2(\omega + \omega_1)} \right]^\mu \\
			& + 3c_{31} \int_0^T \mathrm{d}t \iiint_{\omega_1 > \omega_2 + \omega_3} \mathrm{d}\omega_1\, \mathrm{d}\omega_2\, \mathrm{d}\omega_3\,
			\mathfrak{A}_1\, \mathfrak{A}_2\, \mathfrak{A}_3\, f_1 f_2 f_3 \\
			& \quad \times \left\{ 
			[\mathfrak{A}(\omega_1 + \omega_2 + \omega_3) - \mathfrak{A}(\omega_1 - \omega_2 - \omega_3)] 
			\frac{(\omega_2+\omega_3)(\omega_1 - \omega_2 - \omega_3)\alpha(1-\alpha)}{\omega_1^{2-\alpha}} \chi_{(0,1)}(\omega_1)\right. \\
			& \qquad + [\mathfrak{A}(\omega_1 + \omega_2 + \omega_3) - \mathfrak{A}(\omega_1 - \omega_2 - \omega_3)] 
			\frac{\omega_2 \omega_3\alpha (1-\alpha)}{(\omega_2+\omega_3)^{2-\alpha}} \chi_{(0,1)}(\omega_2+\omega_3)\\
			& \qquad \left. + \mathfrak{A}(\omega_1 + \omega_2 + \omega_3)
			\frac{(\omega_2+\omega_3)\omega_1\alpha (1-\alpha)}{(\omega_1 + \omega_2 + \omega_3)^{2-\alpha}} \chi_{(0,1)}(\omega_1+\omega_2+\omega_3)
			\right\} \\
			&+ c_{31} \int_0^T \mathrm{d}t \iiint_{\omega_1 = \omega_2 + \omega_3} \mathrm{d}\omega_1\, \mathrm{d}\omega_2\, \mathrm{d}\omega_3\,
			\mathfrak{A}_1\, \mathfrak{A}_2\, \mathfrak{A}_3\, f_1 f_2 f_3\, \mathfrak{A}(2\omega_1)\,
			\frac{\omega_1^2 \alpha(1-\alpha)}{(2\omega_1)^{2-\alpha}}\chi_{(0,1)}(2\omega_1).
		\end{aligned}
	\end{equation*}
	Inequality \eqref{Lemma:Concave:1} follows from \eqref{Lemma:Concave:E15}.

\end{proof}

\begin{lemma}
	\label{Lemma:MassEner} 	We assume Assumption A and Assumption B. 

		Let $f$ be a radial solution in the sense of \eqref{4wavemild} of the wave kinetic equation \eqref{4wave}. Then, for all $t > 0$, the following estimates hold:
		\begin{equation} \label{Lemma:MassEner:1}
			\int_{\mathbb{R}^3} \mathrm{d}k\, f(t,k) \le \int_{\mathbb{R}^3} \mathrm{d}k\, f(0,k) = \mathscr{M}.
		\end{equation}
		 
\end{lemma}

\begin{proof}
Using \eqref{Lemma:Concave:E8a} with $\phi = 1$, we have:
\begin{equation}\label{Lemma:MassEner:E1}
	\begin{aligned}
		& \int_{\mathbb{R}_+} \mathrm{d}\omega\, C_{12}[f]\, \phi(\omega)\, \mathfrak{A}(\omega) \\
		=\, & -2 c_{12} \iint_{\omega_1 > \omega_2} \mathrm{d}\omega_1\, \mathrm{d}\omega_2\, \mathfrak{A}(\omega_1)\, \mathfrak{A}(\omega_2)\, f(\omega_1) f(\omega_2)  \left[ \mathfrak{A}(\omega_1 + \omega_2) - \mathfrak{A}(\omega_1 - \omega_2)\right] \\
		& - c_{12} \iint_{\omega_1 = \omega_2} \mathrm{d}\omega_1\, \mathrm{d}\omega_2\, \mathfrak{A}(\omega_1)\, \mathfrak{A}(\omega_2)\, f(\omega_1) f(\omega_2)\, \mathfrak{A}(2\omega_1) \ \le \ 0.
	\end{aligned}
\end{equation}

Using \eqref{Lemma:Concave:E2} with $\phi = 1$, we have:
\begin{equation}\label{Lemma:MassEner:E2}
	\int_{\mathbb{R}_+} \mathrm{d}\omega\, C_{22}[f]\, \mathfrak{A}(\omega) = 0.
\end{equation}

Using \eqref{Lemma:Concave:E12a} with $\phi = 1$, we obtain:
\begin{equation}\label{Lemma:MassEner:E3}
	\begin{aligned}
		& \int_{\mathbb{R}_+} \mathrm{d}\omega\, C_{31}[f]\, \mathfrak{A}(\omega) \\
		=\, & -6c_{31} \iiint_{\omega_1 > \omega_2 + \omega_3} \mathrm{d}\omega_1\, \mathrm{d}\omega_2\, \mathrm{d}\omega_3\, \mathfrak{A}(\omega_1)\, \mathfrak{A}(\omega_2)\, \mathfrak{A}(\omega_3)\, f(\omega_1) f(\omega_2) f(\omega_3) \\
		& \quad \times  \left[ \mathfrak{A}(\omega_1 + \omega_2 + \omega_3) - \mathfrak{A}(\omega_1 - \omega_2 - \omega_3) \right] \\
		& \quad - 2 c_{31} \iiint_{\omega_1 = \omega_2 + \omega_3} \mathrm{d}\omega_1\, \mathrm{d}\omega_2\,\mathrm{d}\omega_3\, \mathfrak{A}(\omega_1)\, \mathfrak{A}(\omega_2)\, \mathfrak{A}(\omega_3)\, f(\omega_1) f(\omega_2) f(\omega_3)\, \mathfrak{A}(3\omega_1) \\
		& \quad - 2 c_{31} \iiint_{\mathbb{R}_+^3 \setminus \left( \{\omega_1 > \omega_2 + \omega_3\} \cup \{\omega_2 > \omega_1 + \omega_3\} \cup \{\omega_3 > \omega_1 + \omega_2\} \right)} \mathrm{d}\omega_1\, \mathrm{d}\omega_2\, \mathrm{d}\omega_3 \\
		& \quad \times \mathfrak{A}(\omega_1 + \omega_2 + \omega_3)\, \mathfrak{A}(\omega_1)\, \mathfrak{A}(\omega_2)\, \mathfrak{A}(\omega_3)\, f(\omega_1) f(\omega_2) f(\omega_3) \ \le \ 0.
	\end{aligned}
\end{equation}

Combining \eqref{Lemma:MassEner:E1}--\eqref{Lemma:MassEner:E3}, we obtain
\begin{equation*}
	\partial_t \int_{\mathbb{R}^3} \mathrm{d}k\, f(t,k) 
	= \int_{\mathbb{R}_+} \mathrm{d}\omega\, C_{12}[f]\, \mathfrak{A}(\omega) 
	+ \int_{\mathbb{R}_+} \mathrm{d}\omega\, C_{22}[f]\, \mathfrak{A}(\omega) 
	+ \int_{\mathbb{R}_+} \mathrm{d}\omega\, C_{31}[f]\, \mathfrak{A}(\omega) 
	\le 0,
\end{equation*}
which yields
\begin{equation*}
	\int_{\mathbb{R}^3} \mathrm{d}k\, f(t,k) \le \int_{\mathbb{R}^3} \mathrm{d}k\, f(0,k) = \mathscr{M}.
\end{equation*}

\end{proof}
\section{Global existence result}\label{Sec:Global}
\begin{proposition}
	\label{Lemma:Global} Let \( f_0(k) = f_0(|k|) \geq 0 \) be a radial initial condition satisfying
	\begin{equation}
		\label{Propo:Global:1}
		\int_{\mathbb{R}^3} \mathrm{d}k  f_0(k)\, = \mathscr{M}, \ \ \ 
		\int_{\mathbb{R}^3} \mathrm{d}k  f_0(k)\,\omega(k)\, = \mathscr{E}.
	\end{equation}

	Then, there exists a global in time radial mild solution \( f(t,k) = f(t,|k|) \) to \eqref{4wave} in the sense of \eqref{4wavemild}, such that
	\begin{equation}
		\label{Propo:Global:2}
		\int_{\mathbb{R}^3} \mathrm{d}k  f(t,k)\, \leq \mathscr{M},
	\end{equation}
	for all \( t \geq 0 \).
\end{proposition}
\begin{proof} The proof of the Proposition follows the strategy proposed in \cite{staffilani2024energy}. We recall that \( G \) is defined in \eqref{Theorem1:4}.
Using  \eqref{Lemma:C22:1}, \eqref{Lemma:Concave:E7a}, and \eqref{Lemma:Concave:E12aa}, we deduce that
\begin{equation}\label{Lemma:Global:E1}
	\begin{aligned}
		\int_{\mathbb{R}_+} \mathrm{d}\omega  \partial_t G(t,\omega)\, \varphi(\omega)\, 
		=\ & \mathfrak{R}_1[\varphi] + \mathfrak{R}_2[\varphi] + \mathfrak{R}_3[\varphi],
	\end{aligned}
\end{equation}
in which $\varphi$ is a  test function,
\begin{equation}\label{Lemma:Global:E1a}
	\begin{aligned}
		\mathfrak{R}_1 :=\ 
		& c_{12} \iint_{\mathbb{R}_+^2} \mathrm{d}\omega_1\, \mathrm{d}\omega_2\, \mathfrak{A}(\omega_1)\, \mathfrak{A}(\omega_2)\, f(\omega_1) f(\omega_2)\, \mathfrak{A}(\omega_1 + \omega_2)  \left[ \varphi(\omega_1 + \omega_2) - \varphi(\omega_1) - \varphi(\omega_2) \right] \\
		& - 2c_{12} \iint_{\omega_1 \ge \omega_2} \mathrm{d}\omega_1\, \mathrm{d}\omega_2\, \mathfrak{A}(\omega_1)\, \mathfrak{A}(\omega_1 - \omega_2)\, \mathfrak{A}(\omega_2)\, f(\omega_1) f(\omega_2)   \left[ \varphi(\omega_1) - \varphi(\omega_1 - \omega_2) - \varphi(\omega_2) \right],
	\end{aligned}
\end{equation}
\begin{equation}\label{Lemma:Global:E1b}
	\begin{aligned}
		\mathfrak{R}_2 :=\ 
		& c_{22} \iiint_{\mathbb{R}_+^3} \mathrm{d}\omega_1\, \mathrm{d}\omega_2\, \mathrm{d}\omega\, \delta(\omega + \omega_1 - \omega_2 - \omega_3)\, \Theta(\omega)\, \Theta(\omega_1)\, \Theta(\omega_2)\, \Theta(\omega+\omega_1-\omega_2) \\
		& \quad \times \left[ \frac{\max\{|\omega_1 - \omega_2|,\ |\omega_2 - \omega|\}}{2(\omega + \omega_1)} \right]^\mu \min\left\{ |k|(\omega), |k|(\omega_1), |k|(\omega_2), |k|(\omega+\omega_1-\omega_2) \right\} \\
		& \quad \times f_1 f_2 f \left[-\varphi(\omega) - \varphi(\omega_1) + \varphi(\omega_2) + \varphi(\omega+\omega_1-\omega_2)\right],
	\end{aligned}
\end{equation}
\begin{equation}\label{Lemma:Global:E1c}
	\begin{aligned}
		\mathfrak{R}_3 :=\ 
		& c_{31} \iiint_{\mathbb{R}_+^3} \mathrm{d}\omega_1\, \mathrm{d}\omega_2\, \mathrm{d}\omega_3\, \mathfrak{A}(\omega_1 + \omega_2 + \omega_3)\, \mathfrak{A}(\omega_1)\, \mathfrak{A}(\omega_2)\, \mathfrak{A}(\omega_3) \\
		& \quad \times f(\omega_1) f(\omega_2) f(\omega_3) \left[ \varphi(\omega_1 + \omega_2 + \omega_3) - \varphi(\omega_1) - \varphi(\omega_2) - \varphi(\omega_3) \right] \\
		& - 3c_{31} \iiint_{\omega_1 \ge \omega_2 + \omega_3} \mathrm{d}\omega_1\, \mathrm{d}\omega_2\, \mathrm{d}\omega_3\, \mathfrak{A}(\omega_1 - \omega_2 - \omega_3)\, \mathfrak{A}(\omega_1)\, \mathfrak{A}(\omega_2)\, \mathfrak{A}(\omega_3) \\
		& \quad \times f(\omega_1) f(\omega_2) f(\omega_3)\, \left[ \varphi(\omega_1) - \varphi(\omega_1 - \omega_2 - \omega_3) - \varphi(\omega_2) - \varphi(\omega_3) \right].
	\end{aligned}
\end{equation}

Those terms can be rewritten as
\begin{equation}\label{Lemma:Global:E2a}
	\begin{aligned}
		\mathfrak{R}_1 =\ 
		& c_{12} \iint_{\mathbb{R}_+^2} \mathrm{d}\omega_1\, \mathrm{d}\omega_2\, G_1 G_2\, \mathfrak{A}(\omega_1 + \omega_2) \left[ \varphi(\omega_1 + \omega_2) - \varphi(\omega_1) - \varphi(\omega_2) \right] \\
		& - 2c_{12} \iint_{\omega_1 \ge \omega_2} \mathrm{d}\omega_1\, \mathrm{d}\omega_2\, \mathfrak{A}(\omega_1 - \omega_2)\, G_1 G_2 \left[ \varphi(\omega_1) - \varphi(\omega_1 - \omega_2) - \varphi(\omega_2) \right],
	\end{aligned}
\end{equation}
\begin{equation}\label{Lemma:Global:E2b}
	\begin{aligned}
		\mathfrak{R}_2 =\ 
		& c_{22} \iiint_{\mathbb{R}_+^3} \mathrm{d}\omega_1\, \mathrm{d}\omega_2\, \mathrm{d}\omega\,   \Theta(\omega+\omega_1-\omega_2)\, \mathbf{1}_{\omega+\omega_1 \ge \omega_2} \\
		& \quad \times \left[ \frac{\max\{|\omega_1 - \omega_2|,\ |\omega_2 - \omega|\}}{2(\omega + \omega_1)} \right]^\mu G_1 G_2 G   \frac{\min\left\{ |k|, |k_1|, |k_2|, |k|(\omega + \omega_1 - \omega_2)\right\} }{|k|\,|k_1|\,|k_2|}\,\\
		& \quad \times \left[-\varphi(\omega) - \varphi(\omega_1) + \varphi(\omega_2) + \varphi(\omega+\omega_1-\omega_2)\right],
	\end{aligned}
\end{equation}
\begin{equation}\label{Lemma:Global:E2c}
	\begin{aligned}
		\mathfrak{R}_3 =\ 
		& c_{31} \iiint_{\mathbb{R}_+^3} \mathrm{d}\omega_1\, \mathrm{d}\omega_2\, \mathrm{d}\omega_3\, \mathfrak{A}(\omega_1 + \omega_2 + \omega_3)\, G_1 G_2 G_3   \left[ \varphi(\omega_1 + \omega_2 + \omega_3) - \varphi(\omega_1) - \varphi(\omega_2) - \varphi(\omega_3) \right] \\
		& - 3c_{31} \iiint_{\omega_1 \ge \omega_2 + \omega_3} \mathrm{d}\omega_1\, \mathrm{d}\omega_2\, \mathrm{d}\omega_3\, \mathfrak{A}(\omega_1 - \omega_2 - \omega_3)\,  \\
		& \quad \times G_1 G_2 G_3 \left[ \varphi(\omega_1) - \varphi(\omega_1 - \omega_2 - \omega_3) - \varphi(\omega_2) - \varphi(\omega_3) \right].
	\end{aligned}
\end{equation}

We define the  kernels as
\begin{equation}\label{Lemma:Global:E3}
	\begin{aligned}
		\mathfrak{K}^+_{12} &= \mathfrak{A}(\omega_1 + \omega_2), \ \ \ \ \ \ \ \ 
		\mathfrak{K}^-_{12} = \mathfrak{A}(\omega_1 - \omega_2),
	\end{aligned}
\end{equation}

\begin{equation}\label{Lemma:GlobalExist:E4}
	\begin{aligned}
		\mathfrak{K}_{22} =\ & \Theta(\omega+\omega_1-\omega_2)\left[ \frac{\max\{|\omega_1 - \omega_2|,\ |\omega_2 - \omega|\}}{2(\omega + \omega_1)} \right]^\mu \\
		& \times \frac{\min\left\{ |k|, |k_1|, |k_2|, |k|(\omega + \omega_1 - \omega_2)\right\} }{|k|\,|k_1|\,|k_2|}\, \mathbf{1}_{\omega+\omega_1 \ge \omega_2},
	\end{aligned}
\end{equation}

\begin{equation}\label{Lemma:Global:E5}
	\begin{aligned}
		\mathfrak{K}^+_{31} &= \mathfrak{A}(\omega_1 + \omega_2 + \omega_3),  \ \ \ \ \ \ \ \ \ \ 
		\mathfrak{K}^-_{31} = \mathfrak{A}(\omega_1 - \omega_2 - \omega_3)
	\end{aligned}
\end{equation}
We can bound
\begin{equation}\label{Lemma:Global:E5a}
	\begin{aligned}
		|\mathfrak{K}^+_{12}| &\le \omega_1+\omega_2, \ \ \   
		|\mathfrak{K}^-_{12}| \le  \omega_1+\omega_2, \ \ \ \Theta(\omega+\omega_1-\omega_2) \le \omega+\omega_1+\omega_2,\\
		|\mathfrak{K}^+_{31}| &\le \omega_1 + \omega_2 + \omega_3, \ \ \ \ \ 	|\mathfrak{K}^-_{31}|  \le \omega_1 + \omega_2 + \omega_3.
	\end{aligned}
\end{equation}

We then rewrite the equation as:
\begin{equation}\label{Lemma:Global:E6}
	\int_{\mathbb{R}_+} \, \mathrm{d}\omega\, \partial_t G(t,\omega)\, \varphi(\omega) = \left\langle\mathfrak{T}[G],\varphi\right\rangle \ := \  \mathfrak{R}_1[\varphi] + \mathfrak{R}_2[\varphi] + \mathfrak{R}_3[\varphi],
\end{equation}
where
\begin{equation}\label{Lemma:Global:E6a}
	\begin{aligned}
		\mathfrak{R}_1 =\ 
		& c_{12} \iint_{\mathbb{R}_+^2} \mathrm{d}\omega_1\, \mathrm{d}\omega_2\, G_1 G_2 \, \mathfrak{K}^+_{12} \left[ \varphi(\omega_1 + \omega_2) - \varphi(\omega_1) - \varphi(\omega_2) \right] \\
		& - 2c_{12} \iint_{\omega_1 \ge \omega_2} \mathrm{d}\omega_1\, \mathrm{d}\omega_2\, \mathfrak{K}^-_{12}\, G_1G_2\left[ \varphi(\omega_1) - \varphi(\omega_1 - \omega_2) - \varphi(\omega_2) \right],
	\end{aligned}
\end{equation}
\begin{equation}\label{Lemma:Global:E6b}
	\begin{aligned}
		\mathfrak{R}_2 =\ 
		& c_{22} \iiint_{\mathbb{R}_+^3} \mathrm{d}\omega_1\, \mathrm{d}\omega_2\, \mathrm{d}\omega\, \delta(\omega + \omega_1 - \omega_2 - \omega_3)\, \mathfrak{K}_{22} \\
		& \quad \times G G_1 G_2 \left[ -\varphi(\omega) - \varphi(\omega_1) + \varphi(\omega_2) + \varphi(\omega+\omega_1-\omega_2) \right],
	\end{aligned}
\end{equation}
\begin{equation}\label{Lemma:Global:E6c}
	\begin{aligned}
		\mathfrak{R}_3 =\ 
		& c_{31} \iiint_{\mathbb{R}_+^3} \mathrm{d}\omega_1\, \mathrm{d}\omega_2\, \mathrm{d}\omega_3\, \mathfrak{K}^+_{31}\, G_1 G_2 G_3 \left[ \varphi(\omega_1 + \omega_2 + \omega_3) - \varphi(\omega_1) - \varphi(\omega_2) - \varphi(\omega_3) \right] \\
		& - 3c_{31} \iiint_{\omega_1 \ge \omega_2 + \omega_3} \mathrm{d}\omega_1\, \mathrm{d}\omega_2\, \mathrm{d}\omega_3\, \mathfrak{K}^-_{31}G_1G_2 G_3 \left[ \varphi(\omega_1) - \varphi(\omega_1 - \omega_2 - \omega_3) - \varphi(\omega_2) - \varphi(\omega_3) \right],
	\end{aligned}
\end{equation}
with initial condition
\[
G(0,\omega) = \Theta|k| {f_0(\omega)}.
\]
Similar to \eqref{Lemma:MassEner:1}, we have 
\begin{equation}\label{Lemma:Global:E7}
	\int_{\mathbb{R}_+} \mathrm{d}\omega\, G(t) \, (1+\omega)
	= \int_{\mathbb{R}_+} \mathrm{d}\omega\, G(0)\, (1+\omega)
	= \int_{\mathbb{R}_+} \mathrm{d}\omega\, \Theta\, f_0(\omega) \, {|k|} \, (1+\omega)
	= \mathscr{M} + \mathscr{E}.
\end{equation}

{\it Step 1:} It is clear that for any \( \varphi \in C^2(\mathbb{R}_+) \) such that the set \( \{ \omega \mid \varphi(\omega) \ne 0 \} \) is a compact subset of \( [0,\infty) \), 
 the expressions  
\[
\mathfrak{K}^+_{12} \left[ \varphi(\omega_1 + \omega_2) - \varphi(\omega_1) - \varphi(\omega_2) \right], \quad
\mathfrak{K}^-_{12} \left[ \varphi(\omega_1) - \varphi(\omega_1 - \omega_2) - \varphi(\omega_2) \right],
\]
\[
\mathfrak{K}^+_{31} \left[ \varphi(\omega_1 + \omega_2 + \omega_3) - \varphi(\omega_1) - \varphi(\omega_2) - \varphi(\omega_3) \right], \quad
\mathfrak{K}^-_{31} \left[ \varphi(\omega_1) - \varphi(\omega_1 - \omega_2 - \omega_3) - \varphi(\omega_2) - \varphi(\omega_3) \right]
\]
are continuous and uniformly bounded using the support of the test function.

We will now show that the quantity  
\[
\mathfrak{K} := \mathfrak{K}_{22} \left[ -\varphi(\omega) - \varphi(\omega_1) + \varphi(\omega_2) + \varphi(\omega + \omega_1 - \omega_2) \right]
\]
is also continuous and uniformly bounded on the support of \(  \left[ -\varphi(\omega) - \varphi(\omega_1) + \varphi(\omega_2) + \varphi(\omega + \omega_1 - \omega_2) \right]
 \).

We need to show that \( \mathfrak{K} \) remains uniformly bounded near the  singularities at \( |k| = 0 \), \( |k_1| = 0 \), and \( |k_2| = 0 \).

Firstly,  consider the case where exactly one of the quantities \( |k|, |k_1|, |k_2| \) approaches zero, while the other two remain bounded below by a fixed constant \( C_0 > 0 \).
Suppose \( |k| \) is small while \( |k_1|, |k_2| \ge C_0 \). Then we have
\begin{equation*}
	\begin{aligned}
		\mathfrak{K}_{22} 
		&\le \frac{1}{|k_1||k_2|} \, \Theta(\omega + \omega_1 - \omega_2) \,  \left[ \frac{\max\{|\omega_1 - \omega_2|,\ |\omega_2 - \omega|\}}{2(\omega + \omega_1)} \right]^\mu \\
		&\lesssim \frac{1}{C_0^{2+\mu}}\,  \left[ \frac{\max\{|\omega_1 - \omega_2|,\ |\omega_2 - \omega|\}}{2} \right]^\mu\, \Theta(\omega + \omega_1 - \omega_2),
	\end{aligned}
\end{equation*}
which is uniformly bounded on the support of \(  \left[ -\varphi(\omega) - \varphi(\omega_1) + \varphi(\omega_2) + \varphi(\omega + \omega_1 - \omega_2) \right]
 \).

Secondly, consider the case in which two of the quantities \( |k|, |k_1|, |k_2| \) approach zero, while the remaining one is bounded below by a constant \( C_0 > 0 \). Due to the restriction \( \omega + \omega_1 \ge \omega_2 \) on the domain of integration, the scenario where both \( |k| \) and \( |k_1| \) are small while \( |k_2| \ge C_0 \) cannot occur. Owing to the symmetry between \( k \) and \( k_1 \), it suffices to consider the case where \( |k| \) and \( |k_2| \) are small, while \( |k_1| \ge C_0 \). We write

\begin{equation*}
	\begin{aligned}
		\mathfrak{K} 
		=\, & \mathfrak{K}_{22} \big[-\varphi(\omega) - \varphi(\omega_1) + \varphi(\omega_2) + \varphi(\omega + \omega_1 - \omega_2)\big] \\
		=\, & \mathfrak{K}_{22} \Big[- \int_0^{\omega - \omega_2}\mathrm d\xi \big(\varphi'(\omega_2 + \xi) - \varphi'(0)\big) \,  - \int_0^{\omega - \omega_2} \mathrm d\xi\varphi'(0) \,  \\
		& \quad + \int_0^{\omega - \omega_2} \, \mathrm d\xi \big(\varphi'(\omega_1 + \xi) - \varphi'(\omega_1)\big) + \int_0^{\omega - \omega_2}\, \mathrm d\xi  \varphi'(\omega_1) \Big] \\
		=\, & \mathfrak{K}_{22} \Big[- \int_0^{\omega - \omega_2} \int_0^{\omega_2 + \xi} \, \mathrm d\xi'\mathrm d\xi\varphi''(\xi')  - \varphi'(0)(\omega - \omega_2) \\
		& \quad + \int_0^{\omega - \omega_2} \int_0^{\xi}  \, \mathrm d\xi' \mathrm d\xi \varphi''(\omega_1 + \xi') + \varphi'(\omega_1)(\omega - \omega_2) \Big].
	\end{aligned}
\end{equation*}
This expression can be bounded, using the compact support of $\varphi$, as

\begin{equation*}
	\begin{aligned}
		|\mathfrak{K}| 
		\lesssim\, & \frac{\min\{|k_1|, |k_2|, |k_3|, |k|\}}{|k||k_1||k_2|} \mathbf{1}_{\omega + \omega_1 \geq \omega_2}  \left[ \frac{1}{2(\omega + \omega_1)} \right]^\mu \\
		& \times \left( \left| \varphi'(\omega_1)(\omega - \omega_2) \right| + \left| \varphi'(0)(\omega - \omega_2) \right| + C \left[ \omega^2 + \omega_2^2 \right] \right) \\
		\lesssim\, & \frac{\min\{|k_1|, |k_2|, |k_3|, |k|\}}{|k||k_1||k_2|} \mathbf{1}_{\omega + \omega_1 \geq \omega_2}  \left[ \frac{1}{2 C_{\mathrm{disper}} C_0^{1/\delta}} \right]^\mu   \left( |\varphi'(\omega_1)|\, |\omega - \omega_2| + C \left[ \omega^2 + \omega_2^2 \right] \right),
	\end{aligned}
\end{equation*}
where \( C > 0 \) depends only on \( \varphi, \varphi', \varphi'' \).
We can further bound the quantity \( |\mathfrak{K}| \) as
\[
|\mathfrak{K}| \lesssim \frac{\min\{|k_1|, |k_2|, |k_3|, |k|\}}{|k||k_1||k_2|} \mathbf{1}_{\omega + \omega_1 \ge \omega_2} \left( | \omega - \omega_2 |\, |\varphi'(\omega_1)| + \omega^2 + \omega_2^2 \right).
\]

Since \(\min\{|k_1|, |k_2|, |k_3|, |k|\} \le \min\{|k|, |k_2|\}\), it follows that
\begin{equation}\label{Lemma:Global:E8}
	|\mathfrak{K}| \lesssim \frac{1}{\max\{|k|, |k_2|\} |k_1|} \mathbf{1}_{\omega + \omega_1 \ge \omega_2} \left( | \omega - \omega_2 |\, |\varphi'(\omega_1)| + \big(\max\{|k|, |k_2|\}\big)^{\frac{2}{\delta'}} \right),
\end{equation}
where we have used Assumption A.
Next, we estimate the second term on the right-hand side of \eqref{Lemma:Global:E8}:
\begin{equation}\label{Lemma:Global:E9}
	\frac{\left(\max\{|k|, |k_2|\}\right)^{\frac{2}{\delta'}}}{\max\{|k|, |k_2|\} |k_1|} \le \frac{\max\{|k|, |k_2|\}^{\frac{2}{\delta'}-1}}{C_o} \le \frac{1}{C_o}.
\end{equation}

Now we bound the first term on the right-hand side:
\begin{equation}\label{Lemma:Global:E10}
	\begin{aligned}
		\frac{1}{\max\{|k|, |k_2|\} |k_1|} | \omega - \omega_2 |\, |\varphi'(\omega_1)| 
		& \le \frac{\max\{\omega, \omega_2\}}{\max\{|k|, |k_2|\}} \cdot \frac{\omega_1}{|k_1|} \left| \frac{\varphi'(\omega_1)}{\omega_1} \right| \\
		& \le \frac{\omega \big( \max\{|k|, |k_2|\}\big)}{\max\{|k|, |k_2|\}} \cdot \frac{\omega_1}{|k_1|} \left| \frac{\varphi'(\omega_1)}{\omega_1} \right|.
	\end{aligned}
\end{equation}

We note that \(\frac{\omega}{|k|}\) is bounded for \(0 \le |k| \le 1\) when \(\omega \in \mathrm{supp}(\varphi)\). Furthermore, since
\[
\lim_{\omega_1 \to 0} \frac{\varphi'(\omega_1)}{\omega_1} = \varphi''(0),
\]
the quantity \(\frac{\varphi'(\omega_1)}{\omega_1}\) is also uniformly bounded. Hence, the expression in \eqref{Lemma:Global:E10} can be bounded by a constant.

Combining this with \eqref{Lemma:Global:E8} and \eqref{Lemma:Global:E9}, we conclude that \( |\mathfrak{K}| \) is uniformly bounded.

	Lastly, we consider the case when all \( |k|, |k_1|, |k_2| \) are close to zero. We develop the expression
	\begin{equation*}		\begin{aligned}
		\mathfrak{K}
		= &  \Theta(\omega + \omega_1 - \omega_2) \left[ \frac{\max\{|\omega_1 - \omega_2|,\ |\omega_2 - \omega|\}}{2(\omega + \omega_1)} \right]^\mu \\
		& \times \frac{\min\{|k_1|, |k_2|, |k|\}}{|k|\, |k_1|\, |k_2|} \mathbf{1}_{\omega + \omega_1 \ge \omega_2} \left[ -\int_0^{\omega - \omega_2} \mathrm{d}\xi \int_0^{\omega_2 - \omega_1} \mathrm{d}\xi'\, \varphi''(\omega_1 + \xi + \xi') \right],		\end{aligned}
	\end{equation*}
	which can be estimated on the support of $\varphi$ as
	\begin{equation}\label{Lemma:Global:E11}
		\begin{aligned}
			|\mathfrak{K}| \lesssim\ 
			&  \left[ \frac{\max\{|\omega_1 - \omega_2|,\ |\omega_2 - \omega|\}}{2(\omega + \omega_1)} \right]^\mu \\
			& \times \frac{\min\{|k_1|, |k_2|, |k|\}}{|k|\, |k_1|\, |k_2|} \mathbf{1}_{\omega + \omega_1 \ge \omega_2} |\omega - \omega_2|\, |\omega_2 - \omega_1|\, \|\varphi''\|_{L^\infty}.
		\end{aligned}
	\end{equation}
	
	We now consider $3$ cases.
	
\textit{Case 1: \( |k_2| = \min\{|k_1|, |k_2|, |k|\} \).} Since \( |k_2| = \min\{|k_1|, |k_2|, |k|\} \), it follows that \( \omega - \omega_2 \geq 0 \) and \( \omega_1 - \omega_2 \geq 0 \). We estimate \eqref{Lemma:Global:E11} as
\begin{equation*}
	\begin{aligned}
		|\mathfrak{K}| 
		\lesssim\ 
		& \left[ \frac{\max\{|\omega_1 - \omega_2|,\ |\omega_2 - \omega|\}}{2(\omega + \omega_1)} \right]^\mu 
		\frac{1}{|k|\,|k_1|} \mathbf{1}_{\omega + \omega_1 \geq \omega_2} |\omega - \omega_2|\, |\omega_2 - \omega_1|\, \|\varphi''\|_{L^\infty} \\
		\lesssim\ 
		&  
		\frac{\omega - \omega_2}{|k|} \cdot \frac{\omega_1 - \omega_2}{|k_1|} \|\varphi''\|_{L^\infty} \mathbf{1}_{\omega + \omega_1 \geq \omega_2}\\
		\lesssim\ 
		& 
		\frac{\omega - \omega_2}{|k|} \cdot \frac{\omega_1 - \omega_2}{|k_1|} \|\varphi''\|_{L^\infty} \mathbf{1}_{\omega + \omega_1 \geq \omega_2} \\
		\lesssim\ 
		& \frac{\omega}{|k|} \cdot \frac{\omega_1}{|k_1|} \|\varphi''\|_{L^\infty},
	\end{aligned}
\end{equation*}
which implies the boundedness of \( |\mathfrak{K}| \).

\textit{Case 2: \( |k_1| = \min\{|k_1|, |k_2|, |k|\} \).} We estimate \eqref{Lemma:Global:E11} as
\begin{equation*}
	\begin{aligned}
		|\mathfrak{K}|
		\lesssim\ 
		& \left[ \frac{\max\{|\omega_1 - \omega_2|,\ |\omega_2 - \omega|\}}{2(\omega + \omega_1)} \right]^\mu 
		\frac{1}{|k|\,|k_2|} \mathbf{1}_{\omega + \omega_1 \geq \omega_2} |\omega - \omega_2|\, |\omega_2 - \omega_1|\, \|\varphi''\|_{L^\infty} \\
		\lesssim\ 
		&  
		\frac{|\omega - \omega_2|}{|k|} \cdot \frac{\omega_2 - \omega_1}{|k_2|} \|\varphi''\|_{L^\infty} \mathbf{1}_{\omega + \omega_1 \geq \omega_2} \\
		\lesssim\ 
		& \frac{|\omega - \omega_2|}{|k|} \cdot \frac{\omega_2}{|k_2|} \|\varphi''\|_{L^\infty} \mathbf{1}_{\omega + \omega_1 \geq \omega_2}.
	\end{aligned}
\end{equation*}
In the case \( \omega \geq \omega_2 \), we bound \( \frac{|\omega - \omega_2|}{|k|} \leq \frac{\omega}{|k|} \), which implies the boundedness of \( |\mathfrak{K}| \). 

Suppose instead that \( \omega < \omega_2 \); then we estimate
\[
\frac{|\omega - \omega_2|}{|k|} = \frac{\omega_2 - \omega}{|k|} \leq \frac{\omega_1}{|k|} \leq \frac{\omega}{|k|},
\]
which also leads to the boundedness of \( |\mathfrak{K}| \).

\textit{Case 3: \( |k| = \min\{|k_1|, |k_2|, |k|\} \).} We estimate \eqref{Lemma:Global:E11} as
\begin{equation*}
	\begin{aligned}
		|\mathfrak{K}|
		\lesssim\ 
		& \left[ \frac{\max\{|\omega_1 - \omega_2|,\ |\omega_2 - \omega|\}}{2(\omega + \omega_1)} \right]^\mu 
		\frac{1}{|k_1|\,|k_2|} \mathbf{1}_{\omega + \omega_1 \geq \omega_2} |\omega - \omega_2|\, |\omega_2 - \omega_1|\, \|\varphi''\|_{L^\infty} \\
		\lesssim\ 
		&   
		\frac{\omega_2 - \omega}{|k_2|} \cdot \frac{|\omega_2 - \omega_1|}{|k_1|} \|\varphi''\|_{L^\infty} \mathbf{1}_{\omega + \omega_1 \geq \omega_2} \\
		\lesssim\ 
		& \frac{|\omega_2 - \omega_1|}{|k_1|} \cdot \frac{\omega_2}{|k_2|} \|\varphi''\|_{L^\infty} \mathbf{1}_{\omega + \omega_1 \geq \omega_2}.
	\end{aligned}
\end{equation*}

If \( \omega_1 \geq \omega_2 \), then since \(\omega + \omega_1 \geq \omega_2\), we have
\[
\frac{|\omega_2 - \omega_1|}{|k_1|} \leq \frac{\omega_1}{|k_1|},
\]
which implies that \( |\mathfrak{K}| \) is bounded.

If instead \( \omega_1 < \omega_2 \), we estimate
\[
\frac{|\omega_2 - \omega_1|}{|k_1|} = \frac{\omega_2 - \omega_1}{|k_1|} \leq \frac{\omega}{|k_1|} \leq \frac{\omega_1}{|k_1|},
\]
which again implies that \( |\mathfrak{K}| \) is bounded.

\textit{Step 2: Global existence for the cut-off models of \eqref{Lemma:Global:E6}.}

We define the cut-off kernels, for \( n > 0 \), as
\begin{equation}\label{Lemma:Global:E3n}
	\begin{aligned}
		\mathfrak{K}^+_{12n} &=  \mathfrak{A}(\omega_1 + \omega_2)\,\chi_{\{\omega_1 \le n\}}\,\chi_{\{\omega_2 \le n\}}, \qquad
		\mathfrak{K}^-_{12n} =   \mathfrak{A}(\omega_1 - \omega_2)_2)\,\chi_{\{\omega_1 \le n\}}\,\chi_{\{\omega_2 \le n\}},
	\end{aligned}
\end{equation}

\begin{equation}\label{Lemma:GlobalExist:E4n}
	\begin{aligned}
		\mathfrak{K}_{22n} =\ & \Theta(\omega + \omega_1 - \omega_2)\left[ \frac{\min\{\max\{|\omega_1 - \omega_2|,\ |\omega_2 - \omega|\}, n\}}{2(\omega + \omega_1)} \right]^\mu \\
		& \times \frac{ |k|, |k_1|, |k_2|, |k|(\omega + \omega_1 - \omega_2)}{|k|\,|k_1|\,|k_2|} \, \mathbf{1}_{\omega + \omega_1 \ge \omega_2}\,\chi_{\{\omega_1 \le n\}}\,\chi_{\{\omega_2 \le n\}}\,\chi_{\{\omega\le n\}},
	\end{aligned}
\end{equation}

\begin{equation}\label{Lemma:Global:E5n}
	\begin{aligned}
		\mathfrak{K}^+_{31n} &= \mathfrak{A}(\omega_1 + \omega_2 + \omega_3)\,\chi_{\{\omega_1 \le n\}}\,\chi_{\{\omega_2 \le n\}}\,\chi_{\{\omega_3\le n\}}, \ \qquad \\
		\mathfrak{K}^-_{31n} &=  \mathfrak{A}(\omega_1 - \omega_2 - \omega_3)\,\chi_{\{\omega_1 \le n\}}\,\chi_{\{\omega_2 \le n\}}\,\chi_{\{\omega_3\le n\}}.
	\end{aligned}
\end{equation}

We consider the approximated equations
\begin{equation}\label{Lemma:Global:E6n}
	\int_{\mathbb{R}_+} \mathrm{d}\omega\, \partial_t G^n(t,\omega)\, \varphi(\omega) 
	= \left\langle \mathfrak{T}^n[G^n], \varphi \right\rangle 
	:= \mathfrak{R}_1^n[\varphi] + \mathfrak{R}_2^n[\varphi] + \mathfrak{R}_3^n[\varphi],
\end{equation}
where
\begin{equation}\label{Lemma:Global:E6an}
	\begin{aligned}
		\mathfrak{R}_1^n[\varphi] =\ 
		& c_{12} \iint_{\mathbb{R}_+^2} \mathrm{d}\omega_1\, \mathrm{d}\omega_2\, G_1^n G_2^n \, \mathfrak{K}^+_{12n} \left[ \varphi(\omega_1 + \omega_2) - \varphi(\omega_1) - \varphi(\omega_2) \right] \\
		& - 2c_{12} \iint_{\omega_1 \ge \omega_2} \mathrm{d}\omega_1\, \mathrm{d}\omega_2\, \mathfrak{K}^-_{12n}\, G_1^n G_2^n \left[ \varphi(\omega_1) - \varphi(\omega_1 - \omega_2) - \varphi(\omega_2) \right],
	\end{aligned}
\end{equation}

\begin{equation}\label{Lemma:Global:E6bn}
	\begin{aligned}
		\mathfrak{R}_2^n[\varphi] =\ 
		& c_{22} \iiint_{\mathbb{R}_+^3} \mathrm{d}\omega\, \mathrm{d}\omega_1\, \mathrm{d}\omega_2\, \delta(\omega + \omega_1 - \omega_2 - \omega_3)\, \mathfrak{K}_{22n} \\
		& \quad \times G^n G_1^n G_2^n \left[ -\varphi(\omega) - \varphi(\omega_1) + \varphi(\omega_2) + \varphi(\omega + \omega_1 - \omega_2) \right],
	\end{aligned}
\end{equation}

\begin{equation}\label{Lemma:Global:E6cn}
	\begin{aligned}
		\mathfrak{R}_3^n[\varphi] =\ 
		& c_{31} \iiint_{\mathbb{R}_+^3} \mathrm{d}\omega_1\, \mathrm{d}\omega_2\, \mathrm{d}\omega_3\, \mathfrak{K}^+_{31n}\, G_1^n G_2^n G_3^n \left[ \varphi(\omega_1 + \omega_2 + \omega_3) - \varphi(\omega_1) - \varphi(\omega_2) - \varphi(\omega_3) \right] \\
		& - 3c_{31} \iiint_{\omega_1 \ge \omega_2 + \omega_3} \mathrm{d}\omega_1\, \mathrm{d}\omega_2\, \mathrm{d}\omega_3\, \mathfrak{K}^-_{31n}\, G_1^n G_2^n G_3^n \left[ \varphi(\omega_1) - \varphi(\omega_1 - \omega_2 - \omega_3) - \varphi(\omega_2) - \varphi(\omega_3) \right].
	\end{aligned}
\end{equation}

with initial condition
\[
G^n(0,\omega) = \Theta\, |k|\, f_0(\omega)\chi_{\{\omega\le n\}}.
\]

Similar to \eqref{Lemma:Global:E7}, we estimate
\begin{equation}\label{Lemma:Global:E12n}
	\int_{\mathbb{R}_+} \mathrm{d}\omega\, G^n(t) 
	\le \int_{\mathbb{R}_+} \mathrm{d}\omega\, G^n(0)
	\le \int_{\mathbb{R}_+} \mathrm{d}\omega\, \Theta\, f_0(\omega)\, |k| 
	= \mathscr{M},
\end{equation}
and
\begin{equation}\label{Lemma:Global:E12n}
	\int_{\mathbb{R}_+} \mathrm{d}\omega\, G^n(t) \,\omega
	\le \int_{\mathbb{R}_+} \mathrm{d}\omega\, G^n(0)\,\omega
	\le \int_{\mathbb{R}_+} \mathrm{d}\omega\, \Theta\, f_0(\omega)\, |k| \,\omega
	= \mathscr{E}.
\end{equation}
Using \eqref{Lemma:Global:E3n}-\eqref{Lemma:Global:E5n} estimate
\begin{equation*}\label{Lemma:Global:E12}
	\begin{aligned}
		\sup_{\|\varphi\|_{L^\infty} = 1} \left| \left\langle \mathfrak{T}^n[G^n], \varphi \right\rangle \right|
		\lesssim\ 
		&  \iiint_{\mathbb{R}_+^3} \mathrm{d}\omega_1\, \mathrm{d}\omega_2\, \mathrm{d}\omega\, |G_1^n G_2^n G^n| 
		+  \iint_{\mathbb{R}_+^2} \mathrm{d}\omega_1\, \mathrm{d}\omega\, |G_1^nG^n|  \\
		\lesssim\ 
		& \left( \int_{\mathbb{R}_+} \mathrm{d}\omega\, |G^n| \right)^3
		+ \left( \int_{\mathbb{R}_+} \mathrm{d}\omega\, |G^n| \right)^2.
	\end{aligned}
\end{equation*}

Next, we recall the definition \eqref{Radon} and observe that \( \mathfrak{T}^n \) is Lipschitz continuous on the set
\[
S_{\mathfrak{T}} := \left\{ G ~\middle|~ G \geq 0,\ \|G\|_{\mathscr{R}_+} \leq \mathscr{M} \right\}.
\]
Indeed, for any \( G, \bar G \in S_{\mathfrak{T}} \), we estimate
\begin{equation*}
	\begin{aligned}
		& \sup_{\|\varphi\|_{C^2} = 1} \left| \left\langle \mathfrak{T}^n[G] - \mathfrak{T}^n[\bar G], \varphi \right\rangle \right| \\
		\lesssim\ & \left| c_{12} \iint_{\mathbb{R}_+^2} \mathrm{d}\omega_1\, \mathrm{d}\omega_2\, (G_1 G_2 - \bar G_1  \bar G_2)\, \mathfrak{K}^+_{12n} \left[ \varphi(\omega_1 + \omega_2) - \varphi(\omega_1) - \varphi(\omega_2) \right] \right| \\
		& + \left| 2c_{12} \iint_{\omega_1 \ge \omega_2} \mathrm{d}\omega_1\, \mathrm{d}\omega_2\, \mathfrak{K}^-_{12n}\, (G_1 G_2 - \bar G_1\bar G_2) \left[ \varphi(\omega_1) - \varphi(\omega_1 - \omega_2) - \varphi(\omega_2) \right] \right| \\
		& + \left| c_{22} \iiint_{\mathbb{R}_+^3} \mathrm{d}\omega\, \mathrm{d}\omega_1\, \mathrm{d}\omega_2\, \delta(\omega + \omega_1 - \omega_2 - \omega_3)\, \mathfrak{K}_{22n} \right. \\
		& \qquad \left. \times (G G_1 G_2 - \bar G\bar G_1\bar G_2) \left[ -\varphi(\omega) - \varphi(\omega_1) + \varphi(\omega_2) + \varphi(\omega + \omega_1 - \omega_2) \right] \right| \\
		& + c_{31} \left| \iiint_{\mathbb{R}_+^3} \mathrm{d}\omega_1\, \mathrm{d}\omega_2\, \mathrm{d}\omega_3\, \mathfrak{K}^+_{31n}\, (G G_1 G_2 - \bar G\bar G_1\bar G_2)  \left[ \varphi(\omega_1 + \omega_2 + \omega_3) - \varphi(\omega_1) - \varphi(\omega_2) - \varphi(\omega_3) \right] \right| \\
		& + 3c_{31} \left| \iiint_{\omega_1 \ge \omega_2 + \omega_3} \mathrm{d}\omega_1\, \mathrm{d}\omega_2\, \mathrm{d}\omega_3\, \mathfrak{K}^-_{31n}\, (G G_1 G_2 - \bar G\bar G_1\bar G_2) \right. \\
		& \qquad \left. \times \left[ \varphi(\omega_1) - \varphi(\omega_1 - \omega_2 - \omega_3) - \varphi(\omega_2) - \varphi(\omega_3) \right] \right|.
	\end{aligned}
\end{equation*}

This yields the bound
\begin{equation*}
	\begin{aligned}
		\sup_{\|\varphi\|_{L^\infty} = 1} \left| \left\langle \mathfrak{T}^n[G] - \mathfrak{T}^n[\bar G], \varphi \right\rangle \right|
		\lesssim\ 
		& \iiint_{\mathbb{R}_+^3} \mathrm{d}\omega_1\, \mathrm{d}\omega_2\, \mathrm{d}\omega\, \left| G_1 G_2 G - \bar G_1\bar G_2\bar G \right|  + \iint_{\mathbb{R}_+^2} \mathrm{d}\omega_1\, \mathrm{d}\omega_2\, \left| G_1 G_2 - \bar G_1\bar G_2 \right| \\
		\lesssim\ 
		& \left( \mathscr{M}  + \mathscr{M}^2 \right) \| G - \bar G \|_{\mathscr{R}_+}.
	\end{aligned}
\end{equation*}

The local existence of a solution \( G^n \in C^1([0, T], {\mathscr{R}_+}([0,\infty))) \) to \eqref{Lemma:Global:E6n} then follows from a standard fixed point argument for small time \( T > 0 \) depending on \( \mathscr{M} \). 

Using the a priori bound \eqref{Lemma:Global:E12n}, we can iterate this argument from \( [0, T] \) to \( [T, 2T] \), then to \( [2T, 3T] \), and so on. This yields global existence of a solution \( G^n \in C^1([0, \infty), {\mathscr{R}_+}([0,\infty))) \).

\textit{Step 3: Global existence for \eqref{Lemma:Global:E6}.}

Let us consider the sequence of solutions obtained in Step 2, \( \{G^m\} \subset C^1([0, \infty), {\mathscr{R}_+}([0,\infty))) \). We bound, using \eqref{Lemma:Global:E5a} and \eqref{Lemma:Global:E12n}
\[
\left| \int_0^\infty\, \mathrm{d}\omega  G^m(t_2, \omega)\, \varphi(\omega)- \int_0^\infty\, \mathrm{d}\omega G^m(t_1, \omega)\, \varphi(\omega) \right| \leq C |t_2 - t_1|,
\]
where \( C > 0 \) is independent of \( m \). 

By the Arzela-Ascoli theorem (see \cite{giri2011continuous,stewart1989global}), there exists a subsequence \( \{G^{i_m}\} \) and a function 
\[
G^\infty \in C^1([0, \infty), {\mathscr{R}_+}([0,\infty)))
\]
such that \( G^{i_m} \rightharpoonup G^\infty \) weakly in \( {\mathscr{R}_+}([0,\infty)) \) and uniformly in $t$.

The function \( G^\infty \) is a mild solution to \eqref{Lemma:Global:E6} with \( n = \infty \). The inequalities \eqref{Propo:Global:2} follow from the same argument used to prove Lemma \ref{Lemma:MassEner}.

\end{proof}

\section{Multiscale setting}\label{Sec:Multiscale}

In this section, we fully extend the approach previously introduced in \cite{staffilani2024condensation}, which is inspired by the Domain Decomposition Method (DDM) commonly employed in parallel computing \cite{halpern2009nonlinear, Lions:1989:OSA, toselli2004domain}. This technique is based on the classical ``divide and conquer'' principle, wherein a large  domain is partitioned into smaller subdomains to facilitate both efficient computation and scale separation.

We first choose \( \mathfrak{m} \in \mathbb{N} \), with \( \mathfrak{m} > 10^6 \), and set
\[
R_{\mathfrak{m}} = 2^{-\mathfrak{m}}, \quad 
h_\mathfrak{m} = \frac{2^{-\mathfrak{m}}}{2^{\mathfrak{N}_\mathfrak{m}}}, \quad 
\mathfrak{N}_\mathfrak{m} = \left\lfloor \frac{\mathfrak{m}}{4(2 + \mu + \varrho)} \left( 2\delta - \frac{1}{2} - \varrho  \right) \right\rfloor,
\]
where we have used the inequality \( 2\delta  - \tfrac{1}{2} > \varrho \) in \eqref{Parameters}.

	We define:
	
	\begin{equation}\label{Sec:Growthlemmas:1}
		\begin{aligned}
			\text{(I) Number of Subdomains:} \quad & \mathscr{M}_{h_\mathfrak{m}, R_{\mathfrak{m}}} = 2^{\mathfrak{N}_\mathfrak{m}}, \\
			\text{(II) Nonoverlapping Subdomains:} \quad & \Omega_{i}^{h_\mathfrak{m}, R_\mathfrak{m}} = \left[ i h_\mathfrak{m}, (i+1) h_\mathfrak{m} \right), \quad i = 0, \ldots, 2^{\mathfrak{N}_\mathfrak{m}} - 2, \\
			& \Omega_{2^{\mathfrak{N}_\mathfrak{m}} - 1}^{h_\mathfrak{m}, R_\mathfrak{m}} = \left[ (2^{\mathfrak{N}_\mathfrak{m}} - 1) h_\mathfrak{m}, R_\mathfrak{m} \right), \\
			\text{(III) Overlapping Subdomains:} \quad & \mathscr{O}_{i}^{h_\mathfrak{m}, R_\mathfrak{m}} = \left[ (i-1) h_\mathfrak{m}, (i+2) h_\mathfrak{m} \right), \quad i = 1, \ldots, 2^{\mathfrak{N}_\mathfrak{m}} - 3, \\
			& \mathscr{O}_{0}^{h_\mathfrak{m}, R_\mathfrak{m}} = \left[ 0, 2 h_\mathfrak{m} \right), \quad \mathscr{O}_{1}^{h_\mathfrak{m}, R_\mathfrak{m}} = \left[ 0, 3 h_\mathfrak{m} \right), \\
			& \mathscr{O}_{2^{\mathfrak{N}_\mathfrak{m}} - 2}^{h_\mathfrak{m}, R_\mathfrak{m}} = \left[ (2^{\mathfrak{N}_\mathfrak{m}} - 3) h_\mathfrak{m}, R_\mathfrak{m} \right), \\
			& \mathscr{O}_{2^{\mathfrak{N}_\mathfrak{m}} - 1}^{h_\mathfrak{m}, R_\mathfrak{m}} = \left[ (2^{\mathfrak{N}_\mathfrak{m}} - 2) h_\mathfrak{m}, R_\mathfrak{m} \right).
		\end{aligned}
	\end{equation}
	
Using the parameters in Assumption A, we set
\begin{equation}\label{rho}\begin{aligned}
&	0< \rho < \min\left\{ \frac{2\delta - \frac{1}{2} - \varrho }{10(2 + \mu + \varrho)},\ {2\delta - \varrho},\ \tfrac{2}{3}\delta,\  \theta, \, \tfrac{\frac{2\delta - \frac{1}{2} - \varrho }{5(2 + \mu + \varrho)} + \varrho}{2} \right\}, 
	\quad \\
	&
	\max\left\{ 0,\ 2\rho - \varrho \right\} < \varepsilon < \min\left\{ 2\delta - \varrho -\rho,\ \frac{2\delta - \frac{1}{2} - \varrho }{5(2 + \mu + \varrho)} \right\},\end{aligned}
\end{equation}
	and define, for a fixed constant \( \mathcal{C}_* >0 \)
	\begin{equation}\label{Sec:Growthlemmas:2}
		\mathcal{A}_\mathfrak{m}^T := \left\{ t \in \left[0, T\right] : 
		\int_{\left[0, R_\mathfrak{m} \right)} \mathrm{d}\omega\, G(t, \omega) 
		\ge \mathcal{C}_* R_\mathfrak{m}^\rho \right\},
	\end{equation}
	\begin{equation}\label{Sec:Growthlemmas:3}
		\mathcal{A}_{\mathfrak{m},i}^T := \left\{ t \in \left[0, T\right] :
		\int_{\mathscr{O}_{i}^{h_\mathfrak{m}, R_\mathfrak{m}}} \mathrm{d}\omega\, G(t, \omega)
		\ge \mathcal{C}_* R_{\mathfrak{m}+1}^\rho \right\}, 
	\end{equation}
	for \( i = 0, \ldots, 2^{\mathfrak{N}_\mathfrak{m}} - 1 \), where $G$ is defined in \eqref{FDefinition}.

 We then define:
 \begin{equation}\label{Sec:Growthlemmas:4}
 	\begin{aligned}
 		\mathscr{B}_{\mathfrak{m}}^T &:= \mathcal{A}_{\mathfrak{m}}^T \setminus 
 		\bigcup_{i=0}^{2^{\mathfrak{N}_\mathfrak{m}} - 1} \mathcal{A}_{\mathfrak{m},i}^T, \ \ \ 
 		\mathscr{C}_{\mathfrak{m}}^T  := \bigcup_{i=2^{\mathfrak{N}_\mathfrak{m} - 1} - 1}^{2^{\mathfrak{N}_\mathfrak{m}} - 1}
 		\mathcal{A}_{\mathfrak{m},i}^T, \ \ \ 
 		\mathscr{D}_{\mathfrak{m}}^T := \bigcup_{i=0}^{2^{\mathfrak{N}_\mathfrak{m} - 1} - 2}
 		\mathcal{A}_{\mathfrak{m},i}^T.
 	\end{aligned}
 \end{equation}
 
Now, we present a plan for the multiscale estimates given in Sections~\ref{Sec:First}, \ref{Sec:Second}, \ref{Sec:Third}, and \ref{Sec:CondensateGrowth}, based on this multiscale setting.

\begin{itemize}
	\item In Section~\ref{Sec:First}, we estimate the size of the set $\mathscr{C}_{\mathfrak{m}}^T$, as stated in~\eqref{Lemma:Growth2:1} below.
	
	\item In Section~\ref{Sec:Second}, we estimate the size of the set $\mathscr{B}_{\mathfrak{m}}^T$, as given in~\eqref{Lemma:Growth1:1}. The proof of this estimate is based on~\eqref{Propo:Collision:1}.
	
	\item In Section~\ref{Sec:Third}, using the estimates for the sizes of the sets $\mathscr{B}_{\mathfrak{m}}^T$ and $\mathscr{C}_{\mathfrak{m}}^T$, we provide an estimate for the size of $\mathcal{A}_{\mathfrak{m}}^T$, given in~\eqref{Lemma:Growth3:1}.
	
	\item From~\eqref{Lemma:Growth3:1}, we deduce that the size of the set $\mathcal{A}_{\mathfrak{m}}^T$ tends to zero as $\mathfrak{m} \to 0$. In Section~\ref{Sec:CondensateGrowth}, Proposition~\ref{Lemma:Growth4} shows that this implies immediate condensation. Furthermore, Proposition~\ref{Lemma:Growth5} demonstrates that a weaker assumption on the initial condition leads to condensation in finite time.
\end{itemize}

\section{The first multiscale estimates}\label{Sec:First}
\begin{lemma}\label{Lemma:Super}
Let $\mathcal{K}(t,z) \in L^\infty([0,\infty), {\mathscr{R}_+}([0,\infty)))$ and assume that $\mathcal{K}(t,z) \ge 0$ for almost every $(t,z) \in [0,\infty)^2$. We consider the equation
\begin{equation}
	\label{Lemma:Super:1}
	\partial_t\phi(t,x) + \int_{0}^\infty \mathrm{d}z\, \mathcal{K}(t,z)\left[\phi(t,x+z) - \phi(t,x)\right] = 0.
\end{equation}

Let $\mathcal{P} : [0,\infty) \to [0,\infty)$ be a function satisfying
\begin{equation}
	\label{Lemma:Super:2}
	\mathcal{P}'(z_1+z_2) \ge \mathcal{P}'(z_1) \quad \text{for a.e. } z_1, z_2 \in [0,\infty).
\end{equation}

For a fixed time \( T_1 > 0 \), we define
\begin{equation}
	\label{Lemma:Super:3}
	\begin{aligned}
		\mathbb{X}_{\mathcal{K}}(t) &= \int_{t}^{T_1} \mathrm{d}s \int_{0}^\infty \mathrm{d}z\, \mathcal{K}(s,z)z, \\
		\mathbb{Y}_{\mathcal{K}}(t) &= \int_{0}^{t} \mathrm{d}s \int_{0}^\infty \mathrm{d}z\, \mathcal{K}(s,z)z.
	\end{aligned}
\end{equation}

Then the function
\begin{equation}
	\label{Lemma:Super:4}
	\psi(t,x) := e^{\mathbb{Y}_{\mathcal{K}}(t)} \mathcal{P}\left(\mathbb{X}_{\mathcal{K}}(t) + x\right)
\end{equation}
is a supersolution of \eqref{Lemma:Super:1}, in the sense that
\begin{equation}
	\label{Lemma:Super:5}
	\partial_t\psi(t,x) + \int_{0}^\infty \mathrm{d}z\, \mathcal{K}(t,z)\left[\psi(t,x+z) - \psi(t,x)\right] \ge 0,
\end{equation}
for almost every $(t,x) \in [0,\infty)^2$.

\end{lemma}
\begin{proof}

Plugging the expression~\eqref{Lemma:Super:4} into~\eqref{Lemma:Super:1}, we obtain
\begin{equation*}
	\begin{aligned}
		& \partial_t\psi(t,x) + \int_{0}^\infty \mathrm{d}z\, \mathcal{K}(t,z)\big[\psi(t,x+z) - \psi(t,z)\big] \\
		=\; & \mathbb{Y}_{\mathcal{K}}'(t)\, e^{\mathbb{Y}_{\mathcal{K}}(t)} \mathcal{P}\big(\mathbb{X}_{\mathcal{K}}(t)+x\big) 
		+ \mathbb{X}_{\mathcal{K}}'(t)\, e^{\mathbb{Y}_{\mathcal{K}}(t)} \mathcal{P}'\big(\mathbb{X}_{\mathcal{K}}(t)+x\big) \\
		& + \int_{0}^\infty \mathrm{d}z\, \mathcal{K}(t,z)\left[e^{\mathbb{Y}_{\mathcal{K}}(t)} \mathcal{P}\big(\mathbb{X}_{\mathcal{K}}(t)+x+z\big) - e^{\mathbb{Y}_{\mathcal{K}}(t)} \mathcal{P}\big(\mathbb{X}_{\mathcal{K}}(t)+x\big)\right] \\
		=\; & \mathbb{Y}_{\mathcal{K}}'(t)\, e^{\mathbb{Y}_{\mathcal{K}}(t)} \mathcal{P}\big(\mathbb{X}_{\mathcal{K}}(t)+x\big) 
		+ \mathbb{X}_{\mathcal{K}}'(t)\, e^{\mathbb{Y}_{\mathcal{K}}(t)} \mathcal{P}'\big(\mathbb{X}_{\mathcal{K}}(t)+x\big) \\
		& + \int_{0}^\infty \mathrm{d}z\, \mathcal{K}(t,z)\, e^{\mathbb{Y}_{\mathcal{K}}(t)} \int_{0}^z \mathrm{d}z'\,\mathcal{P}'\big(\mathbb{X}_{\mathcal{K}}(t)+x+z'\big),
	\end{aligned}
\end{equation*}
which, together with~\eqref{Lemma:Super:2} and~\eqref{Lemma:Super:3}, yields
\begin{equation*} \label{Lemma:Supersol:E3}
	\begin{aligned}
		& \partial_t\psi(t,x) + \int_{0}^\infty \mathrm{d}z\, \mathcal{K}(t,z)\big[\psi(t,x+z) - \psi(t,z)\big] \\
		=\; & \int_{0}^\infty \mathrm{d}z\, \mathcal{K}(t,z)\, e^{\mathbb{Y}_{\mathcal{K}}(t)} \mathcal{P}\big(\mathbb{X}_{\mathcal{K}}(t)+x\big) z
		- \int_{0}^\infty \mathrm{d}z\, \mathcal{K}(t,z)\, e^{\mathbb{Y}_{\mathcal{K}}(t)} \mathcal{P}'\big(\mathbb{X}_{\mathcal{K}}(t)+x\big)z \\
		& + \int_{0}^\infty \mathrm{d}z\, \mathcal{K}(t,z)\, e^{\mathbb{Y}_{\mathcal{K}}(t)} \int_{0}^z \mathrm{d}z'\,\mathcal{P}'\big(\mathbb{X}_{\mathcal{K}}(t)+x+z'\big) \\
		=\; & \int_{0}^\infty \mathrm{d}z\, \mathcal{K}(t,z)\, e^{\mathbb{Y}_{\mathcal{K}}(t)} \mathcal{P}\big(\mathbb{X}_{\mathcal{K}}(t)+x\big) z\\
		& + \int_{0}^\infty \mathrm{d}z\, \mathcal{K}(t,z)\, e^{\mathbb{Y}_{\mathcal{K}}(t)} \int_{0}^z \mathrm{d}z'\,\Big[\mathcal{P}'\big(\mathbb{X}_{\mathcal{K}}(t)+x+z'\big) - \mathcal{P}'\big(\mathbb{X}_{\mathcal{K}}(t)+x\big)\Big] \\
		\ge\; & \int_{0}^\infty \mathrm{d}z\, \mathcal{K}(t,z)\, e^{\mathbb{Y}_{\mathcal{K}}(t)} \mathcal{P}\big(\mathbb{X}_{\mathcal{K}}(t)+x\big)z \ \ge \ 0.
	\end{aligned}
\end{equation*}

This concludes the proof of the lemma.

\end{proof}

\begin{proposition}[First multiscale estimates] 	We assume Assumption A and Assumption B. 

	\label{Lemma:Growth2} Let \( T_0 \) be as defined in~\eqref{T0}. We choose $0\le T<T_0$ if $T_0>0$ and $ T=0$ if $T_0=0$. By the definitions in~\eqref{Sec:Growthlemmas:2},~\eqref{Sec:Growthlemmas:3}, and~\eqref{Sec:Growthlemmas:4}, there exists a constant \( \mathfrak{M}_* \in \mathbb{N} \) such that for all \( \mathfrak{m} > \mathfrak{M}_* \) and \( \mathfrak{m} \in \mathbb{N} \),
	\begin{equation} \label{Lemma:Growth2:1}
		\left| \mathscr{C}_{\mathfrak{m}}^T \right| \le {C}_\mathscr{C} R_{\mathfrak{m}}^{2\delta - \varepsilon - \varrho},
	\end{equation}
	for some universal constant \( C_\mathscr{C} > 0 \), independent of \( \delta, \varepsilon, \varrho, T, T_0 \).

\end{proposition}
\begin{proof}
We first prove, by contradiction, that for a constant $\mathfrak{C}_1 > 0$ independent of $\mathfrak{m}$ and $G$, there does not exist an infinite set 
$\mathfrak{S}_* = \{\mathfrak{m}_1, \mathfrak{m}_2, \dots\}$ such that for all $\mathfrak{m}_j \in \mathfrak{S}_*$, there exists a subset $\mathscr{C}_{\mathfrak{m}_j}'^{T}\subset \mathscr{C}_{\mathfrak{m}_j}^{T}$ and for each $t\in\mathscr{C}_{\mathfrak{m}_j}'^{T}$, there exists 
 $\sigma(t) \in \{2^{\mathfrak{N}_{\mathfrak{m}_j } - 1} - 1,\,\cdots, 2^{\mathfrak{N}_{\mathfrak{m}_j }} - 1\}$ satisfying
\begin{equation} \label{Lemma:Growth2:E0a}
	\int_{0}^{T} \mathrm{d}t \left[\int_{\mathbb{R}_+} \mathrm{d}\omega \,
	G(t, \omega)\chi_{\omega \in \mathscr O_{\sigma(t)}^{h_{\mathfrak{m}_j}, R_{{\mathfrak{m}_j }}}} \chi_{\mathscr{C}_{\mathfrak{m}_j}'^{T}}(t)\right]^2  
	> \frac{\mathfrak{C}_1 R_{{\mathfrak{m}_j }}^{2\delta + 2\rho - \varepsilon}}{\Theta(R_{\mathrm{m}_j}/2)}.
\end{equation}

If the above claim does not hold, then there exists a constant \(\mathfrak{M}_* \in \mathbb{N}\) such that for all \(\mathfrak{m} > \mathfrak{M}_*\) 
and for all $\sigma \in \{2^{\mathfrak{N}_{\mathfrak{m} } - 1} - 1,\,\cdots, 2^{\mathfrak{N}_{\mathfrak{m} }} - 1\}$
\begin{equation} \label{Lemma:Growth2:E0:1}
	\int_{0}^{T} \mathrm{d}t \left[\int_{\mathbb{R}_+} \mathrm{d}\omega \,
	G(t, \omega)\, \chi_{\left\{ \omega \in \mathscr{O}_{\sigma}^{h_{\mathfrak{m}}, R_{\mathfrak{m}}} \right\}}\, \chi_{\mathscr{C}_{\mathfrak{m}}^{T}}(t)\right]^2  
	\le \frac{\mathfrak{C}_1 R_{\mathfrak{m}}^{2\delta + 2\rho - \varepsilon}}{\Theta(R_{\mathfrak{m}}/2)},
\end{equation}
which, in combination with~\eqref{Sec:Growthlemmas:3} and Assumption A, implies 	
\begin{equation} \label{Lemma:Growth2:E0:2}
	\left| \mathscr{C}_{\mathfrak{m}}^T \right|\, R_{\mathfrak{m} +1}^{2\rho} 
	\le\ C' R_{\mathfrak{m}}^{2\delta + 2\rho - \varepsilon - \varrho},
\end{equation}
for some constant \(C' > 0\), yielding
\begin{equation} \label{Lemma:Growth2:E0:3}
	\left| \mathscr{C}_{\mathfrak{m}}^T \right| \le\ C' R_{\mathfrak{m}}^{2\delta - \varepsilon - \varrho}.
\end{equation}
We then obtain the desired estimate~\eqref{Lemma:Growth2:1}.

Suppose, for contradiction, that \eqref{Lemma:Growth2:E0a} holds. Then there must exist a time ${\mathscr{T}}_{{\mathfrak{m}_j }} \in [0, T]$ such that
\begin{equation} \label{Lemma:Growth2:E10b}
	\int_{0}^{\mathscr{T}_{{\mathfrak{m}_j }}} \mathrm{d}t \left[\int_{\mathbb{R}_+} \mathrm{d}\omega \,
	G(t, \omega)\chi_{\omega \in \mathscr O_{\sigma(t)}^{h_{\mathfrak{m}_j}, R_{{\mathfrak{m}_j }}}} \chi_{\mathscr{C}_{\mathfrak{m}_j}'^{{\mathscr{T}}_{{\mathfrak{m}_j }}}}(t)\right]^2  
	= \frac{\mathfrak{C}_1 R_{{\mathfrak{m}_j }}^{2\delta + 2\rho - \varepsilon}}{ \Theta(R_{{\mathfrak{m}_j }}/2)}.
\end{equation}

We define
\[
\mathfrak{G}_{h_{{\mathfrak{m}_j }}, R_{{\mathfrak{m}_j }}}(t, \omega) := G(t, \omega) \, \chi_{\omega \in \mathscr{O}_{\sigma(t)}^{h_{\mathfrak{m}_j}, R_{{\mathfrak{m}_j }}}} \, \chi_{\mathscr{C}_{\mathfrak{m}_j}'^{\mathscr{T}_{{\mathfrak{m}_j }}}}(t),
\]
and
\[
\mathcal{K}(t, z) :=  \frac{c_{22}C_{\text{disper}}^{2\delta}\Theta(R_{{\mathfrak{m}_j }}/2)}{ R_{{\mathfrak{m}_j }}^{2\delta}} 
\int_{\mathbb{R}_+} \mathrm{d}\omega \,
\mathfrak{G}_{h_{{\mathfrak{m}_j }}, R_{{\mathfrak{m}_j }}}(t, \omega) \,
\mathfrak{G}_{h_{{\mathfrak{m}_j }}, R_{{\mathfrak{m}_j }}}(t, \omega + z),
\]
for $(t, z) \in [0, \infty)\times[0,3h_\mathfrak m]$.  
We now divide the remainder of the proof into several smaller steps.

\subsubsection*{Step 1. \textit{The test function}}

We apply Lemma \ref{Lemma:Super} by defining the function $\mathcal{P}(z): [0,\infty) \to [0,\infty)$, using the same notation as in Lemma \ref{Lemma:Super}, as follows:
\begin{equation}\label{Lemma:ConcreteSuper:E2}
	\mathcal{P}(z) := \exp\left[\frac{1}{R_{{\mathfrak{m}_j }}}\left(\frac{R_{\mathfrak{m}_j}}{20} - z\right)_+\right] - 1 = \exp\left[\frac{1}{R_{\mathfrak{m}_j}} \max\left\{\frac{R_{\mathfrak{m}_j}}{20} - z, 0\right\}\right] - 1.
\end{equation}

Next, we verify that $\mathcal{P}$ satisfies condition \eqref{Lemma:Super:2}:
\begin{equation*}
	\mathcal{P}'(z) = -\frac{1}{R_{{\mathfrak{m}_j}}} \exp\left[\frac{1}{R_{\mathfrak{m}_j}} \max\left\{\frac{R_{\mathfrak{m}_j}}{20} - z, 0\right\}\right],
\end{equation*}
for almost every \( z \in [0, \infty) \). Note that the derivative is understood in the weak sense.

We now define the auxiliary functions:
\begin{equation}
	\begin{aligned}
		\label{Lemma:Supersolu:E3}
		\mathbb{X}_{\mathcal{K}}(t) &= \int_{t}^{\mathscr{T}_{h_{{\mathfrak{m}_j}}, R_{{\mathfrak{m}_j}}}} \mathrm{d}s \int_{0}^{3h_{\mathfrak{m}_j}} \mathrm{d}z\, \mathcal{K}(s,z)z, \\
		\mathbb{Y}_{\mathcal{K}}(t) &= \int_{0}^{t} \mathrm{d}s \int_{0}^{3h_{\mathfrak{m}_j}} \mathrm{d}z\, \mathcal{K}(s,z).
	\end{aligned}
\end{equation}

Then, by Lemma~\ref{Lemma:Super}, the function
\begin{equation}
	\label{Lemma:Supersolu:E4}
	\Psi(t,x) \ = \ \Psi_t(x) \  := \ e^{\mathbb{Y}_{\mathcal{K}}(t)}\, 	\mathcal{P}\left(\mathbb{X}_{\mathcal{K}}(t) + x\right)
\end{equation}
satisfies~\eqref{Lemma:Super:5}, and thus is a supersolution:
\begin{equation}
	\label{Lemma:Supersolu:E4a}
	\partial_t\Psi_t(x) + \int_{0}^\infty \mathrm{d}z\, \mathcal{K}\left[\Psi_t(x+z) - \Psi_t(x)\right] \ge 0,
\end{equation}
for almost every $(t,x) \in [0,\infty)^2$.

We observe that
\begin{equation}
	\label{DerivativePsi}
	\partial_x \Psi_t(x) = \mathcal{P}'\left(\mathbb{X}_{\mathcal{K}}(t) + x\right) \le 0, \qquad 
	\partial_{xx} \Psi_t(x) = \mathcal{P}''\left(\mathbb{X}_{\mathcal{K}}(t) + x\right) \ge 0
\end{equation}
for  almost every $(t,x) \in [0,\infty)^2$.

Next, we estimate
\begin{equation}
	\begin{aligned}
		\label{Lemma:Supersolu:E5}
		\mathbb{X}_{\mathcal{K}}(t) &= \int_{t}^{\mathscr{T}_{h_{\mathfrak{m}_j}, R_{\mathfrak{m}_j}}} \mathrm{d}s \int_{0}^{3h_{\mathfrak{m}_j}} \mathrm{d}z\, \mathcal{K}(s,z) z
		\ \le\ \int_{0}^{\mathscr{T}_{h_{\mathfrak{m}_j}, R_{\mathfrak{m}_j}}} \mathrm{d}s \int_{0}^{3h_{\mathfrak{m}_j}} \mathrm{d}z\, \mathcal{K}(s,z)z \\
		&\le\ \frac{ c_{22} C_{\text{disper}}^{2\delta} \Theta(R_{\mathfrak{m}_j}/2)}{R_{\mathfrak{m}_j}^{2\delta}} \cdot 3 h_{\mathfrak{m}_j} 
		\int_{0}^{\mathscr T_{{h}_{\mathfrak{m}_j}, {R}_{\mathfrak{m}_j}}} \mathrm{d}s \left[ \int_{\mathbb{R}_+} \mathrm{d}\omega\, \mathfrak{G}_{{h}_{\mathfrak{m}_j}, R_{\mathfrak{m}_j}}(s,\omega) \right]^2 \\
		&\le\ \frac{ c_{22} C_{\text{disper}}^{2\delta} \Theta(R_{\mathfrak{m}_j}/2)}{R_{\mathfrak{m}_j}^{2\delta}} \cdot 3 h_{\mathfrak{m}_j} 
		\cdot \frac{\mathfrak{C}_1 R_{\mathfrak{m}_j}^{2\delta + 2\rho - \varepsilon}}{\Theta(R_{\mathfrak{m}_j}/2)} \\
		&\le\  c_{22} C_{\text{disper}}^{2\delta} \cdot 3 \mathfrak{C}_1 h_{\mathfrak{m}_j} R_{\mathfrak{m}_j}^{2\rho - \varepsilon} \\
		&\le\  c_{22} C_{\text{disper}}^{2\delta} \cdot 3 \mathfrak{C}_1 \cdot \frac{R_{\mathfrak{m}_j}^{1 + 2\rho - \varepsilon}}{2^{\left\lfloor \frac{\mathfrak{m}_j}{4(2 + \mu + \varrho)} \left(2\delta - \tfrac{1}{2} - \varrho  \right) \right\rfloor}} \\
		&\le\  c_{22} C_{\text{disper}}^{2\delta} \cdot 3 \mathfrak{C}_1 \cdot \frac{R_{\mathfrak{m}_j}2^{\varepsilon\mathfrak{m}_j}}{2^{\left\lfloor \frac{\mathfrak{m}_j}{4(2 + \mu + \varrho)} \left(2\delta - \tfrac{1}{2} - \varrho  \right) \right\rfloor}} \  \le\ \frac{R_{\mathfrak{m}_j}}{100},
	\end{aligned}
\end{equation}
for $\mathfrak{m}_j$ sufficiently large, where we have used the definitions of ${h}_{\mathfrak{m}_j}$ and $R_{\mathfrak{m}_j}$ and the fact that 
$$\varepsilon  \le \frac{2\delta - \frac{1}{2} - \varrho }{5(2 + \mu + \varrho)}.$$

Similarly, we bound
\begin{equation}
	\begin{aligned}
		\label{Lemma:Supersolu:E6}    
		\mathbb{Y}_{\mathcal{K}}(t) 
		&= \int_{0}^{t} \mathrm{d}s \int_{0}^{3h_{\mathfrak{m}_j}} \mathrm{d}z\, \mathcal{K}(s,z) z
		\ \le\ \int_{0}^{\mathscr{T}_{h_{\mathfrak{m}_j}, R_{\mathfrak{m}_j}}} \mathrm{d}s \int_{0}^{3h_{\mathfrak{m}_j}} \mathrm{d}z\, \mathcal{K}(s,z)z \\
		&\le\ c_{22} C_{\text{disper}}^{2\delta} \cdot  3\mathfrak{C}_1 h_{\mathfrak{m}_j}  R_{\mathfrak{m}_j}^{2\rho - \varepsilon} \ \le\ \frac{1}{10^6},
	\end{aligned}
\end{equation}
for $j$ sufficiently large.

Combining \eqref{Lemma:Supersolu:E5} and \eqref{Lemma:Supersolu:E6}, we estimate
\begin{equation}
	\begin{aligned}
		\label{Lemma:Supersolu:E7} 
		\Psi(t,x) &= e^{\mathbb{Y}_{\mathcal{K}}(t)}\, \mathcal{P}\left(\mathbb{X}_{\mathcal{K}}(t) + x\right) \\
		&\le\ e^{\frac{1}{10^6}} \left[ \exp\left( \frac{1}{R_{\mathfrak{m}_j}} \max\left\{ \frac{R_{\mathfrak{m}_j}}{20} - x - \mathbb{X}_{\mathcal{K}}(t), 0 \right\} \right) - 1 \right] \\
		&\le\ e^{\frac{1}{10^6}} \left( e^{1/20} - 1 \right) \ \le\ 0.06.
	\end{aligned}
\end{equation}

When $x \in \left[0, \frac{R_{\mathfrak{m}_j}}{50}\right]$, we find
\begin{equation}
	\begin{aligned}
		\label{Lemma:Supersolu:E8} 
		\Psi(t,x) &= e^{\mathbb{Y}_{\mathcal{K}}(t)}\, \mathcal{P}\left(\mathbb{X}_{\mathcal{K}}(t) + x\right) \\
		&\ge\ \exp\left( \frac{1}{R_{\mathfrak{m}_j}} \max\left\{ \frac{R_{\mathfrak{m}_j}}{20} - x - \mathbb{X}_{\mathcal{K}}(t), 0 \right\} \right) - 1 \ 
		\ge\ e^{1/50} - 1 \ \ge\ 0.02.
	\end{aligned}
\end{equation}

Using $\Psi_t$ as a test function, similar to \eqref{Lemma:Concave:E8a}, \eqref{Lemma:Concave:E2}, and \eqref{Lemma:Concave:E12a}, we find
\begin{equation} \label{Lemma:Growth2:E1}
	\partial_t\left( \int_{\mathbb{R}^{+}} \mathrm{d}\omega\, G \Psi_t \right)
	= \int_{\mathbb{R}^{+}} \mathrm{d}\omega\, G \, \partial_t \Psi_t + \mathfrak{F}_1 + \mathfrak{F}_2 + \mathfrak{F}_3,
\end{equation}
where
\begin{equation} \label{Lemma:Growth2:E2}
	\begin{aligned}
		\mathfrak{F}_1 :=\; & 2c_{12} \iint_{\omega_1 > \omega_2} \mathrm{d}\omega_1 \mathrm{d}\omega_2\, \mathfrak{A}(\omega_1) \mathfrak{A}(\omega_2)\, f(\omega_1) f(\omega_2) \\
		& \quad \times \Big[ \mathfrak{A}(\omega_1 + \omega_2) \big( \Psi_t(\omega_1 + \omega_2) - \Psi_t(\omega_1) - \Psi_t(\omega_2) \big) \\
		& \qquad\quad - \mathfrak{A}(\omega_1 - \omega_2) \big( \Psi_t(\omega_1) - \Psi_t(\omega_1 - \omega_2) - \Psi_t(\omega_2) \big) \Big] \\
		& + c_{12} \iint_{\omega_1 = \omega_2} \mathrm{d}\omega_1 \mathrm{d}\omega_2\, \mathfrak{A}(\omega_1)^2 f(\omega_1)^2\, \mathfrak{A}(2\omega_1)\, \big[ \Psi_t(2\omega_1) - 2\Psi_t(\omega_1) \big],
	\end{aligned}
\end{equation}

\begin{equation} \label{Lemma:Growth2:E3}
	\begin{aligned}
		\mathfrak{F}_2 :=\; & c_{22} \iiint_{\mathbb{R}_+^3} \mathrm{d}\omega_1\,\mathrm{d}\omega_2\,\mathrm{d}\omega\, f_1 f_2 f\, \mathbf{1}_{\omega+\omega_1-\omega_2 \ge 0} 
		\left[ \frac{\max\{|\omega_1 - \omega_2|,\ |\omega_2 - \omega|\}}{2(\omega + \omega_1)} \right]^\mu \\
		& \times \Big\{ 
		[-\Psi_t(\omega_{\text{Max}}) - \Psi_t(\omega_{\text{Min}}) + \Psi_t(\omega_{\text{Mid}}) + \Psi_t(\omega_{\text{Max}} + \omega_{\text{Min}} - \omega_{\text{Mid}})] \\
		& \qquad \times \prod_{x \in \{\omega_{\text{Max}}, \omega_{\text{Min}}, \omega_{\text{Mid}}, \omega_{\text{Max}} + \omega_{\text{Min}} - \omega_{\text{Mid}}\}}\Theta(x)\, |k_{\text{Min}}| \\
		& + [-\Psi_t(\omega_{\text{Max}}) - \Psi_t(\omega_{\text{Mid}}) + \Psi_t(\omega_{\text{Min}}) + \Psi_t(\omega_{\text{Max}} + \omega_{\text{Mid}} - \omega_{\text{Min}})] \\
		& \qquad \times \prod_{x \in \{\omega_{\text{Max}}, \omega_{\text{Mid}}, \omega_{\text{Min}}, \omega_{\text{Max}} + \omega_{\text{Mid}} - \omega_{\text{Min}}\}}\Theta(x)\, |k_{\text{Min}}| \\
		& + [-\Psi_t(\omega_{\text{Min}}) - \Psi_t(\omega_{\text{Mid}}) + \Psi_t(\omega_{\text{Max}}) + \Psi_t(\omega_{\text{Min}} + \omega_{\text{Mid}} - \omega_{\text{Max}})] \\
		& \qquad \times \mathbf{1}_{\omega_{\text{Min}} + \omega_{\text{Mid}} - \omega_{\text{Max}} \ge 0} \prod_{x \in \{\omega_{\text{Max}}, \omega_{\text{Mid}}, \omega_{\text{Min}}, \omega_{\text{Min}} + \omega_{\text{Mid}} - \omega_{\text{Max}}\}}\Theta(x) \\
		& \qquad \times \min \left\{ |k|(\omega_{\text{Max}}), |k|(\omega_{\text{Min}}), |k|(\omega_{\text{Mid}}), |k|(\omega_{\text{Min}} + \omega_{\text{Mid}} - \omega_{\text{Max}}) \right\}
		\Big\}
	\end{aligned}
\end{equation}

and
\begin{equation} \label{Lemma:Growth2:E4}
	\begin{aligned}
		\mathfrak{F}_3 :=\; & 3c_{31} \iiint_{\omega_1 > \omega_2 + \omega_3} \mathrm{d}\omega_1\, \mathrm{d}\omega_2\, \mathrm{d}\omega_3\, \mathfrak{A}(\omega_1) \mathfrak{A}(\omega_2) \mathfrak{A}(\omega_3)\, f(\omega_1) f(\omega_2) f(\omega_3) \\
		& \times \Big[ \mathfrak{A}(\omega_1 + \omega_2 + \omega_3) \big( \Psi_t(\omega_1 + \omega_2 + \omega_3) - \Psi_t(\omega_1) - \Psi_t(\omega_2) - \Psi_t(\omega_3) \big) \\
		& \quad - \mathfrak{A}(\omega_1 - \omega_2 - \omega_3) \big( \Psi_t(\omega_1) - \Psi_t(\omega_1 - \omega_2 - \omega_3) - \Psi_t(\omega_2) - \Psi_t(\omega_3) \big) \Big] \\
		& + c_{31} \iint_{\omega_1 = \omega_2 + \omega_3} \mathrm{d}\omega_1\, \mathrm{d}\omega_2\, \mathfrak{A}(\omega_1) \mathfrak{A}(\omega_2) \mathfrak{A}(\omega_3)\, f(\omega_1) f(\omega_2) f(\omega_3) \\
		& \quad \times \mathfrak{A}(3\omega_1) \big[ \Psi_t(3\omega_1) - 3\Psi_t(\omega_1) \big] \\
		& + c_{31} \iiint_{\mathbb{R}_+^3 \setminus \left( \{\omega_1 > \omega_2 + \omega_3\} \cup \{\omega_2 > \omega_1 + \omega_3\} \cup \{\omega_3 > \omega_1 + \omega_2\} \right)} \mathrm{d}\omega_1\, \mathrm{d}\omega_2\, \mathrm{d}\omega_3 \\
		& \quad \times \mathfrak{A}(\omega_1 + \omega_2 + \omega_3)\, \mathfrak{A}(\omega_1) \mathfrak{A}(\omega_2) \mathfrak{A}(\omega_3)\, f(\omega_1) f(\omega_2) f(\omega_3) \\
		& \quad \times \big[ \Psi_t(\omega_1 + \omega_2 + \omega_3) - \Psi_t(\omega_1) - \Psi_t(\omega_2) - \Psi_t(\omega_3) \big].
	\end{aligned}
\end{equation}

\subsubsection*{Step 2. \textit{Controlling} \(\mathfrak{F}_1\)}
Similarly to \eqref{Lemma:Concave:E8b}, we compute
\begin{equation*}
	\begin{aligned}
		& \mathfrak{A}(\omega_1 + \omega_2) \left( \Psi_t(\omega_1 + \omega_2) - \Psi_t(\omega_1) - \Psi_t(\omega_2) \right) \  -\ \mathfrak{A}(\omega_1 - \omega_2) \left( \Psi_t(\omega_1) - \Psi_t(\omega_1 - \omega_2) - \Psi_t(\omega_2) \right) \\
		=\ & \left[\mathfrak{A}(\omega_1 + \omega_2) - \mathfrak{A}(\omega_1 - \omega_2)\right] 
		\int_0^{\omega_2} \int_0^{\omega_1 - \omega_2} \mathrm{d}\zeta\, \mathrm{d}\zeta_0\, \partial_{\omega}^2\Psi_t(\zeta + \zeta_0) \\
		& +\ \mathfrak{A}(\omega_1 + \omega_2) 
		\int_0^{\omega_2} \int_{\omega_1 - \omega_2}^{\omega_1} \mathrm{d}\zeta\, \mathrm{d}\zeta_0\, \partial_{\omega}^2\Psi_t(\zeta + \zeta_0) \  -\ \left(\mathfrak{A}(\omega_1 + \omega_2)-\mathfrak{A}(\omega_1 - \omega_2)\right)\Psi_t(0).
	\end{aligned}
\end{equation*}
Since $\partial_{\omega}^2\Psi_t(\zeta + \zeta_0) \ge 0$, we can estimate, taking into account the fact that $\mathfrak{A}(\omega_1 + \omega_2) - \mathfrak{A}(\omega_1 - \omega_2)\ge0$,
\begin{equation}\label{Lemma:Growth2:E5}
	\begin{aligned}
		& \mathfrak{A}(\omega_1 + \omega_2) \left( \Psi_t(\omega_1 + \omega_2) - \Psi_t(\omega_1) - \Psi_t(\omega_2) \right) \\
		&\quad -\ \mathfrak{A}(\omega_1 - \omega_2) \left( \Psi_t(\omega_1) - \Psi_t(\omega_1 - \omega_2) - \Psi_t(\omega_2) \right) \\
		\ge\ & -\left( \mathfrak{A}(\omega_1 + \omega_2) - \mathfrak{A}(\omega_1 - \omega_2) \right) \Psi_t(0).
	\end{aligned}
\end{equation}
Similarly, we also have
\begin{equation}\label{Lemma:Growth2:E6}
	\mathfrak{A}(2\omega_1)\, \big[ \Psi_t(2\omega_1) - 2\Psi_t(\omega_1) \big] \ \ge\ -\mathfrak{A}(2\omega_1)\, \Psi_t(0).
\end{equation}

Combining \eqref{Lemma:Growth2:E5} and \eqref{Lemma:Growth2:E6}, we obtain the estimate
\begin{equation}\label{Lemma:Growth2:E7}
	\begin{aligned}
		\mathfrak{F}_1 \ \ge\ & -2c_{12}\Psi_t(0) \iint_{R_{\mathfrak{m}_j} \ge \omega_1 > \omega_2\ge 0} \mathrm{d}\omega_1\, \mathrm{d}\omega_2\, \mathfrak{A}(\omega_1)\, \mathfrak{A}(\omega_2)\, f(\omega_1) f(\omega_2)\left( \mathfrak{A}(\omega_1 + \omega_2) - \mathfrak{A}(\omega_1 - \omega_2) \right)  \\
		& - c_{12}\Psi_t(0) \iint_{R_{\mathfrak{m}_j} \ge \omega_1 = \omega_2\ge 0} \mathrm{d}\omega_1\, \mathrm{d}\omega_2\, \mathfrak{A}(\omega_1)\, \mathfrak{A}(\omega_2)\, f(\omega_1) f(\omega_2)\, \mathfrak{A}(2\omega_1),
	\end{aligned}
\end{equation}
where we have used the fact that the supports of both $\mathfrak{A}(\omega_1 + \omega_2) \left( \Psi_t(\omega_1 + \omega_2) - \Psi_t(\omega_1) - \Psi_t(\omega_2) \right) - \mathfrak{A}(\omega_1 - \omega_2) \left( \Psi_t(\omega_1) - \Psi_t(\omega_1 - \omega_2) - \Psi_t(\omega_2) \right)$ and $\Psi_t(2\omega_1) - 2\Psi_t(\omega_1)$ are contained within the interval $[0,R_{\mathfrak{m}_j}/2]$. Using Assumption A, we deduce from \eqref{Lemma:Growth2:E7} that 
\begin{equation}\label{Lemma:Growth2:E8}
	\begin{aligned}
		\mathfrak{F}_1 \ \ge\ & -C_0\Psi_t(0)R_{\mathfrak{m}_j}^{\theta} \iint_{ \omega_1, \omega_2 \in [0,R_{\mathfrak{m}_j}]} \mathrm{d}\omega_1\, \mathrm{d}\omega_2\, \mathfrak{A}(\omega_1)\, \mathfrak{A}(\omega_2)\, f(\omega_1) f(\omega_2) \\
		\ge\ & -C_0\Psi_t(0)R_{\mathfrak{m}_j}^{\theta} \iint_{ \omega_1, \omega_2 \in [0,R_{\mathfrak{m}_j}]} \mathrm{d}\omega_1\, \mathrm{d}\omega_2\, G(\omega_1) G(\omega_2) \ 
		\ge\   -C_0R_{\mathfrak{m}_j}^{\theta}\Psi_t(0) \left[ \int_{ \omega \in [0,R_{\mathfrak{m}_j}]} \mathrm{d}\omega\,  G(\omega) \right]^2,
	\end{aligned}
\end{equation}
for some constant $C_0$ that is independent of $f$, \(\mathfrak m_j\) and may vary from line to line.

\subsubsection*{Step 3. \textit{Controlling} \(\mathfrak{F}_3\)}	
Similar to \eqref{Lemma:Concave:E13a}, we compute
\begin{equation*}
	\begin{aligned}
		& \mathfrak{A}(\omega_1 + \omega_2 + \omega_3) \left( \Psi_t(\omega_1 + \omega_2 + \omega_3) - \Psi_t(\omega_1) - \Psi_t(\omega_2)  - \Psi_t(\omega_3)\right) \\
		& \quad - \mathfrak{A}(\omega_1 - \omega_2 - \omega_3) \left( \Psi_t(\omega_1) - \Psi_t(\omega_1 - \omega_2 - \omega_3) - \Psi_t(\omega_2)  - \Psi_t(\omega_3)\right) \\
		=\ & [\mathfrak{A}(\omega_1 + \omega_2 + \omega_3) - \mathfrak{A}(\omega_1 - \omega_2 - \omega_3)] \\
		& \quad \times \left( \int_0^{\omega_1 - \omega_2 - \omega_3} \int_0^{\omega_2 + \omega_3} \mathrm{d}\zeta\, \mathrm{d}\zeta_0\, \partial_{\omega}^2\Psi_t(\zeta + \zeta_0) + \int_0^{\omega_2} \int_0^{\omega_3} \mathrm{d}\zeta\, \mathrm{d}\zeta_0\, \partial_{\omega}^2\Psi_t(\zeta + \zeta_0) \right) \\
		& \quad + \mathfrak{A}(\omega_1 + \omega_2 + \omega_3) \left( \int_{\omega_1 - \omega_2 - \omega_3}^{\omega_1} \int_0^{\omega_2 + \omega_3} \mathrm{d}\zeta\, \mathrm{d}\zeta_0\, \partial_{\omega}^2\Psi_t(\zeta + \zeta_0) \right) \\
		& \quad - 2\left(\mathfrak{A}(\omega_1 + \omega_2 + \omega_3) - \mathfrak{A}(\omega_1 - \omega_2 - \omega_3)\right)\Psi_t(0).
	\end{aligned}
\end{equation*}

Since \(\partial_{\omega}^2\Psi_t(\zeta + \zeta_0) \ge 0\) by \eqref{DerivativePsi}, and noting that \(\mathfrak{A}(\omega_1 + \omega_2 + \omega_3) - \mathfrak{A}(\omega_1 - \omega_2 - \omega_3) \ge 0\), we obtain the estimate
\begin{equation}\label{Lemma:Growth2:E9}
	\begin{aligned}
		& \mathfrak{A}(\omega_1 + \omega_2 + \omega_3) \left( \Psi_t(\omega_1 + \omega_2 + \omega_3) - \Psi_t(\omega_1) - \Psi_t(\omega_2) - \Psi_t(\omega_3) \right) \\
		& \quad - \mathfrak{A}(\omega_1 - \omega_2 - \omega_3) \left( \Psi_t(\omega_1) - \Psi_t(\omega_1 - \omega_2 - \omega_3) - \Psi_t(\omega_2) - \Psi_t(\omega_3) \right) \\
		\ge\ & -2\left(\mathfrak{A}(\omega_1 + \omega_2 + \omega_3) - \mathfrak{A}(\omega_1 - \omega_2 - \omega_3)\right)\Psi_t(0).
	\end{aligned}
\end{equation}

Similarly, we can also bound
\begin{equation}\label{Lemma:Growth2:E10}
	\begin{aligned}
		& \mathfrak{A}(\omega_1 + \omega_2 + \omega_3) \left( \Psi_t(\omega_1 + \omega_2 + \omega_3) - \Psi_t(\omega_1) - \Psi_t(\omega_2) - \Psi_t(\omega_3)\right) \\
		\ge\ & -2\Psi_t(0) \mathfrak{A}(\omega_1 + \omega_2 + \omega_3),
	\end{aligned}
\end{equation}
and
\begin{equation}\label{Lemma:Growth2:E11}
	\begin{aligned}
		\mathfrak{A}(3\omega_1) \left[ \Psi_t(3\omega_1) - 3\Psi_t(\omega_1) \right] 
		\ge\ -2\Psi_t(0)\mathfrak{A}(3\omega_1).
	\end{aligned}
\end{equation}

Combining \eqref{Lemma:Growth2:E9}, \eqref{Lemma:Growth2:E10}, and \eqref{Lemma:Growth2:E11}, we bound
\begin{equation} \label{Lemma:Growth2:E12}
	\begin{aligned}
		\mathfrak{F}_3 \ \ge \; & -6c_{31}\Psi_t(0) \iiint_{R_{\mathfrak{m}_j}\ge \omega_1 > \omega_2 + \omega_3\ge 0} \mathrm{d}\omega_1\, \mathrm{d}\omega_2\, \mathrm{d}\omega_3\, \mathfrak{A}(\omega_1) \mathfrak{A}(\omega_2) \mathfrak{A}(\omega_3)\, f(\omega_1) f(\omega_2) f(\omega_3) \\
		& \quad \times \left(\mathfrak{A}(\omega_1 + \omega_2 + \omega_3) - \mathfrak{A}(\omega_1 - \omega_2 - \omega_3)\right) \\
		& -2c_{31} \Psi_t(0)\iint_{R_{\mathfrak{m}_j}\ge \omega_1 = \omega_2 + \omega_3\ge0 } \mathrm{d}\omega_1\, \mathrm{d}\omega_2\, \mathfrak{A}(\omega_1) \mathfrak{A}(\omega_2) \mathfrak{A}(\omega_3)\, f(\omega_1) f(\omega_2) f(\omega_3) \mathfrak{A}(3\omega_1) \\
		& - 2c_{31} \Psi_t(0) \iiint_{[0,R_{\mathfrak{m}_j}]^3 \setminus \left( \{\omega_1 > \omega_2 + \omega_3\} \cup \{\omega_2 > \omega_1 + \omega_3\} \cup \{\omega_3 > \omega_1 + \omega_2\} \right)} \mathrm{d}\omega_1\, \mathrm{d}\omega_2\, \mathrm{d}\omega_3 \\
		& \quad \times \mathfrak{A}(\omega_1 + \omega_2 + \omega_3)\, \mathfrak{A}(\omega_1) \mathfrak{A}(\omega_2) \mathfrak{A}(\omega_3)\, f(\omega_1) f(\omega_2) f(\omega_3),
	\end{aligned}
\end{equation}
where we have used the fact that the supports of  $ \Psi_t(\omega_1 + \omega_2 + \omega_3) - \Psi_t(\omega_1) - \Psi_t(\omega_2) - \Psi_t(\omega_3)$, $ \Psi_t(\omega_1) - \Psi_t(\omega_1 - \omega_2 - \omega_3) - \Psi_t(\omega_2) - \Psi_t(\omega_3) $ and $\Psi_t(3\omega_1) - 3\Psi_t(\omega_1)$ are contained within the interval $[0,R_{\mathfrak{m}_j}/4]$. Using  Assumption A,  we deduce from \eqref{Lemma:Growth2:E12} that
\begin{equation}\label{Lemma:Growth2:E13}
	\begin{aligned}
		\mathfrak{F}_3 \ \ge\ & -C_0 R_{\mathfrak{m}_j}^{\theta} \Psi_t(0) \left[ \int_{\omega \in [0,R_{\mathfrak{m}_j}]} \mathrm{d}\omega\, G(\omega) \right]^3,
	\end{aligned}
\end{equation}
for some constant \(C_0\) that is independent of \(f\), \(\mathfrak m_j\) and may vary from line to line.

\subsubsection*{Step 4. \textit{Controlling} the leftover quantities}	

Combining \eqref{Lemma:Growth2:E1}, \eqref{Lemma:Growth2:E8} and \eqref{Lemma:Growth2:E13}, we bound
\begin{equation}	\label{Lemma:Growth2:E14}
	\begin{aligned}
		\partial_t\left( \int_{\mathbb{R}^{+}} \mathrm{d}\omega\, G \Psi_t \right)
		= \ &\int_{\mathbb{R}^{+}} \mathrm{d}\omega\, G \, \partial_t \Psi_t + \mathfrak{F}_2\ -\ C_0 R_{\mathfrak{m}_j}^{\theta} \Psi_t(0) \left[ \int_{\omega \in [0,R_{\mathfrak{m}_j}]} \mathrm{d}\omega\, G(\omega) \right]^2\\
		&\ -\ C_0 R_{\mathfrak{m}_j}^{\theta} \Psi_t(0) \left[ \int_{\omega \in [0,R_{\mathfrak{m}_j}]} \mathrm{d}\omega\, G(\omega) \right]^3.
\end{aligned}\end{equation}
Now, we bound the terms inside $\mathfrak F_2$. Similar to \eqref{Lemma:Concave:E3}, we bound
\begin{equation}\label{Lemma:Growth2:E15}
	\begin{aligned}
		& \left[-\Psi_t(\omega_{\text{Min}}) - \Psi_t(\omega_{\text{Mid}}) + \Psi_t(\omega_{\text{Max}}) + \Psi_t(\omega_{\text{Min}} + \omega_{\text{Mid}} - \omega_{\text{Max}})\right] \\
		=\, & \int_{0}^{\omega_{\text{Max}} - \omega_{\text{Min}}} \mathrm{d}\xi_1 \int_{0}^{\omega_{\text{Max}} - \omega_{\text{Mid}}} \mathrm{d}\xi_2 \, \partial_{\omega}^2\Psi_t(\xi_1 + \xi_2 + \omega_{\text{Min}}) \ 
		\ge\, 0,
	\end{aligned}
\end{equation}
since \(\partial_{\omega}^2\Psi_t \ge 0\).

Similar to \eqref{Lemma:Concave:E4}, we  estimate
\begin{equation}\label{Lemma:Growth2:E16}
	\begin{aligned}
		& \left[-\Psi_t(\omega_{\text{Max}}) - \Psi_t(\omega_{\text{Min}}) + \Psi_t(\omega_{\text{Mid}}) + \Psi_t(\omega_{\text{Max}} + \omega_{\text{Min}} - \omega_{\text{Mid}})\right]\Theta(\omega_{\text{Max}} + \omega_{\text{Min}} - \omega_{\text{Mid}}) \\
		& + \left[-\Psi_t(\omega_{\text{Max}}) - \Psi_t(\omega_{\text{Mid}}) + \Psi_t(\omega_{\text{Min}}) + \Psi_t(\omega_{\text{Max}} + \omega_{\text{Mid}} - \omega_{\text{Min}})\right]\Theta(\omega_{\text{Max}} + \omega_{\text{Mid}} - \omega_{\text{Min}}) \\
		=\, & \int_{0}^{\omega_{\text{Mid}} - \omega_{\text{Min}}} \mathrm{d}s \int_{0}^{\omega_{\text{Max}} - \omega_{\text{Mid}}} \mathrm{d}s_0\, \partial_{\omega}^2\Psi_t(\omega_{\text{Min}} + s + s_0) \\
		& \quad \times \left[\Theta(\omega_{\text{Max}} - \omega_{\text{Min}} + \omega_{\text{Mid}}) -\Theta(\omega_{\text{Max}} + \omega_{\text{Min}} - \omega_{\text{Mid}})\right] \\
		& + \int_{0}^{\omega_{\text{Mid}} - \omega_{\text{Min}}} \mathrm{d}s \int_{0}^{\omega_{\text{Mid}} - \omega_{\text{Min}}} \mathrm{d}s_0\,\Theta(\omega_{\text{Max}} - \omega_{\text{Min}} + \omega_{\text{Mid}})\, \partial_{\omega}^2\Psi_t(\omega_{\text{Min}} + s + s_0) \\
		\ge\, & \int_{0}^{\omega_{\text{Mid}} - \omega_{\text{Min}}} \mathrm{d}s \int_{0}^{\omega_{\text{Mid}} - \omega_{\text{Min}}} \mathrm{d}s_0\,\Theta(\omega_{\text{Max}} - \omega_{\text{Min}} + \omega_{\text{Mid}})\, \partial_{\omega}^2\Psi_t(\omega_{\text{Min}} + s + s_0) \\
		\ge\ & 0.
	\end{aligned}
\end{equation}

Using \eqref{Lemma:Growth2:E15} and \eqref{Lemma:Growth2:E16}, we can restrict the domain of integration of \( \mathfrak{F}_2 \) to   \( \omega_2 \le \omega_1 \le \omega \). We deduce from \eqref{Lemma:Growth2:E14} that
\begin{equation} \label{Lemma:Growth2:E17}
	\begin{aligned}
		\partial_t\left( \int_{\mathbb{R}^{+}} \mathrm{d}\omega\, G \Psi_t \right)
		\ge\ & \int_{\mathbb{R}^{+}} \mathrm{d}\omega\, G\, \partial_t \Psi_t 
		- C_0 R_{\mathfrak{m}_j}^{\theta} \Psi_t(0) \left[ \int_0^{R_{\mathfrak{m}_j}} \mathrm{d}\omega\, G(\omega) \right]^2
		- C_0 R_{\mathfrak{m}_j}^{\theta} \Psi_t(0) \left[ \int_0^{R_{\mathfrak{m}_j}} \mathrm{d}\omega\, G(\omega) \right]^3 \\
		& + c_{22} \iiint_{\mathbb{R}_+^3} \mathrm{d}\omega_1\, \mathrm{d}\omega_2\, \mathrm{d}\omega\,
		f_1 f_2 f\, \mathbf{1}_{\omega_2 \le \omega_1 \le \omega} 
		\left[ \frac{\max\{|\omega_1 - \omega_2|,\ |\omega_2 - \omega|\}}{2(\omega + \omega_1)} \right]^\mu \\
		& \times \Big\{ [-\Psi_t(\omega_2) - \Psi_t(\omega_1) + \Psi_t(\omega) + \Psi_t(\omega_2 + \omega_1 - \omega)] \\
		& \qquad \times\Theta(\omega)\Theta(\omega_2)\Theta(\omega_1)\Theta(\omega_2 + \omega_1 - \omega)\min\{|k_3|, |k_1|\} \mathbf{1}_{\omega_2 + \omega_1 - \omega\ge0}\\
		& + [-\Psi_t(\omega) - \Psi_t(\omega_2) + \Psi_t(\omega_1) + \Psi_t(\omega + \omega_2 - \omega_1)] \\
		& \qquad \times\Theta(\omega)\Theta(\omega_2)\Theta(\omega_1)\Theta(\omega + \omega_2 - \omega_1)\, |k_2| \\
		& + [-\Psi_t(\omega) - \Psi_t(\omega_1) + \Psi_t(\omega_2) + \Psi_t(\omega + \omega_1 - \omega_2)] \\
		& \qquad \times\Theta(\omega)\Theta(\omega_2)\Theta(\omega_1)\Theta(\omega + \omega_1 - \omega_2)\, |k_2| \Big\},
	\end{aligned}
\end{equation}
where \( |k_3| \) is associated to \( \omega_2 + \omega_1 - \omega \).

We can now bound \eqref{Lemma:Growth2:E17} from below as follows:
\begin{equation} \label{Lemma:Growth2:E18}
	\begin{aligned}
		\partial_t\left( \int_{\mathbb{R}^{+}} \mathrm{d}\omega\, G \Psi_t \right)
		\ge\ & \int_{\mathbb{R}^{+}} \mathrm{d}\omega\, G\, \partial_t \Psi_t 
		- C_0 R_{\mathfrak{m}_j}^{\theta} \Psi_t(0) \left[ \int_0^{R_{\mathfrak{m}_j}} \mathrm{d}\omega\, G(\omega) \right]^2
		- C_0 R_{\mathfrak{m}_j}^{\theta} \Psi_t(0) \left[ \int_0^{R_{\mathfrak{m}_j}} \mathrm{d}\omega\, G(\omega) \right]^3 \\
		& + c_{22} \iiint_{\mathbb{R}_+^3} \mathrm{d}\omega_1\, \mathrm{d}\omega_2\, \mathrm{d}\omega\,
		\frac{G(\omega)}{|k|\Theta(\omega)} \frac{G(\omega_2)}{|k_2|\Theta(\omega_2)} \frac{G(\omega_1)}{|k_1|\Theta(\omega_1)}\,
		\mathbf{1}_{\omega_2 \le \omega_1 \le \omega}  \\
		& \times \left[ \frac{\max\{|\omega_1 - \omega_2|,\ |\omega_2 - \omega|\}}{2(\omega + \omega_1)} \right]^\mu \Big\{ [-\Psi_t(\omega_2) - \Psi_t(\omega_1) + \Psi_t(\omega) + \Psi_t(\omega_2 + \omega_1 - \omega)] \\
		& \qquad \times\Theta(\omega)\Theta(\omega_2)\Theta(\omega_1)\Theta(\omega_2 + \omega_1 - \omega)\min\{|k_3|, |k_2|\}\mathbf{1}_{\omega_2 + \omega_1 - \omega\ge0} \\
		& + [-\Psi_t(\omega) - \Psi_t(\omega_2) + \Psi_t(\omega_1) + \Psi_t(\omega + \omega_2 - \omega_1)] \\
		& \qquad \times\Theta(\omega)\Theta(\omega_2)\Theta(\omega_1)\Theta(\omega + \omega_2 - \omega_1)\, |k_2| \\
		& + [-\Psi_t(\omega) - \Psi_t(\omega_1) + \Psi_t(\omega_2) + \Psi_t(\omega + \omega_1 - \omega_2)] \\
		& \qquad \times\Theta(\omega)\Theta(\omega_2)\Theta(\omega_1)\Theta(\omega + \omega_1 - \omega_2)\, |k_2| \Big\}.
	\end{aligned}
\end{equation}

We can simplify \eqref{Lemma:Growth2:E18} as
\begin{equation} \label{Lemma:Growth2:E19}
	\begin{aligned}
		\partial_t\left( \int_{\mathbb{R}^{+}} \mathrm{d}\omega\, G \Psi_t \right)
		\ge\ & \int_{\mathbb{R}^{+}} \mathrm{d}\omega\, G\, \partial_t \Psi_t 
		-  C_0 R_{\mathfrak{m}_j}^{\theta} \Psi_t(0)\left[ \int_0^{R_{\mathfrak{m}_j}} \mathrm{d}\omega\, G(\omega) \right]^2
		- C_0 R_{\mathfrak{m}_j}^{\theta} \Psi_t(0) \left[ \int_0^{R_{\mathfrak{m}_j}} \mathrm{d}\omega\, G(\omega) \right]^3 \\
		& + c_{22} \iiint_{\mathbb{R}_+^3} \mathrm{d}\omega_1\, \mathrm{d}\omega_2\, \mathrm{d}\omega\,
		\frac{G G_1 G_2}{|k||k_1|} \mathbf{1}_{\omega_2 \le \omega_1 \le \omega} \left[ \frac{\max\{|\omega_1 - \omega_2|,\ |\omega_2 - \omega|\}}{2(\omega + \omega_1)} \right]^\mu 
		\\
		& \times \Big\{
		[-\Psi_t(\omega_2) - \Psi_t(\omega_1) + \Psi_t(\omega) + \Psi_t(\omega_2 + \omega_1 - \omega)] \\
		& \quad \times\Theta(\omega_2 + \omega_1 - \omega)\, \frac{\min\{|k_3|, |k_1|\}}{|k_2|}\mathbf{1}_{\omega_2 + \omega_1 - \omega\ge0} \\
		& + [-\Psi_t(\omega) - \Psi_t(\omega_2) + \Psi_t(\omega_1) + \Psi_t(\omega + \omega_2 - \omega_1)] \,
		\Theta(\omega + \omega_2 - \omega_1) \\
		& + [-\Psi_t(\omega) - \Psi_t(\omega_1) + \Psi_t(\omega_2) + \Psi_t(\omega + \omega_1 - \omega_2)] \,
		\Theta(\omega + \omega_1 - \omega_2)
		\Big\},
	\end{aligned}
\end{equation}
where \( G = G(\omega), \ G_1 = G(\omega_1), \ G_2 = G(\omega_2) \).

Since $\omega(k) \ge C_{\text{disper}} |k|^{1/\delta}$, or equivalently 
\[
\frac{1}{|k|} \ge \frac{C_{\text{disper}}^\delta}{\omega(k)^\delta},
\]
we bound
\begin{equation} \label{Lemma:Growth2:E20}
	\begin{aligned}
		\partial_t\left( \int_{\mathbb{R}^{+}} \mathrm{d}\omega\, G \Psi_t \right)
		\ge\ & \int_{\mathbb{R}^{+}} \mathrm{d}\omega\, G\, \partial_t \Psi_t 
		- C_0 R_{\mathfrak{m}_j}^{\theta} \Psi_t(0) \left[ \int_0^{R_{\mathfrak{m}_j}} \mathrm{d}\omega\, G(\omega) \right]^2
		- C_0 R_{\mathfrak{m}_j}^{\theta} \Psi_t(0) \left[ \int_0^{R_{\mathfrak{m}_j}} \mathrm{d}\omega\, G(\omega) \right]^3 \\
		& + c_{22} C_{\text{disper}}^{2\delta} \iiint_{\mathbb{R}_+^3}
		\mathrm{d}\omega_1\, \mathrm{d}\omega_2\, \mathrm{d}\omega\,
		\frac{G G_1 G_2}{R_{\mathfrak{m}_j}^{2\delta}}\, 
		\left[ \frac{\max\{|\omega_1 - \omega_2|,\ |\omega_2 - \omega|\}}{2(\omega + \omega_1)} \right]^\mu \\
		& \quad \times \chi_{\left\{ \omega_2 \leq \frac{R_{\mathfrak{m}_j}}{20} \right\}}
		\chi_{\left\{ \omega_1 \le \omega; \ \omega_1, \omega \in \mathscr{O}_{\sigma(t)}^{h_{\mathfrak{m}_j}, R_{\mathfrak{m}_j}} \right\}} 
		\chi_{\mathscr{C}_{\mathfrak{m}_j}'^{\mathscr{T}_{\mathfrak{m}_j}}}(t) \\
		& \quad \times \Big\{ 
		[-\Psi_t(\omega) - \Psi_t(\omega_2) + \Psi_t(\omega_1) + \Psi_t(\omega + \omega_2 - \omega_1)]\,\Theta(\omega + \omega_2 - \omega_1) \\
		& \qquad + 
		[-\Psi_t(\omega) - \Psi_t(\omega_1) + \Psi_t(\omega_2) + \Psi_t(\omega + \omega_1 - \omega_2)]\,\Theta(\omega + \omega_1 - \omega_2)
		\Big\}.
	\end{aligned}
\end{equation}

The above inequality means that, when $t\notin \mathscr{C}_{\mathfrak{m}_j}'^{\mathscr{T}_{\mathfrak{m}_j}}$, we simply bound

\begin{equation*} 
	\begin{aligned}
		\partial_t\left( \int_{\mathbb{R}^{+}} \mathrm{d}\omega\, G \Psi_t \right)
		\ge\ & \int_{\mathbb{R}^{+}} \mathrm{d}\omega\, G\, \partial_t \Psi_t 
		- C_0 R_{\mathfrak{m}_j}^{\theta} \Psi_t(0) \left[ \int_0^{R_{\mathfrak{m}_j}} \mathrm{d}\omega\, G(\omega) \right]^2
		- C_0 R_{\mathfrak{m}_j}^{\theta} \Psi_t(0) \left[ \int_0^{R_{\mathfrak{m}_j}} \mathrm{d}\omega\, G(\omega) \right]^3.
	\end{aligned}
\end{equation*}

Since \(\omega_1, \omega \in \mathscr{O}_{\sigma(t)}^{h_{\mathfrak{m}_j}, R_{\mathfrak{m}_j}}\), it follows that \(\omega_1, \omega > \frac{R_{\mathfrak{m}_j}}{20}\). Consequently, we have \(\Psi_t(\omega) = \Psi_t(\omega_1) = \Psi_t(\omega + \omega_1 - \omega_2) = 0\). Substituting into \eqref{Lemma:Growth2:E20}, we obtain:
\begin{equation} \label{Lemma:Growth2:E21}
	\begin{aligned}
		\partial_t\left( \int_{\mathbb{R}^{+}} \mathrm{d}\omega\, G \Psi_t \right)
		\ge\ & \int_{\mathbb{R}^{+}} \mathrm{d}\omega\, G\, \partial_t \Psi_t 
		- C_0 R_{\mathfrak{m}_j}^{\theta} \Psi_t(0) \left[ \int_0^{R_{\mathfrak{m}_j}} \mathrm{d}\omega\, G(\omega) \right]^2
		- C_0 R_{\mathfrak{m}_j}^{\theta} \Psi_t(0) \left[ \int_0^{R_{\mathfrak{m}_j}} \mathrm{d}\omega\, G(\omega) \right]^3 \\
		& + c_{22} C_{\text{disper}}^{2\delta} \iiint_{\mathbb{R}_+^3}
		\mathrm{d}\omega_1\, \mathrm{d}\omega_2\, \mathrm{d}\omega\,
		\frac{G G_1 G_2}{R_{\mathfrak{m}_j}^{2\delta}}\, 
		\left[ \frac{\max\{|\omega_1 - \omega_2|,\ |\omega_2 - \omega|\}}{2(\omega + \omega_1)} \right]^\mu \\
		& \quad \times \chi_{\left\{ \omega_2 \leq \frac{R_{\mathfrak{m}_j}}{20} \right\}}
		\chi_{\left\{ \omega_1 \le \omega; \ \omega_1, \omega \in \mathscr{O}_{\sigma(t)}^{h_{\mathfrak{m}_j}, R_{\mathfrak{m}_j}} \right\}} 
		\chi_{\mathscr{C}_{\mathfrak{m}_j}'^{\mathscr{T}_{\mathfrak{m}_j}}}(t) \\
		& \quad \times \Big\{ 
		[-\Psi_t(\omega_2) + \Psi_t(\omega + \omega_2 - \omega_1)]\, \Theta(\omega + \omega_2 - \omega_1) 
		+ \Psi_t(\omega_2)\,\Theta(\omega + \omega_1 - \omega_2)
		\Big\}.
	\end{aligned}
\end{equation}

Next, since  
\[
\omega + \omega_1 - \omega_2 \ge \frac{R_{\mathfrak{m}_j}}{2} \quad \text{and} \quad 0 \le \omega + \omega_2 - \omega_1 \le \omega_2 + 3 h_{\mathfrak{m}_j} \le \frac{R_{\mathfrak{m}_j}}{20} + 3 h_{\mathfrak{m}_j} \le \frac{R_{\mathfrak{m}_j}}{2},
\]  
we estimate
\begin{equation*}
	\begin{aligned}
		& \left[ -\Psi_t(\omega_2) + \Psi_t(\omega + \omega_2 - \omega_1) \right] \Theta(\omega + \omega_2 - \omega_1) 
		+ \Psi_t(\omega_2) \Theta(\omega + \omega_1 - \omega_2) \\
		\ge\ & \left[ -\Psi_t(\omega_2) + \Psi_t(\omega + \omega_2 - \omega_1) \right] \Theta\left( \frac{R_{\mathfrak{m}_j}}{2} \right) 
		+ \Psi_t(\omega_2) \Theta\left( \frac{R_{\mathfrak{m}_j}}{2} \right),
	\end{aligned}
\end{equation*}
and  \begin{equation}\label{muestimate}
	1\ \ge \ \left( \frac{\max\{|\omega_1 - \omega|, |\omega_2 - \omega|\}}{2(\omega + \omega_1)} \right)^\mu \ \ge  \ \left( \frac{\frac{R_{\mathfrak{m}_j}}{4}}{4R_{\mathfrak{m}_j}}  \right)^\mu  \  = \ 2^{-4\mu} \end{equation}
which, in combination with \eqref{Lemma:Growth2:E21}, yields
\begin{equation} \label{Lemma:Growth2:E22}
	\begin{aligned}
		\partial_t\left( \int_{\mathbb{R}^{+}} \mathrm{d}\omega\, G \Psi_t \right)
		\ge\ & \int_{\mathbb{R}^{+}} \mathrm{d}\omega\, G\, \partial_t \Psi_t 
		- C_0 R_{\mathfrak{m}_j}^{\theta} \Psi_t(0) \left( \int_0^{R_{\mathfrak{m}_j}} \mathrm{d}\omega\, G(\omega) \right)^2
		- C_0 R_{\mathfrak{m}_j}^{\theta} \Psi_t(0) \left( \int_0^{R_{\mathfrak{m}_j}} \mathrm{d}\omega\, G(\omega) \right)^3 \\
		& + c_{22} C_{\text{disper}}^{2\delta} \iiint_{\mathbb{R}_+^3}
		\mathrm{d}\omega_1\, \mathrm{d}\omega_2\, \mathrm{d}\omega\,
		\frac{G G_1 G_2}{R_{\mathfrak{m}_j}^{2\delta}} \left( \frac{\max\{|\omega_1 - \omega_2|,\ |\omega_2 - \omega|\}}{2(\omega + \omega_1)} \right)^\mu
		\\
		& \quad \times \chi_{\left\{ \omega_2 \leq \frac{R_{\mathfrak{m}_j}}{20} \right\}}
		\chi_{\left\{ \omega_1 \le \omega;\ \omega_1, \omega \in \mathscr{O}_{\sigma(t)}^{h_{\mathfrak{m}_j}, R_{\mathfrak{m}_j}} \right\}} 
		\chi_{\mathscr{C}_{\mathfrak{m}_j}'^{\mathscr{T}_{\mathfrak{m}_j}}}(t) \\
		& \quad \times \left\{ \left[ -\Psi_t(\omega_2) + \Psi_t(\omega + \omega_2 - \omega_1) \right] \Theta\left( \frac{R_{\mathfrak{m}_j}}{2} \right) 
		 \right\}\\
		& + c_{22} C_{\text{disper}}^{2\delta} \iiint_{\mathbb{R}_+^3}
		\mathrm{d}\omega_1\, \mathrm{d}\omega_2\, \mathrm{d}\omega\,
		\frac{G G_1 G_2}{R_{\mathfrak{m}_j}^{2\delta}} \left( \frac{\max\{|\omega_1 - \omega_2|,\ |\omega_2 - \omega|\}}{2(\omega + \omega_1)} \right)^\mu
		\\
		& \quad \times \chi_{\left\{ \omega_2 \leq \frac{R_{\mathfrak{m}_j}}{20} \right\}}
		\chi_{\left\{ \omega_1 \le \omega;\ \omega_1, \omega \in \mathscr{O}_{\sigma(t)}^{h_{\mathfrak{m}_j}, R_{\mathfrak{m}_j}} \right\}} 
		\chi_{\mathscr{C}_{\mathfrak{m}_j}'^{\mathscr{T}_{\mathfrak{m}_j}}}(t)  \Psi_t(\omega_2) \Theta\left( \frac{R_{\mathfrak{m}_j}}{2} \right).
	\end{aligned}
\end{equation}

From \eqref{Lemma:Growth2:E22}, we deduce
\begin{equation} \label{Lemma:Growth2:E23}
	\begin{aligned}
		\partial_t\left( \int_{\mathbb{R}^{+}} \mathrm{d}\omega\, G \Psi_t \right)
		\ge\ & \int_{\mathbb{R}^{+}} \mathrm{d}\omega\, G\, \partial_t \Psi_t 
		- C_0 R_{\mathfrak{m}_j}^{\theta} \Psi_t(0) \left( \int_0^{R_{\mathfrak{m}_j}} \mathrm{d}\omega\, G(\omega) \right)^2
		- C_0 R_{\mathfrak{m}_j}^{\theta} \Psi_t(0) \left( \int_0^{R_{\mathfrak{m}_j}} \mathrm{d}\omega\, G(\omega) \right)^3 \\
		& + c_{22} C_{\text{disper}}^{2\delta} \iiint_{\mathbb{R}_+^3} \mathrm{d}\omega_1\, \mathrm{d}\omega_2\, \mathrm{d}\omega\,
		\frac{G G_1 G_2 \Theta\left( \frac{R_{\mathfrak{m}_j}}{2} \right)}{R_{\mathfrak{m}_j}^{2\delta}}\chi_{\left\{ \omega_2 \le \frac{R_{\mathfrak{m}_j}}{20} \right\}} \\
		& \quad \times 
		\chi_{\left\{ \omega_1 \le \omega;\ \omega_1, \omega \in \mathscr{O}_{\sigma(t)}^{h_{\mathfrak{m}_j}, R_{\mathfrak{m}_j}} \right\}}
		\chi_{\mathscr{C}_{\mathfrak{m}_j}'^{\mathscr{T}_{\mathfrak{m}_j}}}(t) 	 \left[ -\Psi_t(\omega_2) + \Psi_t(\omega + \omega_2 - \omega_1) \right] \\
		& + 2^{-4\mu}c_{22} C_{\text{disper}}^{2\delta} \iiint_{\mathbb{R}_+^3} \mathrm{d}\omega_1\, \mathrm{d}\omega_2\, \mathrm{d}\omega\,
		\frac{G G_1 G_2}{R_{\mathfrak{m}_j}^{2\delta}}\chi_{\left\{ \omega_2 \le \frac{R_{\mathfrak{m}_j}}{20} \right\}} 	 \\
		& \quad \times 
		\chi_{\left\{ \omega_1 \le \omega;\ \omega_1, \omega \in \mathscr{O}_{\sigma(t)}^{h_{\mathfrak{m}_j}, R_{\mathfrak{m}_j}} \right\}} 
		\chi_{\mathscr{C}_{\mathfrak{m}_j}'^{\mathscr{T}_{\mathfrak{m}_j}}}(t) 
		\Psi_t(\omega_2) \Theta\left( \frac{R_{\mathfrak{m}_j}}{2} \right).
	\end{aligned}
\end{equation}

Using \eqref{Lemma:Supersolu:E4a}, we obtain from \eqref{Lemma:Growth2:E23}:
\begin{equation*}
	\begin{aligned}
		\partial_t\left( \int_{\mathbb{R}^{+}} \mathrm{d}\omega\, G \Psi_t \right)
		\ge\ &
		- C_0 R_{\mathfrak{m}_j}^{\theta} \Psi_t(0) \left( \int_0^{R_{\mathfrak{m}_j}} \mathrm{d}\omega\, G(\omega) \right)^2
		- C_0 R_{\mathfrak{m}_j}^{\theta} \Psi_t(0) \left( \int_0^{R_{\mathfrak{m}_j}} \mathrm{d}\omega\, G(\omega) \right)^3 \\
		&+ 2^{-4\mu} c_{22} C_{\text{disper}}^{2\delta}
		\iiint_{\mathbb{R}_+^3} \mathrm{d}\omega\, \mathrm{d}\omega_1\, \mathrm{d}\omega_2\,
		\frac{G(\omega) G(\omega_1) G(\omega_2)}{R_{\mathfrak{m}_j}^{2\delta}} \\
		&\qquad \times
		\chi_{\left\{ \omega_2 \le \frac{R_{\mathfrak{m}_j}}{20} \right\}}
		\chi_{\left\{ \omega_1 \le \omega;\ \omega_1, \omega \in \mathscr{O}_{\sigma(t)}^{h_{\mathfrak{m}_j}, R_{\mathfrak{m}_j}} \right\}}
		\chi_{\mathscr{C}_{\mathfrak{m}_j}'^{\mathscr{T}_{\mathfrak{m}_j}}}(t)
		\Psi_t(\omega_2) \Theta\left( \frac{R_{\mathfrak{m}_j}}{2} \right),
	\end{aligned}
\end{equation*}
for a universal constant ${C}_0 > 0$ that may vary between lines. Therefore,
\begin{equation}
	\label{Lemma:Growth2:E24}
	\begin{aligned}
		\partial_t\left( \int_{\mathbb{R}^{+}} \mathrm{d}\omega\, G \Psi_t \right)
		\ge\ &
		- C_0 R_{\mathfrak{m}_j}^{\theta} \Psi_t(0) \mathscr M^2
		- C_0 R_{\mathfrak{m}_j}^{\theta} \Psi_t(0) \mathscr M^3 \\
		&+ 2^{-4\mu} c_{22} C_{\text{disper}}^{2\delta}R_{\mathfrak{m}_j}^{-2\delta}
		\left[ \int_{\mathbb{R}_+} \mathrm{d}\omega\, G(\omega) \Psi_t(\omega) \right]
		\left[ \int_{\mathbb{R}_+} \mathrm{d}\omega\, G(\omega)
		\chi_{\left\{ \omega \in \mathscr{O}_{\sigma(t)}^{h_{\mathfrak{m}_j}, R_{\mathfrak{m}_j}} \right\}}
		\chi_{\mathscr{C}_{\mathfrak{m}_j}'^{\mathscr{T}_{\mathfrak{m}_j}}}(t) \right]^2.
	\end{aligned}
\end{equation}

Solving~\eqref{Lemma:Growth2:E24}, we obtain
\begin{equation*}
	\begin{aligned}
		&		\int_{\mathbb{R}^{+}} \mathrm{d}\omega\, G(t,\omega) \Psi_t\\ 
		\ge\ & \int_{\mathbb{R}^{+}} \mathrm{d}\omega\, G(0,\omega) \Psi_0 \exp\left( \int_0^t \mathrm{d}s\, 2^{-4\mu} c_{22} C_{\text{disper}}^{2\delta}R_{\mathfrak{m}_j}^{-2\delta}
		\left[ \int_{\mathbb{R}_+} \mathrm{d}\omega\, G(s,\omega)
		\chi_{\left\{ \omega \in \mathscr{O}_{\sigma(s)}^{h_{\mathfrak{m}_j}, R_{\mathfrak{m}_j}} \right\}}
		\chi_{\mathscr{C}_{\mathfrak{m}_j}'^{\mathscr{T}_{\mathfrak{m}_j}}}(s) \right]^2 \right) \\
		& - \left[ C_0 R_{\mathfrak{m}_j}^{\theta} \Psi_t(0) \mathscr{M}^2
		+ C_0 R_{\mathfrak{m}_j}^{\theta} \Psi_t(0) \mathscr{M}^3 \right] \\
		& \times \left[ \int_0^t \mathrm{d}s\, \exp\left( \int_s^t \mathrm{d}s'\, 2^{-4\mu} c_{22} C_{\text{disper}}^{2\delta}R_{\mathfrak{m}_j}^{-2\delta}
		\left[ \int_{\mathbb{R}_+} \mathrm{d}\omega\, G(s',\omega)
		\chi_{\left\{ \omega \in \mathscr{O}_{\sigma(s')}^{h_{\mathfrak{m}_j}, R_{\mathfrak{m}_j}} \right\}}
		\chi_{\mathscr{C}_{\mathfrak{m}_j}'^{\mathscr{T}_{\mathfrak{m}_j}}}(s') \right]^2 \right) \right],
	\end{aligned}
\end{equation*}
which immediately leads to
\begin{equation}
	\label{Lemma:Growth2:E25}
	\begin{aligned}
		&	\int_{\mathbb{R}^{+}} \mathrm{d}\omega\, G(t,\omega) \Psi_t \\
		\ge\ & \left\{ \int_{\mathbb{R}^{+}} \mathrm{d}\omega\, G(0,\omega) \Psi_0 
		- T\left[ C_0 R_{\mathfrak{m}_j}^{\theta} \Psi_t(0) \mathscr{M}^2
		+ C_0 R_{\mathfrak{m}_j}^{\theta} \Psi_t(0) \mathscr{M}^3 \right] \right\} \\
		& \times \exp\left( \int_0^t \mathrm{d}s\, 2^{-4\mu} c_{22} C_{\text{disper}}^{2\delta}R_{\mathfrak{m}_j}^{-2\delta}
		\left[ \int_{\mathbb{R}_+} \mathrm{d}\omega\, G(s,\omega)
		\chi_{\left\{ \omega \in \mathscr{O}_{\sigma(s)}^{h_{\mathfrak{m}_j}, R_{\mathfrak{m}_j}} \right\}}
		\chi_{\mathscr{C}_{\mathfrak{m}_j}'^{\mathscr{T}_{\mathfrak{m}_j}}}(s) \right]^2 \right).
	\end{aligned}
\end{equation}

Using~\eqref{Lemma:Supersolu:E8} and \eqref{Theorem1:4}, we obtain the bound
\begin{equation}
	\label{Lemma:Growth2:E26}
	\begin{aligned}
		\int_{\mathbb{R}_+} \mathrm{d}\omega\, \Psi_0(\omega)\, G(0,\omega)   
		\geq 0.02\, C_{\text{ini}} \left( \frac{R_{\mathfrak{m}_j}}{50} \right)^{c_{\text{ini}}}.
	\end{aligned}
\end{equation}

Substituting~\eqref{Lemma:Growth2:E26} into~\eqref{Lemma:Growth2:E25}, we obtain
\begin{equation}
	\label{Lemma:Growth2:E27}
	\begin{aligned}
		\int_{\mathbb{R}^{+}} \mathrm{d}\omega\, G(t,\omega)\, \Psi_t 
		\ge\ & \left\{ 0.02\, C_{\text{ini}} \left(  \frac{R_{\mathfrak{m}_j}}{50} \right)^{c_{\text{ini}}}
		- T\left[ C_0 R_{\mathfrak{m}_j}^{\theta} \Psi_t(0)\, \mathscr{M}^2
		+ C_0 R_{\mathfrak{m}_j}^{\theta} \Psi_t(0)\, \mathscr{M}^3 \right] \right\} \\
		& \times \exp\left( \int_0^t \mathrm{d}s\, 2^{-4\mu} c_{22} C_{\text{disper}}^{2\delta}R_{\mathfrak{m}_j}^{-2\delta}
		\left[ \int_{\mathbb{R}_+} \mathrm{d}\omega\, G(s,\omega)\,
		\chi_{\left\{ \omega \in \mathscr{O}_{\sigma(s)}^{h_{\mathfrak{m}_j}, R_{\mathfrak{m}_j}} \right\}}\,
		\chi_{\mathscr{C}_{\mathfrak{m}_j}'^{\mathscr{T}_{\mathfrak{m}_j}}}(s) \right]^2 \right).
	\end{aligned}
\end{equation}

Combining~\eqref{Lemma:Growth2:E10b} and~\eqref{Lemma:Growth2:E27}, we obtain 
\begin{equation}
	\label{Lemma:Growth2:E28}
	\begin{aligned}
		&	\int_{\mathbb{R}^{+}} \mathrm{d}\omega\, G(\mathscr{T}_{\mathfrak{m}_j}, \omega)\, \Psi_{\mathscr{T}_{\mathfrak{m}_j}} \\
		\ge\ & \left\{ 0.02\, C_{\text{ini}} \left( \frac{R_{\mathfrak{m}_j}}{50} \right)^{c_{\text{ini}}}
		- T\left[ C_0 R_{\mathfrak{m}_j}^{\theta} \Psi_{\mathscr{T}_{\mathfrak{m}_j}}(0)\, \mathscr{M}^2
		+ C_0 R_{\mathfrak{m}_j}^{\theta} \Psi_{\mathscr{T}_{\mathfrak{m}_j}}(0)\, \mathscr{M}^3 \right] \right\} \\
		& \times \exp\left( \int_0^{\mathscr{T}_{\mathfrak{m}_j}} \mathrm{d}s\, 2^{-4\mu} c_{22} C_{\text{disper}}^{2\delta} R_{\mathfrak{m}_j}^{-2\delta}
		\left[ \int_{\mathbb{R}_+} \mathrm{d}\omega\, G(s,\omega)\,
		\chi_{\left\{ \omega \in \mathscr{O}_{\sigma(s)}^{h_{\mathfrak{m}_j}, R_{\mathfrak{m}_j}} \right\}}\,
		\chi_{\mathscr{C}_{\mathfrak{m}_j}'^{\mathscr{T}_{\mathfrak{m}_j}}}(s) \right]^2 \right) \\
		\ge\ & 0.01\, C_{\text{ini}} \left( \frac{R_{\mathfrak{m}_j}}{50} \right)^{c_{\text{ini}}}
		\exp\left( 2^{-4\mu} c_{22} C_{\text{disper}}^{2\delta}
		\frac{\mathfrak{C}_1 R_{\mathfrak{m}_j}^{2\rho - \varepsilon}}{\Theta(R_{\mathfrak{m}_j}/2)} \right),
	\end{aligned}
\end{equation}
where we have used the fact that $c_{\text{ini}} < \theta$  from the conditions of Theorem \ref{Theorem1}  and assumed that $j$ is sufficiently large. 	 We further estimate~\eqref{Lemma:Growth2:E28} using Assumption A as
\begin{equation}
	\label{Lemma:Growth2:E29}
	\begin{aligned}
		\int_{\mathbb{R}^{+}} \mathrm{d}\omega\, G(\mathscr{T}_{\mathfrak{m}_j}, \omega)\, \Psi_{\mathscr{T}_{\mathfrak{m}_j}} 
		\ge\ & 0.01\, C_{\text{ini}} \left( \frac{R_{\mathfrak{m}_j}}{50} \right)^{c_{\text{ini}}} \ 
		\exp\left( 2^{-4\mu} c_{22} C_{\text{disper}}^{2\delta}
		\frac{\mathfrak{C}_1 R_{\mathfrak{m}_j}^{2\rho - \varepsilon }}{C_\Theta(R_{\mathfrak{m}_j}/2)^\varrho} \right).
	\end{aligned}
\end{equation}

This yields
\begin{equation}
	\label{Lemma:Growth2:E30}
	\begin{aligned}
		\mathscr{M} + \mathscr{E}  
		\ge\ & 0.01\, C_{\text{ini}} \left( \frac{R_{\mathfrak{m}_j}}{50} \right)^{c_{\text{ini}}} \\
		& \times \exp\left( 2^{-4\mu} c_{22} C_{\text{disper}}^{2\delta}
		\frac{\mathfrak{C}_1 R_{\mathfrak{m}_j}^{2\rho - \varrho - \varepsilon}}{C_\Theta(1/2)^\varrho} \right)
		\to \infty \quad \text{as } j \to \infty,
	\end{aligned}
\end{equation}
since $\max\{0,2\rho-\varrho\}  <\varepsilon$ by \eqref{rho}.

We obtain a contradiction, and hence the assumption leading to~\eqref{Lemma:Growth2:E0a} cannot hold true. Therefore, by the computations in~\eqref{Lemma:Growth2:E0:1}-\eqref{Lemma:Growth2:E0:3}, we conclude~\eqref{Lemma:Growth2:1}.

\end{proof}

\section{The second multiscale estimates}\label{Sec:Second}

\subsection{The main second estimate}\label{Sec:DDM}In this subsection, we provide an estimate that will serve as the basis for the second multiscale estimate.

%
%
%

We now define
\begin{equation}\label{Sec:DDM:2}
	\begin{aligned}
		\Omega_{i,j,l}^{h_{\mathfrak m},R_{\mathfrak m}} = \Omega_{i}^{h_{\mathfrak m},R_{\mathfrak m}} \times \Omega_{j}^{h_{\mathfrak m},R_{\mathfrak m}} \times \Omega_{l}^{h_{\mathfrak m},R_{\mathfrak m}}, \quad
		i, j, l = 0, \dots, \mathscr{M}_{h_{\mathfrak m},R_{\mathfrak m}} - 1,
	\end{aligned}
\end{equation}
as well as the set of indices
\begin{equation}\label{Sec:DDM:3}
	\begin{aligned}
		I_{i}^{h_{\mathfrak m},R_{\mathfrak m}} &= \{ i-1, i, i+1 \}, \quad && i = 1, \dots, \mathscr{M}_{h_{\mathfrak m},R_{\mathfrak m}} - 2, \\
		I_{0}^{h_{\mathfrak m},R_{\mathfrak m}} &= \{ 0, 1 \}, \quad
		&& I_{\mathscr{M}_{h,R_{\mathfrak m}}-1}^{h_{\mathfrak m},R_{\mathfrak m}} = \{ \mathscr{M}_{h_{\mathfrak m},R_{\mathfrak m}} - 2, \mathscr{M}_{h_{\mathfrak m},R_{\mathfrak m}} - 1 \}.
	\end{aligned}
\end{equation}

We also define the sets
\begin{equation}\label{Sec:DDM:4}
	\mathcal{E}_{R_{\mathfrak m},h_{\mathfrak m}} = \left\{ (\omega, \omega_{1}, \omega_{2}) \in [0, R_{\mathfrak m})^3 : 
	|\omega_{\mathrm{Mid}} - \omega_{\mathrm{Min}}| \ge 2h_{\mathfrak m} \right\},
\end{equation}
and
\begin{equation}\label{Sec:DDM:5}
	\mathcal{E}_{R_{\mathfrak m},h_{\mathfrak m}}' = \left\{ (\omega, \omega_{1}, \omega_{2}) \in [0, R_{\mathfrak m})^3 : 
	|\omega_{\mathrm{Mid}} - \omega_{\mathrm{Min}}| \ge h_{\mathfrak m} \right\},
\end{equation}
where we have used the notations of \eqref{Sec:DDM:6}.

For a constant $\nu > 0$, we recall \eqref{FDefinition} , \eqref{NoCondensateTime} and define
\begin{equation}\label{TimePump}
	\begin{aligned}
		\Xi_{\nu,R_{\mathfrak m}}^2 := \Big\{ t \in \Xi \ \Big| \ \forall i \in \{0, \dots, \mathscr{M}_{h_{\mathfrak m},R_{\mathfrak m}} - 1\},\ 
		\int_{\mathscr{O}_{i}^{h_{\mathfrak m},R_{\mathfrak m}}} \mathrm{d}\omega\, G(t) 
		< (1 - \nu) \int_{[0, R_{\mathfrak m})} \mathrm{d}\omega\, G(t) \Big\}, \\
		\Xi_{\nu,R_{\mathfrak m}}^1 := \Xi \setminus \Xi_{\nu,R_{\mathfrak m}}^2.
	\end{aligned}
\end{equation}

\begin{proposition}
	\label{Propo:Collision} 	We assume Assumption A and Assumption B. 

	There exists a universal constant \( \mathcal{C}_1 > 0 \) such that the following estimate holds true for \( 0 < h_{\mathfrak m} < \frac{1}{10} \) and \( 0 < \nu < \frac{1}{3} \):
	\begin{equation}\label{Propo:Collision:1}
		\begin{aligned}
			\frac{\mathcal{C}_1 (\mathscr{M} + \mathscr{E})}{\Theta(h_{\mathfrak m}) \nu^4} \cdot \frac{R_{\mathfrak m}^{2\delta + \frac{3}{2} + \mu}}{h_{\mathfrak m}^{2 + \mu}} \ge 
			& \int_{\Xi_{\nu,R_{\mathfrak m}}^2} \mathrm{d}t \left( \int_{[0,R_{\mathfrak m})} \mathrm{d}\omega\, G(t) \right)^3.
		\end{aligned}
	\end{equation}
\end{proposition}

\begin{proof}
	
	We divide the proof into smaller steps.
	
	{\it Step 1: Subdomain inclusions.} 
	
	Setting
	\begin{equation}\label{Propo:Collision:2}
		\mathcal{E}''_{R_{\mathfrak m},h_{\mathfrak m}} =
		\left[
		\bigcup_{\substack{
				\mathscr{M}_{R_{\mathfrak m},h_{\mathfrak m}} > \max\{i, j, l\} \\
				\ge \mathrm{mid}\{i, j, l\} > \min\{i, j, l\} + 1
		}}
		\Omega_{i,j,l}^{h_{\mathfrak m},R_{\mathfrak m}}
		\right],
	\end{equation}
	we will first show that
	\begin{equation}\label{Propo:Collision:3}
		\mathcal{E}_{R_{\mathfrak m},h_{\mathfrak m}} \subset 
		\mathcal{E}''_{R_{\mathfrak m},h_{\mathfrak m}}
		\subset \mathcal{E}_{R_{\mathfrak m},h_{\mathfrak m}}'.
	\end{equation}
	
	To this end, we consider a point 
	$\left( \omega_2, \omega_{1}, \omega \right) \in \mathcal{E}_{R_{\mathfrak m},h_{\mathfrak m}}$. 
	Without loss of generality, we may assume that 
	\[
	\omega = \omega_{\mathrm{Min}} \left( \omega_2, \omega_{1}, \omega \right)
	< \omega_{1} = \omega_{\mathrm{Mid}} \left( \omega_2, \omega_{1}, \omega \right)
	\le \omega_{2} = \omega_{\mathrm{Max}}\left( \omega_2, \omega_{1}, \omega \right).
	\]
	From the definition of $\mathcal{E}_{R_{\mathfrak m},h_{\mathfrak m}}$, there exist indices $i$, $j$, and $l$ such that
	\[
	\omega_{2} \in \Omega_{i}^{h_{\mathfrak m},R_{\mathfrak m}}, \quad
	\omega_{1} \in \Omega_{j}^{h_{\mathfrak m},R_{\mathfrak m}}, \quad
	\omega \in \Omega_{l}^{h_{\mathfrak m},R_{\mathfrak m}}, \quad
	\text{with } l \leq j \leq i.
	\]
	Since 
	\[
	|\omega_{\mathrm{Mid}} - \omega_{\mathrm{Min}}| \ge 2 h_{\mathfrak m},  
	\]
	we conclude that  \( j > l + 1 \). This leads to
	$$
	\mathcal{E}_{R_{\mathfrak m},h_{\mathfrak m}} \subset 
	\mathcal{E}''_{R_{\mathfrak m},h_{\mathfrak m}}.
	$$

	Next, we choose $\left(\omega_2, \omega_1, \omega\right) \in \mathcal{E}''_{R_{\mathfrak m},h_{\mathfrak m}}$. Suppose  
	$\omega_2 \in \Omega_i^{h_{\mathfrak m},R_{\mathfrak m}},\ \omega_1 \in \Omega_j^{h_{\mathfrak m},R_{\mathfrak m}},\ \omega \in \Omega_l^{h_{\mathfrak m},R_{\mathfrak m}}$,  
	with $\mathscr{M}_{R_{\mathfrak m},h_{\mathfrak m}} > i \ge j   > l + 1$. Therefore,  
	$\left| \omega_{\text{Mid}} - \omega_{\text{Min}} \right| \ge h_{\mathfrak m}$. 
	We then have  
	\[
	\mathcal{E}''_{R_{\mathfrak m},h_{\mathfrak m}} \subset \mathcal{E}'_{R_{\mathfrak m},h_{\mathfrak m}},
	\]  
	which completes the proof of \eqref{Propo:Collision:3}.

	{\it Step 2: Constructing the  special  subdomains.}

	Let \(t\) be in \(\Xi_{\nu,R_{\mathfrak m}}^2\) and \( i_t^o \)  be an index in \( \{0, \dots, \mathscr{M}_{R_{\mathfrak m},h_{\mathfrak m}} - 1\} \) such that
	\begin{equation}\label{Propo:Collision:6}
		\int_{\Omega_{i_t^o}^{h_{\mathfrak m},R_{\mathfrak m}}} \mathrm{d}\omega\, G(t) = \max_{i \in \{0, \dots, \mathscr{M}_{R_{\mathfrak m},h_{\mathfrak m}} - 1\}} \left\{ \int_{\Omega_i^{h_{\mathfrak m},R_{\mathfrak m}}} \mathrm{d}\omega\, G(t) \right\}.
	\end{equation}

We consider two cases.

\textit{Case (I):} In this case, we assume that
\begin{equation}\label{Propo:Collision:18}
	\int_{\Omega_{i_t^o}^{h_{\mathfrak m}, R_{\mathfrak m}}} \mathrm{d}\omega\, G(t) 
	< \frac{\nu}{1000} \int_{[0, R_{\mathfrak m})} \mathrm{d}\omega\, G(t).
\end{equation}

We split the treatment of this case into two steps.

\textit{Case (I) - Step (i)}: Creating the first subdomain.

	Let \( \mathfrak{X}_t \) be the set of partitions of \( \{0, \dots, \mathscr{M}_{R_{\mathfrak m},h_{\mathfrak m}} - 1\} \) satisfying the following conditions:
	\begin{itemize}
		\item[(A)] If \( \mathcal{I}_t \in \mathfrak{X}_t \), then \( i_t^o \in \mathcal{I}_t \).
		
		\item[(B)] Let \( \mathcal{I}_t = \{i_1, \dots, i_n\} \in \mathfrak{X}_t \). Then, for all \( i_j \in \mathcal{I}_t \),
		\begin{equation}\label{Propo:Collision:6a}
			\Omega_{i_j}^{h_{\mathfrak m},R_{\mathfrak m}} \cap \bigcup_{i \in \mathcal{I}_t \setminus \{i_j\}} \mathscr{O}_i^{h_{\mathfrak m},R_{\mathfrak m}} = \emptyset.
		\end{equation}
		
		\item[(C)] There exists a constant \( 10>c_* \geq 3 \), independent of \( t \), such that for all \( \mathcal{I}_t \in \mathfrak{X}_t \),
		\begin{equation}\label{Propo:Collision:7}
			c_* \int_{\cup_{i \in \mathcal{I}_t} \Omega_i^{h_{\mathfrak m},R_{\mathfrak m}}} \mathrm{d}\omega\, G(t) 
			\geq 
			\int_{\cup_{i \in \mathcal{I}_t} \mathscr{O}_i^{h_{\mathfrak m},R_{\mathfrak m}}} \mathrm{d}\omega\, G(t).
		\end{equation}
		
		\item[(D)] For all \( \mathcal{I}_t \in \mathfrak{X}_t \),
		\begin{equation}\label{Propo:Collision:8}
			\int_{\cup_{i \in \mathcal{I}_t}\mathscr{O}_i^{h_{\mathfrak m},R_{\mathfrak m}}} \mathrm{d}\omega\, G(t) 
			< 
	(1-\nu) \int_{[0, R_{\mathfrak m})} \mathrm{d}\omega\, G(t).
		\end{equation}
	\end{itemize}
	
	It follows from \eqref{Propo:Collision:6} that
	\begin{equation}\label{Propo:Collision:9}
		3 \int_{\Omega_{i_t^o}^{h_{\mathfrak m},R_{\mathfrak m}}} \mathrm{d}\omega\, G(t) 
		\ge 
		\int_{\mathscr{O}_{i_t^o}^{h_{\mathfrak m},R_{\mathfrak m}}} \mathrm{d}\omega\, G(t),
	\end{equation}
	and thus the set of partitions \( \mathfrak{X}_t \) is non-empty, since at least \( \{i_t^o\} \in \mathfrak{X}_t \).
	
	Next, we define \( \mathscr{Z}_t = \{i_1, \dots, i_n\} \in \mathfrak{X}_t \) to be the partition satisfying the following conditions:
	\begin{itemize}
		\item The first index is defined as \( i_1 = i_t^o \).
		
		\item For \( j = 2, \dots, n \),
		\begin{equation}\label{Propo:Collision:10:1}
			\int_{\Omega_{i_j}^{h_{\mathfrak m},R_{\mathfrak m}}} \mathrm{d}\omega\, G(t)
			= 
			\max_{i \in \{0, \dots, \mathscr{M}_{R_{\mathfrak m},h_{\mathfrak m}} - 1\} \setminus \left( I_{i_1}^{h_{\mathfrak m},R_{\mathfrak m}} \cup \cdots \cup I_{i_{j-1}}^{h_{\mathfrak m},R_{\mathfrak m}} \right)} 
			\left\{ \int_{\Omega_i^{h_{\mathfrak m},R_{\mathfrak m}}} \mathrm{d}\omega\, G(t) \right\},
		\end{equation}
		where we have used the notation introduced in \eqref{Sec:DDM:3}.
		
		\item In addition, for all \( j \in \{0, \dots, \mathscr{M}_{R_{\mathfrak m},h_{\mathfrak m}} - 1\} \setminus \left( I_{i_1}^{h_{\mathfrak m},R_{\mathfrak m}} \cup \cdots \cup I_{i_n}^{h_{\mathfrak m},R_{\mathfrak m}} \right) \), we have
		\begin{equation}\label{Propo:Collision:10}
			\int_{\bigcup_{i \in \mathscr{Z}_t \cup \{j\}} \mathscr{O}_i^{h_{\mathfrak m},R_{\mathfrak m}}} \mathrm{d}\omega\, G(t)
			\ge 
			(1 - \nu) \int_{[0, R_{\mathfrak m})} \mathrm{d}\omega\, G(t).
		\end{equation}
	\end{itemize}

	We now  define the domains $\mathfrak{D}^Y$, $\mathfrak{D}^Z$ and the set $\mathscr{Y}_t$ as follows
	\[
	\mathfrak{D}^Y = [0, R_{\mathfrak m}) \setminus \bigcup_{i \in \mathscr{Z}_t} \mathscr{O}_i^{h_{\mathfrak m},R_{\mathfrak m}} = \bigcup_{i \in \mathscr{Y}_t} \Omega_i^{h_{\mathfrak m},R_{\mathfrak m}}, \quad
	\mathfrak{D}^Z = \bigcup_{i \in \mathscr{Z}_t} \Omega_i^{h_{\mathfrak m},R_{\mathfrak m}}.
	\]
	We then have, by \eqref{Propo:Collision:8},
	\begin{equation}\label{Propo:Collision:11}
		\int_{\mathfrak{D}^Y} \mathrm{d}\omega\, G(t) 
		\ge 
		\nu \int_{[0, R_{\mathfrak m})} \mathrm{d}\omega\, G(t).
	\end{equation}

$\mathfrak{D}^Y$ is our first subdomain. Next, we create the second and third subdomains.

\textit{Case (I) - Step (ii):} Creating the second and third subdomains.

	Let \( i_t^{*} \) be an index in \( \mathscr{Y}_t \) such that \( \Omega_{i_t^{*}}^{h_{\mathfrak m},R_{\mathfrak m}} \cap \bigcup_{i \in \mathscr{Z}_t} \mathscr{O}_i^{h_{\mathfrak m},R_{\mathfrak m}} = \emptyset \), and
	
	\begin{equation}\label{Propo:Collision:12}
		\int_{\Omega_{i_t^{*}}^{h_{\mathfrak m},R_{\mathfrak m}}} \mathrm{d}\omega\, G(t) 
		= 
		\max_{i \in \mathscr{Y}_t} \left\{ \int_{\Omega_i^{h_{\mathfrak m},R_{\mathfrak m}}} \mathrm{d}\omega\, G(t) \right\}.
	\end{equation}

	We define two disjoint sets \( \mathcal{S}_a \) and \( \mathcal{S}_b \) such that
	\[
	\mathcal{S}_a \cup \mathcal{S}_b = \mathscr{O}_{i_t^{*}}^{h_{\mathfrak m},R_{\mathfrak m}}, \quad
	\mathcal{S}_a \cap \mathcal{S}_b  = \emptyset, \quad
	\mathcal{S}_a \subset \bigcup_{i \in \mathscr{Z}_t} \mathscr{O}_i^{h_{\mathfrak m},R_{\mathfrak m}}, \quad
	\mathcal{S}_b \cap \bigcup_{i \in \mathscr{Z}_t} \mathscr{O}_i^{h_{\mathfrak m},R_{\mathfrak m}} = \emptyset.
	\]
	
	From the construction of the set $\mathscr{Z}_t$, we find
	\begin{equation}\label{Propo:Collision:13}	 	c_*\int_{\cup_{i\in\mathscr{Z}_t }\Omega_{i}^{h_{\mathfrak m},R_{\mathfrak m}} }\mathrm{d}\omega G(t) \ \ge \ \int_{\cup_{i\in\mathscr{Z}_t}\mathscr{O}_{i}^{h_{\mathfrak m},R_{\mathfrak m}} }\mathrm{d}\omega G(t),
	\end{equation}
	and from the construction of  $i_t^{*}$, we obtain
	\begin{equation}\label{Propo:Collision:14}	 	c_*\int_{ \Omega_{i_t^{*}}^{h_{\mathfrak m},R_{\mathfrak m}} }\mathrm{d}\omega G(t) \ \ge \ 3\int_{ \Omega_{i_t^{*}}^{h_{\mathfrak m},R_{\mathfrak m}} }\mathrm{d}\omega G(t) \ \ge \ \int_{\mathcal{S}_b }\mathrm{d}\omega G(t).
	\end{equation}
	Using \eqref{Propo:Collision:13}, \eqref{Propo:Collision:14}, and \eqref{Propo:Collision:10}, we find
	\begin{equation}\label{Propo:Collision:15}
		\begin{aligned}
			c_* \int_{\bigcup_{i \in \mathscr{Z}_t \cup \{i_t^{*}\}} \Omega_i^{h_{\mathfrak m},R_{\mathfrak m}}} \mathrm{d}\omega\, G(t) 
			&\ge \int_{\bigcup_{i \in \mathscr{Z}_t} \mathscr{O}_i^{h_{\mathfrak m},R_{\mathfrak m}}} \mathrm{d}\omega\, G(t) 
			+ \int_{\mathcal{S}_b} \mathrm{d}\omega\, G(t) \\
			&\ge \int_{\bigcup_{i \in \mathscr{Z}_t \cup \{i_t^{*}\}} \mathscr{O}_i^{h_{\mathfrak m},R_{\mathfrak m}}} \mathrm{d}\omega\, G(t) \ 
			\ge\ (1 - \nu) \int_{[0, R_{\mathfrak m})} \mathrm{d}\omega\, G(t).
		\end{aligned}
	\end{equation}

	Inequality \eqref{Propo:Collision:18}, in combination with \eqref{Propo:Collision:15}, gives the estimate
	\begin{equation}\label{Propo:Collision:19}
		c_* \int_{\bigcup_{i \in \mathscr{Z}_t} \Omega_i^{h_{\mathfrak m},R_{\mathfrak m}}} \mathrm{d}\omega\, G(t)  
		+ \frac{\nu c_*}{1000} \int_{[0, R_{\mathfrak m})} \mathrm{d}\omega\, G(t)  
		\ge (1 - \nu) \int_{[0, R_{\mathfrak m})} \mathrm{d}\omega\, G(t),
	\end{equation}
	which leads to
	\begin{equation}\label{Propo:Collision:20}
		c_* \int_{\bigcup_{i \in \mathscr{Z}_t} \Omega_i^{h_{\mathfrak m},R_{\mathfrak m}}} \mathrm{d}\omega\, G(t) 
		\ge \left(1 - \nu - \frac{\nu c_*}{1000} \right) \int_{[0, R_{\mathfrak m})} \mathrm{d}\omega\, G(t).
	\end{equation}
	
	Therefore, we obtain
	\begin{equation}\label{Propo:Collision:21}	
		\int_{\mathfrak{D}^Z} \mathrm{d}\omega G(t) 
		\ge \frac{1-2\nu}{c_*}\int_{[0,R_{\mathfrak m})}\mathrm{d}\omega G(t) \ge \frac{\nu}{10c_*}\int_{[0,R_{\mathfrak m})}\mathrm{d}\omega G(t).
	\end{equation}
	
	As a result, we have
	\begin{equation}\label{Propo:Collision:22}	
		\int_{\mathfrak{D}^Z} \mathrm{d}\omega G(t) 
		\ge \frac{\nu}{10c_*}\int_{[0,R_{\mathfrak m})}\mathrm{d}\omega G(t).
	\end{equation}
	
We proceed by further splitting the domain $\mathfrak{D}^Z$ into two subdomains. We repeat \eqref{Propo:Collision:10:1}--\eqref{Propo:Collision:10}, but with the constant \( 1 - \nu \) replaced by \( \frac{\nu}{20c_*} \). Let \( \mathscr{Z}^1_t = \{j_1, \dots, j_p\} \subset \mathscr{Z}_t \) denote the subset satisfying the following conditions:
\begin{itemize}
	
	\item The first index is defined as \( j_1 = i_t^o \).
	
	\item For \( l = 2, \dots, p \),
	\begin{equation}\label{Propo:Collision:22:1}
		\int_{\Omega_{j_l}^{h_{\mathfrak{m}}, R_{\mathfrak{m}}}} \mathrm{d}\omega\, G(t)
		= 
		\max_{i \in \{0, \dots, \mathscr{M}_{R_{\mathfrak{m}}, h_{\mathfrak{m}}} - 1\} \setminus \left( I_{j_1}^{h_{\mathfrak{m}}, R_{\mathfrak{m}}} \cup \cdots \cup I_{j_{l-1}}^{h_{\mathfrak{m}}, R_{\mathfrak{m}}} \right)} 
		\left\{ \int_{\Omega_i^{h_{\mathfrak{m}}, R_{\mathfrak{m}}}} \mathrm{d}\omega\, G(t) \right\},
	\end{equation}
	where the notation is as introduced in \eqref{Sec:DDM:3}.
	
	\item Furthermore, for all \( j \in \{0, \dots, \mathscr{M}_{R_{\mathfrak{m}}, h_{\mathfrak{m}}} - 1\} \setminus \left( I_{j_1}^{h_{\mathfrak{m}}, R_{\mathfrak{m}}} \cup \cdots \cup I_{j_p}^{h_{\mathfrak{m}}, R_{\mathfrak{m}}} \right) \), we have
	\begin{equation}\label{Propo:Collision:22:2}
		\int_{\bigcup_{i \in \mathscr{Z}^1_t \cup \{j\}} \mathscr{O}_i^{h_{\mathfrak{m}}, R_{\mathfrak{m}}}} \mathrm{d}\omega\, G(t)
		\ge 
		\frac{\nu}{20c_*} \int_{[0, R_{\mathfrak{m}})} \mathrm{d}\omega\, G(t).
	\end{equation}
	Moreover,
	\begin{equation}\label{Propo:Collision:22:3}
		\int_{\bigcup_{i \in \mathscr{Z}^1_t} \Omega_i^{h_{\mathfrak{m}}, R_{\mathfrak{m}}}} \mathrm{d}\omega\, G(t) 
		< 
		\frac{\nu}{20c_*} \int_{[0, R_{\mathfrak{m}})} \mathrm{d}\omega\, G(t).
	\end{equation}
	
\end{itemize}

It follows from \eqref{Propo:Collision:6} and \eqref{Propo:Collision:18} that this set is nonempty.

We now define 
\[
\mathscr{Z}^2_t := \mathscr{Z}_t \setminus \mathscr{Z}^1_t, \quad 
\mathfrak{D}^{Z^1} := \bigcup_{i \in \mathscr{Z}^1_t} \Omega_i^{h_{\mathfrak{m}}, R_{\mathfrak{m}}}, \quad 
\mathfrak{D}^{Z^2} := \bigcup_{i \in \mathscr{Z}^2_t} \Omega_i^{h_{\mathfrak{m}}, R_{\mathfrak{m}}}.
\]

It follows from \eqref{Propo:Collision:22} and \eqref{Propo:Collision:22:3} that
\begin{equation}\label{Propo:Collision:22:4}
	\int_{\mathfrak{D}^{Z^2}} \mathrm{d}\omega\, G(t) 
	\ge 
	\frac{\nu}{20c_*} \int_{[0, R_{\mathfrak{m}})} \mathrm{d}\omega\, G(t).
\end{equation}

Repeating the same argument as in 
\eqref{Propo:Collision:13}--\eqref{Propo:Collision:22}, with the constant \( 1 - \nu \) replaced by \( \frac{\nu}{20c_*} \), we find
\begin{equation}\label{Propo:Collision:22:5}	
	\int_{\mathfrak{D}^{Z^1}} \mathrm{d}\omega\, G(t) 
	\ge \frac{1}{c_*} \left( 1 - \nu - \frac{\nu c_*}{1000} \right) \int_{[0, R_{\mathfrak{m}})} \mathrm{d}\omega\, G(t)
	\ge 
	\frac{\nu}{20c_*} \int_{[0, R_{\mathfrak{m}})} \mathrm{d}\omega\, G(t).
\end{equation}

Finally, at the end of this case, we obtain three non-overlapping subdomains: \( \mathfrak{D}^{Y} \), \( \mathfrak{D}^{Z^1} \), and \( \mathfrak{D}^{Z^2} \), together with the three key inequalities \eqref{Propo:Collision:11}, \eqref{Propo:Collision:22:4}, and \eqref{Propo:Collision:22:5}.

We consider \( i \in \mathscr{Y} \), \( j \in \mathscr{Z}^1 \), and \( l \in \mathscr{Z}^2 \), with all indices distinct. Then it is clear that
\begin{equation}\label{Propo:Collision:22:5a}
	\min\{i, j, l\} + 2 < \min\big( \{i, j, l\} \setminus \{\min\{i, j, l\}\} \big) +1 < \max\{i, j, l\}.
\end{equation}

\textit{Case (II):} In this case, we assume that
\begin{equation}\label{Propo:Collision:22:6}
	\int_{\Omega_{i_t^o}^{h_{\mathfrak{m}}, R_{\mathfrak{m}}}} \mathrm{d}\omega\, G(t) 
	\ge \frac{\nu}{1000} \int_{[0, R_{\mathfrak{m}})} \mathrm{d}\omega\, G(t).
\end{equation}

We set \( \mathfrak{D}^{Y} := \Omega_{i_t^o}^{h_{\mathfrak{m}}, R_{\mathfrak{m}}} \) and \( \mathscr{Y} := \{i_t^o\} \), so that
\begin{equation}\label{Propo:Collision:22:7}
	\int_{\mathfrak{D}^{Y}} \mathrm{d}\omega\, G(t) 
	\ge \frac{\nu}{1000} \int_{[0, R_{\mathfrak{m}})} \mathrm{d}\omega\, G(t).
\end{equation}

By our assumption,
\begin{equation*}
	\int_{\mathscr{O}_{i_t^o}^{h_{\mathfrak{m}}, R_{\mathfrak{m}}}} \mathrm{d}\omega\, G(t) 
	< (1 - \nu) \int_{[0, R_{\mathfrak{m}})} \mathrm{d}\omega\, G(t),
\end{equation*}
which implies
\begin{equation}\label{Propo:Collision:22:8}
	\nu \int_{[0, R_{\mathfrak{m}})} \mathrm{d}\omega\, G(t) 
	\le \int_{[0, R_{\mathfrak{m}}) \setminus \mathscr{O}_{i_t^o}^{h_{\mathfrak{m}}, R_{\mathfrak{m}}}} \mathrm{d}\omega\, G(t).
\end{equation}

We define
\[
[0, R_{\mathfrak{m}}) \setminus \mathscr{O}_{i_t^o}^{h_{\mathfrak{m}}, R_{\mathfrak{m}}} 
= \bigcup_{i \in \mathscr{Z}_t} \Omega_i^{h_{\mathfrak{m}}, R_{\mathfrak{m}}},
\]
and set
\[
\mathscr{Z}_t' := \{i \in \mathscr{Z}_t \mid i > i_t^o\}, \quad 
\mathscr{Z}_t'' := \{i \in \mathscr{Z}_t \mid i < i_t^o\}.
\]

Since by \eqref{Propo:Collision:22:8},
\begin{equation*}
	\nu \int_{[0, R_{\mathfrak{m}})} \mathrm{d}\omega\, G(t) 
	\le \int_{\bigcup_{i \in \mathscr{Z}_t'}} \Omega_i^{h_{\mathfrak{m}}, R_{\mathfrak{m}}} \mathrm{d}\omega\, G(t)
	+ \int_{\bigcup_{i \in \mathscr{Z}_t''}} \Omega_i^{h_{\mathfrak{m}}, R_{\mathfrak{m}}} \mathrm{d}\omega\, G(t),
\end{equation*}
then either
\begin{equation}\label{Propo:Collision:22:9}
	\frac{\nu}{2} \int_{[0, R_{\mathfrak{m}})} \mathrm{d}\omega\, G(t) 
	\le \int_{\bigcup_{i \in \mathscr{Z}_t'}} \Omega_i^{h_{\mathfrak{m}}, R_{\mathfrak{m}}} \mathrm{d}\omega\, G(t),
\end{equation}
or
\begin{equation}\label{Propo:Collision:22:10}
	\frac{\nu}{2} \int_{[0, R_{\mathfrak{m}})} \mathrm{d}\omega\, G(t) 
	\le \int_{\bigcup_{i \in \mathscr{Z}_t''}} \Omega_i^{h_{\mathfrak{m}}, R_{\mathfrak{m}}} \mathrm{d}\omega\, G(t).
\end{equation}

If \eqref{Propo:Collision:22:9} holds, we set 
\[
\{i \in \mathscr{Z}_t \mid i > i_t^o\} = \mathscr{Z}_t' = \mathscr{Z}_t^1 = \mathscr{Z}_t^2,
\]
and define
\[
\bigcup_{i \in \mathscr{Z}_t^1} \Omega_i^{h_{\mathfrak{m}}, R_{\mathfrak{m}}} 
= \bigcup_{i \in \mathscr{Z}_t^2} \Omega_i^{h_{\mathfrak{m}}, R_{\mathfrak{m}}} 
=: \mathfrak{D}^{Z^1} = \mathfrak{D}^{Z^2}.
\]

If \eqref{Propo:Collision:22:10} holds, we set 
\[
\{i \in \mathscr{Z}_t \mid i < i_t^o\} = \mathscr{Z}_t'' = \mathscr{Z}_t^1, \quad \mathscr{Y} = \mathscr{Z}_t^2,
\]
and define
\[
\bigcup_{i \in \mathscr{Z}_t^1} \Omega_i^{h_{\mathfrak{m}}, R_{\mathfrak{m}}} =: \mathfrak{D}^{Z^1}, \quad
\bigcup_{i \in \mathscr{Z}_t^2} \Omega_i^{h_{\mathfrak{m}}, R_{\mathfrak{m}}} =: \mathfrak{D}^{Z^2}.
\]

In both cases, we have
\begin{equation}\label{Propo:Collision:22:11}
	C_0 \nu \int_{[0, R_{\mathfrak{m}})} \mathrm{d}\omega\, G(t) 
	\le \int_{\mathfrak{D}^{Z^1}} \mathrm{d}\omega\, G(t), \quad 
	C_0 \nu \int_{[0, R_{\mathfrak{m}})} \mathrm{d}\omega\, G(t) 
	\le \int_{\mathfrak{D}^{Z^2}} \mathrm{d}\omega\, G(t),
\end{equation}
where \( C_0 > 0 \) is a universal constant that may vary across estimates.

Finally, at the end of this case, we obtain three subdomains: 
\( \mathfrak{D}^{Y} \), \( \mathfrak{D}^{Z^1} \), and \( \mathfrak{D}^{Z^2} \), 
together with the key inequalities \eqref{Propo:Collision:22:7} and \eqref{Propo:Collision:22:11}. 

We consider \( i \in \mathscr{Y} \), \( j \in \mathscr{Z}^1 \), \( j \in \mathscr{Z}^2 \). Then it is clear that
\begin{equation}\label{Propo:Collision:22:12}
	\min\{i, j, l\} + 1 < \min\big( \{i, j, l\} \setminus \{\min\{i, j, l\}\} \big).
\end{equation}

	{\it Step 3: The main estimate.}

We set 
\[
\mathfrak{G} := \left\{(i, j, l) \in \mathbb{Z}^3 \,\middle|\, i, j, l \ge 0,\ \mathscr{M}_{R_{\mathfrak{m}}, h_{\mathfrak{m}}} > i \ge j > l + 1 \right\}.
\]

We now consider the two cases analyzed in Step 2.

\textit{Case (I):} 
We make use of the three non-overlapping subdomains 
\( \mathfrak{D}^{Y} \), \( \mathfrak{D}^{Z^1} \), and \( \mathfrak{D}^{Z^2} \),
as well as the index condition \eqref{Propo:Collision:22:5a}, to obtain the bound
\begin{equation*}
	\begin{aligned}
		& \sum_{(i, j, l) \in \mathfrak{G}} 
		\iiint_{\Omega_{i,j,l}^{h_{\mathfrak{m}}, R_{\mathfrak{m}}}} 
		\mathrm{d}\omega_1\, \mathrm{d}\omega_2\, \mathrm{d}\omega\,
		G_1 G_2 G \cdot \mathbf{1}_{\omega + \omega_1 - \omega_2 \ge 0} \cdot 
		\frac{1}{\Theta(h_{\mathfrak{m}})} \cdot 
		\frac{(2R_{\mathfrak{m}})^{2\delta + \frac{3}{2}}}{h_{\mathfrak{m}}^2} 
		\left( \frac{4R_{\mathfrak{m}}}{h_{\mathfrak{m}}} \right)^\mu \\[0.5em]
		& \quad \times 
		\left[ \frac{\max\{|\omega_1 - \omega_2|,\ |\omega_2 - \omega|\}}{2(\omega + \omega_1)} \right]^\mu 
		\cdot \frac{(\omega_{\mathrm{Mid}} - \omega_{\mathrm{Min}})^2}{(2\omega_{\mathrm{Mid}} - \omega_{\mathrm{Min}})^{\frac{3}{2}}}	
		\cdot \frac{\Theta(\omega_{\mathrm{Max}} - \omega_{\mathrm{Min}} + \omega_{\mathrm{Mid}})}{|\omega_{\mathrm{Mid}}|^\delta |\omega_{\mathrm{Max}}|^\delta} \\[0.75em]
		\ge\; &  
		\sum_{\substack{(i,j,l) \in \mathbb{Z}^3 \\ i,j,l \ge 0,\ \mathscr{M}_{R_{\mathfrak{m}}, h_{\mathfrak{m}}} > i \ge j > l + 1,\ i \ge j + 2}} 
		\iiint_{\Omega_{i,j,l}^{h_{\mathfrak{m}}, R_{\mathfrak{m}}}} 
		\mathrm{d}\omega_1\, \mathrm{d}\omega_2\, \mathrm{d}\omega\,
		G_1 G_2 G \cdot \mathbf{1}_{\omega + \omega_1 - \omega_2 \ge 0} \\[0.75em]
		\ge\; & 
		\int_{\mathfrak{D}^{Y}} \mathrm{d}\omega \, G
		\int_{\mathfrak{D}^{Z^1}} \mathrm{d}\omega \, G 
		\int_{\mathfrak{D}^{Z^2}} \mathrm{d}\omega \, G,
	\end{aligned}
\end{equation*}
which, together with the key inequalities 
\eqref{Propo:Collision:11}, 
\eqref{Propo:Collision:22:4}, and 
\eqref{Propo:Collision:22:5}, yields
\begin{equation}\label{Propo:Collision:34}
	\begin{aligned}
		\left( \int_{[0, R_{\mathfrak{m}})} \mathrm{d}\omega\, G(t) \right)^3 
		\ \le\ &    
		\sum_{(i,j,l)\in\mathfrak{G}} 
		\iiint_{\Omega_{i,j,l}^{h_{\mathfrak{m}}, R_{\mathfrak{m}}}} 
		\mathrm{d}\omega_1\,\mathrm{d}\omega_2\,\mathrm{d}\omega\ 
		G_1 G_2 G\, \mathbf{1}_{\omega + \omega_1 - \omega_2 \ge 0} \\[0.5em]
		& \times \frac{C_0}{\Theta(h_{\mathfrak{m}})\nu^4} 
		\cdot \frac{(2R_{\mathfrak{m}})^{2\delta+\frac{3}{2}}}{h_{\mathfrak{m}}^2} 
		\left( \frac{4R_{\mathfrak{m}}}{h_{\mathfrak{m}}} \right)^\mu 
		\left[ \frac{\max\{|\omega_1 - \omega_2|,\ |\omega_2 - \omega|\}}{2(\omega + \omega_1)} \right]^\mu \\[0.5em]
		& \times 
		\frac{(\omega_{\mathrm{Mid}} - \omega_{\mathrm{Min}})^2}{(2\omega_{\mathrm{Mid}} - \omega_{\mathrm{Min}})^{\frac{3}{2}}} 
		\cdot \frac{\Theta(\omega_{\mathrm{Max}} - \omega_{\mathrm{Min}} + \omega_{\mathrm{Mid}})}{|\omega_{\mathrm{Mid}}|^\delta |\omega_{\mathrm{Max}}|^\delta},
	\end{aligned}
\end{equation}
for some universal constant \( C_0 > 0 \), which may vary from estimate to estimate.

\textit{Case (II):} 
We make use of the  subdomains 
\( \mathfrak{D}^{Y} \), \( \mathfrak{D}^{Z^1} \), and \( \mathfrak{D}^{Z^2} \), 
as well as the index inequality \eqref{Propo:Collision:22:12}, to obtain the bound:
\begin{equation*}
	\begin{aligned}
		& \sum_{(i,j,l) \in \mathfrak{G}} 
		\iiint_{\Omega_{i,j,l}^{h_{\mathfrak{m}}, R_{\mathfrak{m}}}} 
		\mathrm{d}\omega_1\, \mathrm{d}\omega_2\, \mathrm{d}\omega\,
		G_1 G_2 G \cdot \mathbf{1}_{\omega + \omega_1 - \omega_2 \ge 0} \cdot 
		\frac{1}{\Theta(h_{\mathfrak{m}})} \cdot 
		\frac{(2R_{\mathfrak{m}})^{2\delta + \frac{3}{2}}}{h_{\mathfrak{m}}^2} 
		\left( \frac{4R_{\mathfrak{m}}}{h_{\mathfrak{m}}} \right)^\mu \\[0.5em]
		& \quad \times 
		\left[ \frac{\max\{|\omega_1 - \omega_2|,\ |\omega_2 - \omega|\}}{2(\omega + \omega_1)} \right]^\mu 
		\cdot \frac{(\omega_{\mathrm{Mid}} - \omega_{\mathrm{Min}})^2}{(2\omega_{\mathrm{Mid}} - \omega_{\mathrm{Min}})^{\frac{3}{2}}}	
		\cdot \frac{\Theta(\omega_{\mathrm{Max}} - \omega_{\mathrm{Min}} + \omega_{\mathrm{Mid}})}{|\omega_{\mathrm{Mid}}|^\delta |\omega_{\mathrm{Max}}|^\delta} \\[0.75em]
		\ge\; &  
		\sum_{\substack{(i,j,l) \in \mathbb{Z}^3 \\ i,j,l \ge 0,\ \mathscr{M}_{R_{\mathfrak{m}}, h_{\mathfrak{m}}} > i \ge j > l + 1,\ i \ge j + 2}} 
		\iiint_{\Omega_{i,j,l}^{h_{\mathfrak{m}}, R_{\mathfrak{m}}}} 
		\mathrm{d}\omega_1\, \mathrm{d}\omega_2\, \mathrm{d}\omega\,
		G_1 G_2 G \cdot \mathbf{1}_{\omega + \omega_1 - \omega_2 \ge 0} \\[0.75em]
		\ge\; & 
		\int_{\mathfrak{D}^{Y}} \mathrm{d}\omega \, G
		\int_{\mathfrak{D}^{Z^1}} \mathrm{d}\omega \, G 
		\int_{\mathfrak{D}^{Z^2}} \mathrm{d}\omega \, G.
	\end{aligned}
\end{equation*}

Combining this with the key inequalities 
\eqref{Propo:Collision:22:7} and 
\eqref{Propo:Collision:22:11}, we conclude:
\begin{equation}\label{Propo:Collision:34a}
	\begin{aligned}
		\left( \int_{[0, R_{\mathfrak{m}})} \mathrm{d}\omega\, G(t) \right)^3 
		\ \le\ &    
		\sum_{(i,j,l) \in \mathfrak{G}} 
		\iiint_{\Omega_{i,j,l}^{h_{\mathfrak{m}}, R_{\mathfrak{m}}}} 
		\mathrm{d}\omega_1\, \mathrm{d}\omega_2\, \mathrm{d}\omega\,
		G_1 G_2 G \cdot \mathbf{1}_{\omega + \omega_1 - \omega_2 \ge 0} \\[0.5em]
		& \times \frac{C_0}{\Theta(h_{\mathfrak{m}}) \nu^4} 
		\cdot \frac{(2R_{\mathfrak{m}})^{2\delta + \frac{3}{2}}}{h_{\mathfrak{m}}^2} 
		\left( \frac{4R_{\mathfrak{m}}}{h_{\mathfrak{m}}} \right)^\mu 
		\left[ \frac{\max\{|\omega_1 - \omega_2|,\ |\omega_2 - \omega|\}}{2(\omega + \omega_1)} \right]^\mu \\[0.5em]
		& \times 
		\frac{(\omega_{\mathrm{Mid}} - \omega_{\mathrm{Min}})^2}{(2\omega_{\mathrm{Mid}} - \omega_{\mathrm{Min}})^{\frac{3}{2}}} 
		\cdot \frac{\Theta(\omega_{\mathrm{Max}} - \omega_{\mathrm{Min}} + \omega_{\mathrm{Mid}})}{|\omega_{\mathrm{Mid}}|^\delta |\omega_{\mathrm{Max}}|^\delta},
	\end{aligned}
\end{equation}
for some universal constant \( C_0 > 0 \), which may vary from estimate to estimate.

By \eqref{Lemma:Concave:1}, we obtain the following bound for \( \alpha = \frac{1}{2} \):
\begin{equation}\label{Propo:Collision:35}
	\begin{aligned}
		\mathscr{M} + \mathscr{E} \ \ge\ 
		& \, \mathcal{C}_1 \int_0^T \mathrm{d}t \iiint_{\mathbb{R}_+^3} \mathrm{d}\omega_1\, \mathrm{d}\omega_2\, \mathrm{d}\omega\,
		G_1 G_2 G\, \frac{\frac{1}{4}(\omega_{\mathrm{Mid}} - \omega_{\mathrm{Min}})^2}{(2\omega_{\mathrm{Mid}} - \omega_{\mathrm{Min}})^{\frac{3}{2}}} \\[0.5em]
		& \quad \times \frac{\Theta(\omega_{\mathrm{Max}} - \omega_{\mathrm{Min}} + \omega_{\mathrm{Mid}})}{|\omega_{\mathrm{Mid}}||\omega_{\mathrm{Max}}|}    
		\left[ \frac{\max\{|\omega_1 - \omega_2|,\ |\omega_2 - \omega|\}}{2(\omega + \omega_1)} \right]^\mu \\[0.75em]
		\gtrsim\ 
		& \int_0^T \mathrm{d}t \iiint_{\mathbb{R}_+^3} \mathrm{d}\omega_1\, \mathrm{d}\omega_2\, \mathrm{d}\omega\,
		G_1 G_2 G\, \frac{(\omega_{\mathrm{Mid}} - \omega_{\mathrm{Min}})^2}{(2\omega_{\mathrm{Mid}} - \omega_{\mathrm{Min}})^{\frac{3}{2}}} \\[0.5em]
		& \quad \times \frac{\Theta(\omega_{\mathrm{Max}} - \omega_{\mathrm{Min}} + \omega_{\mathrm{Mid}})}{|\omega_{\mathrm{Mid}}|^\delta |\omega_{\mathrm{Max}}|^\delta}    
		\left[ \frac{\max\{|\omega_1 - \omega_2|,\ |\omega_2 - \omega|\}}{2(\omega + \omega_1)} \right]^\mu.
	\end{aligned}
\end{equation}

Combining the bounds \eqref{Propo:Collision:34}, \eqref{Propo:Collision:34a}, and \eqref{Propo:Collision:35}, we conclude that
\begin{equation}\label{Propo:Collision:36}
	\frac{\mathcal{C}_1(\mathscr{M} + \mathscr{E})}{\Theta(h_{\mathfrak{m}})\nu^4} \cdot \frac{R_{\mathfrak{m}}^{2\delta + \frac{3}{2} + \mu}}{h_{\mathfrak{m}}^{2 + \mu}} 
	\ \ge\ 
	\int_{\Xi_{\nu, R_{\mathfrak{m}}}^2} \mathrm{d}t \left( \int_{[0, R_{\mathfrak{m}})} \mathrm{d}\omega\, G(t) \right)^3,
\end{equation}
for some universal constant \( \mathcal{C}_1 > 0 \), which completes the proof of the proposition.

\end{proof}

\subsection{The  second multiscale estimates}

\begin{proposition}[Second multiscale estimates] 	We assume Assumption A and Assumption B. 

	\label{Lemma:Growth1} Let \( T_0 \) be as in \eqref{T0}.  We choose $0\le T<T_0$ if $T_0>0$ and $ T=0$ if $T_0=0$. By the definitions in 
	\eqref{Sec:Growthlemmas:2}, \eqref{Sec:Growthlemmas:3}, and \eqref{Sec:Growthlemmas:4}, 
	there exists a universal constant \( \mathfrak{C}_1 > 0 \)  independent of $\delta, \mu, \varrho, \varepsilon,  T, T_0$ such that
	\begin{equation}\label{Lemma:Growth1:1}
		\mathfrak{C}_1(\mathscr{M} + \mathscr{E}) \, R_\mathfrak{m}^{\frac{1}{4}(2\delta  - \frac{1}{2} - \varrho)}	
		\ge \left| \mathscr{B}_{\mathfrak{m}}^T \right|,
	\end{equation}
	for all \( \mathfrak{m} \in \mathbb{N} \), \( \mathfrak{m} > 10^6 \), and \( T \in [0, T_0) \).
	
	We define 
	\[
	c_\mathcal{A} = \min\left\{ \tfrac{1}{4}(2\delta  - \tfrac{1}{2} - \varrho),\, 2\delta - \varepsilon - \varrho \right\}.
	\]
	Then, there exists $\mathfrak{M}_*>0$ such that for all $\mathfrak{m} > \mathfrak{M}_*$, $\mathfrak{m} \in \mathbb{N}$, and for all $T \in [0, T_0)$, we have
	\begin{equation} \label{Lemma:Growth1:2}
		\left\vert \bigcup_{i=\mathfrak{m}}^\infty \left( \mathcal{A}_i^T \setminus \mathscr{D}_i^T \right) \right\vert 
		\le {C}_\mathcal{A} \left( R_\mathfrak{m} \right)^{c_\mathcal{A}},
	\end{equation}
	for a universal constant ${C}_\mathcal{A} > 0$, independent of $\delta, \mu, \varrho, \varepsilon, T, T_0$.

\end{proposition}

\begin{proof}
	
	We will first prove \eqref{Lemma:Growth1:1}.	
	
	Note that the same argument used in the proof of Proposition \ref{Propo:Collision} can be reiterated with 
	\[
	\nu = 1 - \frac{1}{2^{\rho}},
	\]
	where \( \rho \) is defined in \eqref{rho}. We then obtain the same estimate as in \eqref{Propo:Collision:36}:
	\begin{equation*}
		\begin{aligned}
			\frac{\mathcal{C}_1(\mathscr{M} + \mathscr{E})}{\Theta(h_\mathfrak{m})\nu^4} \cdot 
			\frac{R_\mathfrak{m}^{2\delta + \frac{3}{2} + \mu}}{h_\mathfrak{m}^{2 + \mu}}		
			\ \ge\ &\int_{0}^{T} \mathrm{d}t \, \chi_{\mathscr{B}_{\mathfrak{m}}^T}(t) 
			\left( \int_{[0, R_\mathfrak{m})} \mathrm{d}\omega \, G(t) \right)^3,
		\end{aligned}
	\end{equation*}
	which, in combination with \eqref{Sec:Growthlemmas:2}, implies
	\begin{equation}\label{Lemma:Growth1:E1}
		\begin{aligned}
			\frac{{C}_0(\mathscr{M} + \mathscr{E})}{ \nu^4} \cdot 
			\frac{R_\mathfrak{m}^{2\delta + \frac{3}{2} + \mu}}{h_\mathfrak{m}^{2 + \mu + \varrho}}	
			\ \ge\ &\mathcal{C}_* \int_{0}^{T} \mathrm{d}t \, \chi_{\mathscr{B}_{\mathfrak{m}}^T}(t) R_\mathfrak{m}^{3\rho}
			\ = \ \mathcal{C}_* \left| \mathscr{B}_{\mathfrak{m}}^T \right| R_\mathfrak{m}^{3\rho},
		\end{aligned}
	\end{equation}
	for a universal constant \( {C}_0 > 0 \) that may vary from line to line.
	
	Since \( R_\mathfrak{m} = 2^{\mathfrak{N}_\mathfrak{m}} h_\mathfrak{m} \), we deduce from \eqref{Lemma:Growth1:E1}
	\begin{equation}\label{Lemma:Growth1:E2}
		\begin{aligned}
			\frac{{C}_0(\mathscr{M} + \mathscr{E})}{ \nu^4} \cdot 
			{R_\mathfrak{m}^{2\delta - \frac{1}{2}  - 3\rho - \varrho}} \cdot 2^{\mathfrak{N}_\mathfrak{m}(2 + \mu + \varrho)}	
			\ \ge\ & \mathcal{C}_* \left| \mathscr{B}_{\mathfrak{m}}^T \right|.
		\end{aligned}
	\end{equation}
By \eqref{rho}, we have
\[
R_\mathfrak{m}^{-3\rho} = 2^{3\rho\mathfrak{m}} < 2^{\tfrac{\mathfrak{m}}{2} \left( 2\delta - \tfrac{1}{2} - \varrho  \right)} = R_\mathfrak{m}^{-\tfrac{1}{2} \left( 2\delta - \tfrac{1}{2} - \varrho  \right)},
\]
and hence we deduce
\begin{equation*}
	\begin{aligned}
		R_\mathfrak{m}^{2\delta - \frac{1}{2}  - 3\rho - \varrho} \cdot 2^{\mathfrak{N}_\mathfrak{m}(2 + \mu + \varrho)} 
		&\le 2^{-\mathfrak{m} \left( 2\delta  - \frac{1}{2} - \varrho \right)/2} \cdot 2^{\left\lfloor \frac{\mathfrak{m}}{4(2 + \mu + \varrho)} \left( 2\delta - \frac{1}{2} - \varrho  \right) \right\rfloor (2 + \mu + \varrho)} \\
		&\le 2^{-\frac{\mathfrak{m}}{4} \left( 2\delta  - \frac{1}{2} - \varrho \right)},
	\end{aligned}
\end{equation*}

	which in combination with \eqref{Lemma:Growth1:E2}, yields
	\begin{equation}\label{Lemma:Growth1:E3}
		\begin{aligned}
			\mathfrak{C}_1(\mathscr{M} + \mathscr{E}) \, R_\mathfrak{m}^{\frac{1}{4}\left(2\delta  - \frac{1}{2} - \varrho\right)}	
			\ \ge\ & \left| \mathscr{B}_{\mathfrak{m}}^T \right|,
		\end{aligned}
	\end{equation}
	for some universal constant \( \mathfrak{C}_1 > 0 \),  
	leading to the proof of \eqref{Lemma:Growth1:1}.

	To continue, we will prove~\eqref{Lemma:Growth1:2}.
	
	Combining~\eqref{Lemma:Growth1:1} and~\eqref{Lemma:Growth2:1}, we obtain
	\begin{equation*}
		\begin{aligned}
			\left| \mathscr{B}_\mathfrak{m}^T \cup \mathscr{C}_{\mathfrak{m}}^T \right|  
			\le\ & \mathfrak{C}_1(\mathscr{M} + \mathscr{E}) 
			R_\mathfrak{m}^{\frac{1}{4}(2\delta - \frac{1}{2} - \varrho)} 
			+ C_\mathscr{C} R_{\mathfrak{m}}^{2\delta - \varepsilon - \varrho} 
			\le C_0 R_{\mathfrak{m}}^{c_\mathcal{A}},
		\end{aligned}
	\end{equation*}
	where $C_0 > 0$ is a universal constant that may vary from line to line. 
	
	Next, we estimate
	\begin{equation*}
		\left\vert\bigcup_{i=\mathfrak{m}}^\infty \left( \mathcal{A}_i^T \setminus \mathscr{D}_i^T \right) \right\vert 
		\le \left\vert \bigcup_{i=\mathfrak{m}}^\infty \left( \mathscr{B}_i^T \cup \mathscr{C}_i^T \right) \right\vert
		\le \sum_{i=\mathfrak{m}}^{\infty} \left| \mathscr{B}_i^T \cup \mathscr{C}_i^T \right| 
		\le {C}_0 \sum_{i=\mathfrak{m}}^{\infty} R_i^{c_\mathcal{A}} 
		= \frac{{C}_0}{1 - 2^{-c_\mathcal{A}}} R_{\mathfrak{m}}^{c_\mathcal{A}}.
	\end{equation*}
	
	Hence,~\eqref{Lemma:Growth1:2} is proved.

\end{proof}
 \section{The final multiscale estimates}\label{Sec:Third} 
 \begin{proposition}[Third multiscale estimates] 
 	\label{Lemma:Growth3}  	We assume Assumption A and Assumption B. 
 Let \( T_0 \) be as in \eqref{T0}.  We choose $0\le T<T_0$ if $T_0>0$ and $ T=0$ if $T_0=0$. By the definitions in 
 	\eqref{Sec:Growthlemmas:2}, \eqref{Sec:Growthlemmas:3}, and \eqref{Sec:Growthlemmas:4}, 
 	for all  \( \rho > 0 \) satisfying \eqref{rho}, there exists \( \mathfrak{M}^*(\rho) > 0 \) such that for \( \mathfrak{m} > \mathfrak{M}^*(\rho) \), we have
 	\begin{equation} \label{Lemma:Growth3:1}
 		\left| \mathcal{A}_\mathfrak{m}^T \right| \leq C_\mathcal{A} \left( R_\mathfrak{m} \right)^{c_\mathcal{A}},
 	\end{equation}
 	where we have used the same notations as in \eqref{Lemma:Growth1:2}.
 	
 \end{proposition}
 
\begin{proof}

For each $\rho$ satisfying \eqref{rho}, we will prove the existence of $\rho> 0$ and $\mathfrak{M}^*(\rho) > 0$ such that for all $\mathfrak{m} > \mathfrak{M}^*(\rho)$,
\begin{equation} \label{Lemma:Growth3:E1}
	\mathcal{A}_{\mathfrak{m}}^T \subset \bigcup_{i=\mathfrak{m}}^\infty \left( \mathcal{A}_i^T \setminus \mathscr{D}_i^T \right).
\end{equation}
By definition, we observe that
\[
\bigcup_{i=\mathfrak{m}}^\infty \left( \mathcal{A}_i^T \setminus \mathcal{A}_{i+1}^T \right) \subset \bigcup_{i=\mathfrak{m}}^\infty \left( \mathcal{A}_i^T \setminus \mathscr{D}_i^T \right).
\]
Thus, to prove \eqref{Lemma:Growth3:E1}, it suffices to show the existence of  $\mathfrak{M}^*(\rho) > 0$ such that for all $\mathfrak{m} > \mathfrak{M}^*(\rho)$,
\begin{equation} \label{Lemma:Growth3:E2}
	\mathcal{A}_{\mathfrak{m}}^T \setminus \left( \bigcup_{i=\mathfrak{m}}^\infty \left( \mathcal{A}_i^T \setminus \mathcal{A}_{i+1}^T \right) \right) = \emptyset.
\end{equation}

We argue by contradiction. Suppose that there exists $\rho > 0$ satisfying \eqref{rho} and a sequence
\[
\mathfrak{m}(\rho) = \mathfrak{m}_0(\rho) < \mathfrak{m}_1(\rho) < \cdots
\]
such that for each $j \in \mathbb{N}$,
\[
\mathcal{A}_{\mathfrak{m}_j(\rho)}^T \setminus \left( \bigcup_{i=\mathfrak{m}_j(\rho)}^\infty \left( \mathcal{A}_i^T \setminus \mathcal{A}_{i+1}^T \right) \right) \ne \emptyset.
\]

That is, there exists a sequence of time points at which elements of $\mathcal{A}_{\mathfrak{m}_j(\rho)}^T$ do not fall within the nested differences $\mathcal{A}_i^T \setminus \mathcal{A}_{i+1}^T$. This contradiction will yield the desired conclusion.

We divide the proof into several steps.

\medskip
\noindent
\textit{Step 1:} 
Let $t(\rho) \in \mathcal{A}_{\mathfrak{m}(\rho)}^T \setminus \left( \bigcup_{i = \mathfrak{m}(\rho)}^\infty \left( \mathcal{A}_i^T \setminus \mathcal{A}_{i+1}^T \right) \right) \subset [0, T]$. 
We will prove that
\begin{equation} \label{Lemma:Growth3:E3}
	\int_{[0, R_i)} \mathrm{d}\omega \, G\left( t(\rho), \omega \right) \geq \mathcal{C}_* R_i^\rho, 
\end{equation}
for all $i \ge \mathfrak{m}(\rho)$.

To establish this, suppose by contradiction that there exists $i_0 \ge \mathfrak{m}(\rho)$ such that \eqref{Lemma:Growth3:E3} fails to hold. Note, however, that by assumption $t(\rho) \in \mathcal{A}_{\mathfrak{m}(\rho)}^T$, so \eqref{Lemma:Growth3:E3} must hold for $i = \mathfrak{m}(\rho)$. 

Therefore, we may assume that \eqref{Lemma:Growth3:E3} holds for all $i = \mathfrak{m}(\rho), \mathfrak{m}(\rho)+1, \dots, i_0 - 1$, and fails for $i = i_0$. This would imply that
\[
t(\rho) \in \mathcal{A}_{i_0 - 1}^T \setminus \mathcal{A}_{i_0}^T,
\]
contradicting the fact that $t(\rho)$ does not belong to any of the sets $\mathcal{A}_i^T \setminus \mathcal{A}_{i+1}^T$ for $i \ge \mathfrak{m}(\rho)$.

Hence, we conclude that \eqref{Lemma:Growth3:E3} must hold for all $i \ge \mathfrak{m}(\rho)$.

\medskip
\noindent
\textit{Step 2:} 
Replacing
\[
\mathfrak{m}(\rho) = \mathfrak{m}_0(\rho) < \mathfrak{m}_1(\rho) < \cdots
\]
by
\[
\mathfrak{m}(\rho) = \mathfrak{m}_j(\rho) < \mathfrak{m}_{j+1}(\rho) < \cdots
\]
for $j$ sufficiently large. Then, we can assume that $\mathfrak{m}(\rho)$ is sufficiently  large.

 In this step, we will prove that
\begin{equation} \label{Lemma:Growth3:E4}
	\int_{[0, 4R_i)} \mathrm{d}\omega \, G(t, \omega) \geq \frac{\mathcal{C}_*}{3} R_i^\rho, \quad \forall t > t(\rho),
\end{equation}
for all $i \ge \mathfrak{m}(\rho)$.

Following \eqref{Lemma:ConcreteSuper:E2}, we use $\phi(\omega) = \exp\left[\frac{1}{R_{i}}\left(4{R_{i}} - \omega\right)_+\right] - 1$ as a test function. Similar to \eqref{Lemma:Growth2:E1}, we find
\begin{equation} \label{Lemma:Growth3:E5}
	\partial_t \left( \int_{\mathbb{R}^{+}} \mathrm{d}\omega\, \phi(\omega) G(t, \omega) \right)
	= \mathfrak{G}_1 + \mathfrak{G}_2 + \mathfrak{G}_3,
\end{equation}
where
\begin{equation} \label{Lemma:Growth3:E6}
	\begin{aligned}
		\mathfrak{G}_1 :=\; & 2c_{12} \iint_{\omega_1 > \omega_2} \mathrm{d}\omega_1 \mathrm{d}\omega_2\, \mathfrak{A}(\omega_1) \mathfrak{A}(\omega_2)\, f(\omega_1) f(\omega_2) \\
		& \quad \times \Big[ \mathfrak{A}(\omega_1 + \omega_2) \big( \phi(\omega_1 + \omega_2) - \phi(\omega_1) - \phi(\omega_2) \big) \\
		& \qquad\quad - \mathfrak{A}(\omega_1 - \omega_2) \big( \phi(\omega_1) - \phi(\omega_1 - \omega_2) - \phi(\omega_2) \big) \Big] \\
		& + c_{12} \iint_{\omega_1 = \omega_2} \mathrm{d}\omega_1 \mathrm{d}\omega_2\, \mathfrak{A}(\omega_1)^2 f(\omega_1)^2\, \mathfrak{A}(2\omega_1)\, \big[ \phi(2\omega_1) - 2\phi(\omega_1) \big],
	\end{aligned}
\end{equation}

\begin{equation} \label{Lemma:Growth3:E7}
	\begin{aligned}
		\mathfrak{G}_2 :=\; & c_{22} \iiint_{\mathbb{R}_+^3} \mathrm{d}\omega_1\,\mathrm{d}\omega_2\,\mathrm{d}\omega\, f(\omega_1) f(\omega_2) f(\omega)\, \mathbf{1}_{\omega+\omega_1-\omega_2 \ge 0} \\
		& \times \left[ \frac{\max\{|\omega_1 - \omega_2|,\ |\omega_2 - \omega|\}}{2(\omega + \omega_1)} \right]^\mu \cdot \Big\{ 
		[-\phi(\omega_{\text{Max}}) - \phi(\omega_{\text{Min}}) + \phi(\omega_{\text{Mid}}) \\
		& \quad + \phi(\omega_{\text{Max}} + \omega_{\text{Min}} - \omega_{\text{Mid}})] \prod_{x \in \{\omega_{\text{Max}}, \omega_{\text{Min}}, \omega_{\text{Mid}}, \omega_{\text{Max}} + \omega_{\text{Min}} - \omega_{\text{Mid}}\}}\Theta(x)\, |k_{\text{Min}}| \\
		& + [-\phi(\omega_{\text{Max}}) - \phi(\omega_{\text{Mid}}) + \phi(\omega_{\text{Min}}) + \phi(\omega_{\text{Max}} + \omega_{\text{Mid}} - \omega_{\text{Min}})] \\
		& \quad \times \prod_{x \in \{\omega_{\text{Max}}, \omega_{\text{Mid}}, \omega_{\text{Min}}, \omega_{\text{Max}} + \omega_{\text{Mid}} - \omega_{\text{Min}}\}}\Theta(x)\, |k_{\text{Min}}| \\
		& + [-\phi(\omega_{\text{Min}}) - \phi(\omega_{\text{Mid}}) + \phi(\omega_{\text{Max}}) + \phi(\omega_{\text{Min}} + \omega_{\text{Mid}} - \omega_{\text{Max}})] \\
		& \quad \times \mathbf{1}_{\omega_{\text{Min}} + \omega_{\text{Mid}} - \omega_{\text{Max}} \ge 0} \prod_{x \in \{\omega_{\text{Max}}, \omega_{\text{Mid}}, \omega_{\text{Min}}, \omega_{\text{Min}} + \omega_{\text{Mid}} - \omega_{\text{Max}}\}}\Theta(x) \\
		& \quad \times \min \left\{ |k(\omega_{\text{Max}})|, |k(\omega_{\text{Min}})|, |k(\omega_{\text{Mid}})|, |k(\omega_{\text{Min}} + \omega_{\text{Mid}} - \omega_{\text{Max}})| \right\}
		\Big\}
	\end{aligned}
\end{equation}

and
\begin{equation} \label{Lemma:Growth3:E8}
	\begin{aligned}
		\mathfrak{G}_3 :=\; & 3c_{31} \iiint_{\omega_1 > \omega_2 + \omega_3} \mathrm{d}\omega_1\, \mathrm{d}\omega_2\, \mathrm{d}\omega_3\, \mathfrak{A}(\omega_1) \mathfrak{A}(\omega_2) \mathfrak{A}(\omega_3)\, f(\omega_1) f(\omega_2) f(\omega_3) \\
		& \times \Big[ \mathfrak{A}(\omega_1 + \omega_2 + \omega_3) \big( \phi(\omega_1 + \omega_2 + \omega_3) - \phi(\omega_1) - \phi(\omega_2) - \phi(\omega_3) \big) \\
		& \quad - \mathfrak{A}(\omega_1 - \omega_2 - \omega_3) \big( \phi(\omega_1) - \phi(\omega_1 - \omega_2 - \omega_3) - \phi(\omega_2) - \phi(\omega_3) \big) \Big] \\
		& + c_{31} \iint_{\omega_1 = \omega_2 + \omega_3} \mathrm{d}\omega_1\, \mathrm{d}\omega_2\, \mathfrak{A}(\omega_1) \mathfrak{A}(\omega_2) \mathfrak{A}(\omega_3)\, f(\omega_1) f(\omega_2) f(\omega_3) \\
		& \quad \times \mathfrak{A}(3\omega_1) \big[ \phi(2\omega_1) - \phi(\omega_1)- \phi(\omega_2)- \phi(\omega_3) \big] \\
		& + c_{31} \iiint_{\mathbb{R}_+^3 \setminus \left( \{\omega_1 > \omega_2 + \omega_3\} \cup \{\omega_2 > \omega_1 + \omega_3\} \cup \{\omega_3 > \omega_1 + \omega_2\} \right)} \mathrm{d}\omega_1\, \mathrm{d}\omega_2\, \mathrm{d}\omega_3 \\
		& \quad \times \mathfrak{A}(\omega_1 + \omega_2 + \omega_3)\, \mathfrak{A}(\omega_1) \mathfrak{A}(\omega_2) \mathfrak{A}(\omega_3)\, f(\omega_1) f(\omega_2) f(\omega_3) \\
		& \quad \times \big[ \phi(\omega_1 + \omega_2 + \omega_3) - \phi(\omega_1) - \phi(\omega_2) - \phi(\omega_3) \big].
	\end{aligned}
\end{equation}

Similar to \eqref{Lemma:Growth2:E5}--\eqref{Lemma:Growth2:E8}, we can estimate
\begin{equation}\label{Lemma:Growth3:E9}
	\mathfrak{G}_1
	\ge -C_0 R_i^{\theta} \phi(0)\left( \int_{[0,\infty)} \mathrm{d}\omega\, G(\omega) \right)^2,
\end{equation}
for some constant $C_0$ independent of $f$ and $R_i$, and which may vary from line to line.

Similarly, as in \eqref{Lemma:Growth2:E9}--\eqref{Lemma:Growth2:E12}, we estimate
\begin{equation}\label{Lemma:Growth3:E10}
	\mathfrak{G}_3
	\ge -C_0 R_i^{\theta}\phi(0) \left( \int_{[0,\infty)} \mathrm{d}\omega\, G(\omega) \right)^3,
\end{equation}
for some constant $C_0$ independent of $f$ and $R_i$.

Similar to \eqref{Lemma:Growth2:E15}--\eqref{Lemma:Growth2:E16}, we deduce that
\begin{equation}\label{Lemma:Growth3:E11}
	\mathfrak{G}_2 \ge 0.
\end{equation}

Combining \eqref{Lemma:Growth3:E5}--\eqref{Lemma:Growth3:E11}, we obtain the bound
\begin{equation} \label{Lemma:Growth3:E12}
	\partial_t \left( \int_{\mathbb{R}^{+}} \mathrm{d}\omega\, \phi(\omega) G(t, \omega) \right)
	\ge -C_0 \phi(0) R_i^{\theta} \left( \mathscr{E} + \mathscr{M} \right)^2 - C_0 \phi(0) R_i^{\theta} \left( \mathscr{E} + \mathscr{M} \right)^3,
\end{equation}
for some constant \( C_0 > 0 \) independent of \( f \) and \( R_i \), where \( \mathscr{E} \) and \( \mathscr{M} \) denote the energy and mass, respectively.

Developing \eqref{Lemma:Growth3:E12}, we estimate, for $t\ge t(\rho)$,
\begin{equation} \label{Lemma:Growth3:E13}
	\begin{aligned}
		& \int_{\mathbb{R}^{+}} \mathrm{d}\omega\,\left[ \exp\left( \frac{1}{R_i} \left(4R_i - \omega\right)_+  \right)- 1\right] G(t, \omega) \\
		\ge\; & \int_{\mathbb{R}^{+}} \mathrm{d}\omega\, \left[\exp\left( \frac{1}{R_i} \left(4R_i - \omega\right)_+  \right) -1\right]G(t(\rho), \omega)\\
		& \  -\ C_0 \phi(0) \left[ R_i^{\theta} \left( \mathscr{E} + \mathscr{M} \right)^2 + R_i^{\theta} \left( \mathscr{E} + \mathscr{M} \right)^3 \right] (t - t(\rho)) \\
		\ge\; & \int_{[0, R_i]} \mathrm{d}\omega\, \left[e^3-1\right] G(t(\rho), \omega) \  - \ C_0 \phi(0) \left[ R_i^{\theta} \left( \mathscr{E} + \mathscr{M} \right)^2 + R_i^{\theta} \left( \mathscr{E} + \mathscr{M} \right)^3 \right] T,
	\end{aligned}
\end{equation}
which, in combination with \eqref{Lemma:Growth3:E3} and \eqref{rho}, implies
\begin{equation} \label{Lemma:Growth3:E14}
	\begin{aligned}
		\int_{[0, 4R_i]} \mathrm{d}\omega\, G(t, \omega)
		\ge\; & \int_{[0, 4R_i]} \mathrm{d}\omega\, \frac{1}{e^4-1}\left[\exp\left( \frac{1}{R_i} \left(4R_i - \omega\right)_+  \right) -1\right]G(t, \omega) \\
		\ge\; & \int_{[0, R_i]} \mathrm{d}\omega\, \frac{e^3-1}{e^4-1} G(t(\rho), \omega) \\
		& \quad - \frac{1}{e^4-1} C_0 \phi(0) \left[ R_i^{\theta} \left( \mathscr{E} + \mathscr{M} \right)^2 + R_i^{\theta} \left( \mathscr{E} + \mathscr{M} \right)^3 \right] T \\
		\ge\; & \frac{e^3-1}{e^4-1} \mathcal{C}_* R_i^\rho - \frac{1}{e^4-1} C_0 \phi(0) \left[ R_i^{\theta} \left( \mathscr{E} + \mathscr{M} \right)^2 + R_i^{\theta} \left( \mathscr{E} + \mathscr{M} \right)^3 \right] T \ 
		\ge\  \frac{1}{3} \mathcal{C}_* R_i^\rho,
	\end{aligned}
\end{equation}
for \( \mathfrak{m}(\rho) \) sufficiently large. Hence, \eqref{Lemma:Growth3:E4} is proved.

\noindent
\textit{Step 3:} 

We define
\[
t_* = \sup_{0< \rho < \min\left\{ \frac{2\delta - \frac{1}{2} - \varrho }{10(2 + \mu + \varrho)},\ \theta, \, 2\delta - \varrho,\ \tfrac{2}{3}\delta,\ \tfrac{\frac{2\delta - \frac{1}{2} - \varrho }{5(2 + \mu + \varrho)} + \varrho}{2} \right\}} t(\rho),
\]
so that \( 0 \le t_* < T_0 \). Moreover,
\begin{equation} \label{Lemma:Growth3:E15}
	\int_{[0, 4R_i]} \mathrm{d}\omega\, G(t, \omega) \ge \frac{\mathcal{C}_*}{3} R_i^\rho = \frac{\mathcal{C}_*}{4^\rho \cdot 3} (4R_i)^\rho,
\end{equation}
for all \( t > t_*,\ i \ge \mathfrak{m}(\rho)\).

We fix  \( \rho \), and there exists \( \mathfrak{m}^* = \mathfrak{m}^*(\rho) \) such that
\begin{equation} \label{Lemma:Growth3:E16}
	\int_{[0, R_j]} \mathrm{d}\omega\, G(t, \omega) > \mathcal{C}_*' |R_j|^\rho,
\end{equation}
for all \( j > \mathfrak{m}^*(\rho) \) and for all \( t_* < t < T_0 \), where \( \mathcal{C}_*' \) is a positive constant independent of \( j \) and \( t \).

For \( t_* < T_1 < T_0 \), \( j > \mathfrak{m}^*(\rho) \), and \( t_* < t < T_1 \), the goal of this step is to prove the existence of an unbounded set \( \mathscr{S}(t) \) such that for all \( j \in \mathscr{S}(t) \),
\begin{equation} \label{Lemma:Growth3:E17}
	\int_{[R_{j+1}, R_j)} \mathrm{d}\omega\, G(t, \omega) > \frac{2^\rho - 1}{2^\rho} \cdot \frac{\mathcal{C}_*'}{2} |R_j|^\rho,
\end{equation}
while for \( j \notin \mathscr{S}(t) \), inequality \eqref{Lemma:Growth3:E17} does not hold.

This can be shown by a contradiction argument.  
Suppose there exists \( \mathfrak{M}_0 > 0 \) such that for all \( j \ge \mathfrak{M}_0 \),
\[
\int_{[R_{j+1}, R_j)} \mathrm{d}\omega\, G(t, \omega) \le \frac{2^\rho - 1}{2^\rho} \cdot \frac{\mathcal{C}_*'}{2} |R_j|^\rho.
\]
Then summing over \( j \ge \mathfrak{M}_0 \), we obtain
\begin{equation*}
	\int_{\bigcup_{j = \mathfrak{M}_0}^\infty [R_{j+1}, R_j)} \mathrm{d}\omega\, G(t, \omega) 
	\le  \frac{2^\rho - 1}{2^\rho} \cdot \frac{\mathcal{C}_*'}{2} \sum_{j = \mathfrak{M}_0}^\infty |R_j|^\rho.
\end{equation*}
Since \( R_j = R_{\mathfrak{M}_0} \cdot 2^{-j + \mathfrak{M}_0} \), it follows that
\[
|R_j|^\rho = |R_{\mathfrak{M}_0}|^\rho \cdot 2^{-(j - \mathfrak{M}_0)\rho},
\]
and hence,
\[
\sum_{j = \mathfrak{M}_0}^\infty |R_j|^\rho = |R_{\mathfrak{M}_0}|^\rho \sum_{j = 0}^\infty 2^{-j\rho} 
= |R_{\mathfrak{M}_0}|^\rho \cdot \frac{1}{1 - 2^{-\rho}} = |R_{\mathfrak{M}_0}|^\rho \cdot \frac{2^\rho}{2^\rho - 1}.
\]
Therefore,
\[
\int_{\bigcup_{j = \mathfrak{M}_0}^\infty [R_{j+1}, R_j)} \mathrm{d}\omega\, G(t, \omega) 
\le |R_{\mathfrak{M}_0}|^\rho \cdot \frac{\mathcal{C}_*'}{2}.
\]

Thus,
\begin{equation*}
	\int_{\{0\}} \mathrm{d}\omega\, G(t, \omega) \ge |R_{\mathfrak{M}_0}|^\rho \cdot \frac{\mathcal{C}_*'}{2}.
\end{equation*}
This contradicts our original assumption on $T_0$ (see \eqref{T0}), and therefore confirms the existence of the unbounded set \( \mathscr{S}(t) \).

\noindent
\textit{Step 4:} 

For \( \mathfrak{m} \in \mathbb{N} \), we define \( \mathscr{S}_{\mathfrak{m}} \) to be the set of all \( t \in (t_*, T_1) \) such that \( \mathfrak{m} \in \mathscr{S}(t) \).  
That means,
\[
\bigcup_{\mathfrak{m} \in \mathbb{N}} \mathscr{S}_{\mathfrak{m}} = (t_*, T_1).
\]

Let $\varkappa\in(0,1)$ be a constant to be fixed later. We now show that there exists \( \mathfrak{n}_0 \in \mathbb{N} \) such that
 \[
 |\mathscr{S}_{\mathfrak{n}_0}| \ge |T_1 - t_*| \cdot 2^{-\mathfrak{n}_0 \varkappa} \cdot \frac{1}{2} \cdot \frac{2^\varkappa - 1}{2^\varkappa}
 \]
 by a contradiction argument. Suppose, on the contrary, that
 \[
 |\mathscr{S}_{n}| \le |T_1 - t_*| \cdot 2^{-n \varkappa} \cdot \frac{1}{2} \cdot \frac{2^\varkappa - 1}{2^\varkappa}
 \quad \text{for all } n \in \mathbb{N}.
 \]
 Then, summing over \( n \), we obtain
 \[
 \left| \bigcup_{n \in \mathbb{N}} \mathscr{S}_n \right| \le \frac{1}{2} \cdot \frac{2^\varkappa - 1}{2^\varkappa} \cdot |T_1 - t_*| \cdot \sum_{n=0}^\infty 2^{-n \varkappa}
 = \frac{1}{2} \cdot \frac{2^\varkappa - 1}{2^\varkappa} \cdot |T_1 - t_*| \cdot \frac{1}{1 - 2^{-\varkappa}}.
 \]
 This simplifies to
 \[
 \left| \bigcup_{n \in \mathbb{N}} \mathscr{S}_n \right| \le |T_1 - t_*| \cdot \frac{1}{2} < |T_1 - t_*| =  \left| \bigcup_{n \in \mathbb{N}} \mathscr{S}_n \right|,
 \]
 a contradiction. Therefore, there must exist \( \mathfrak{n}_0 \in \mathbb{N} \) such that
 \[
 |\mathscr{S}_{\mathfrak{n}_0}| \ge |T_1 - t_*| \cdot 2^{-\mathfrak{n}_0 \varkappa} \cdot \frac{1}{2} \cdot \frac{2^\varkappa - 1}{2^\varkappa}.
 \]
 
 Since for each \( t \in (t_*, T_1) \), the set \( \mathscr{S}(t) \) is unbounded, the subset \( \{n \in \mathscr{S}(t) \mid n > \mathfrak{n}_0\} \) is also unbounded. As a result,
 \[
 \bigcup_{n \in \mathbb{N},\, n > \mathfrak{n}_0} \mathscr{S}_n = (t_*, T_1).
 \]
 Applying the same argument, we find \( \mathfrak{n}_1 > \mathfrak{n}_0 \) such that
 \[
 |\mathscr{S}_{\mathfrak{n}_1}| \ge |T_1 - t_*| \cdot 2^{-\mathfrak{n}_1 \varkappa} \cdot \frac{1}{2} \cdot \frac{2^\varkappa - 1}{2^\varkappa}.
 \]
 
 Repeating this process, we construct an unbounded sequence \( \{ \mathfrak{n}_0, \mathfrak{n}_1, \dots \} \subset \mathbb{N} \) such that for each \( \mathfrak{n} \) in the sequence,
 \[
 |\mathscr{S}_{\mathfrak{n}}| \ge |T_1 - t_*| \cdot 2^{-\mathfrak{n} \varkappa} \cdot \frac{1}{2} \cdot \frac{2^\varkappa - 1}{2^\varkappa}.
 \]

Let \( \mathfrak{n} \in \{ \mathfrak{n}_0, \mathfrak{n}_1, \dots \} \), where \( \mathfrak{n} \) is large enough, by the pigeonhole principle, there exists a set \( \mathscr{O}_{i}^{h_{\mathfrak{n}}, R_{\mathfrak{n}}} \subset \left[ R_{\mathfrak{n}+1}, R_{\mathfrak{n}} \right) \) with the following property:
\begin{equation} \label{Lemma:Growth3:E18}
	\int_{\mathscr{O}_{i}^{h_{\mathfrak{n}}, R_{\mathfrak{n}}}} \mathrm{d}\omega \, G(t, \omega) >  \frac{2^\rho - 1}{2^\rho} \cdot \frac{\mathcal{C}_*'}{100} \frac{|R_{\mathfrak{n}}|^\rho}{\mathfrak{N}_\mathfrak{n}} \ge \mathcal{C}_*'' R_\mathfrak{n}^{\rho + \frac{2\delta - \frac{1}{2} - \varrho }{4(2 + \mu + \varrho)}}.
\end{equation}
Here, \( \mathcal{C}_*'' > 0 \) is a constant that may vary from line to line, and \( t \in \mathscr{S}_\mathfrak{n} \). This yields
\begin{equation} \label{Lemma:Growth3:E19}
	\int_{t_*}^{T_1} \mathrm{d}t \, \chi_{\mathscr{S}_\mathfrak{n}}(t) \left[ \int_{\mathscr{O}_{i}^{h_{\mathfrak{n}}, R_{\mathfrak{n}}}} \mathrm{d}\omega \, G(t, \omega) \right]^2 > \mathcal{C}_*''' R_\mathfrak{n}^{2\rho + \frac{2\delta - \frac{1}{2} - \varrho }{2(2 + \mu + \varrho)} + \varkappa},
\end{equation}
where \( \mathcal{C}_*''' = \mathcal{C}_*''  |T_1 - t_*|  \cdot \frac{1}{2} \cdot \frac{2^\varkappa - 1}{2^\varkappa}\).

By \eqref{Parameters},  we obtain the bound
\[
\frac{2\delta - \frac{1}{2} - \varrho}{2(2 + \mu + \varrho)} < \delta.
\]

By \eqref{Parameters}, we can choose a sufficiently small \( \varkappa > 0 \) such that \( \varkappa + \varepsilon < \delta \). Then we have
\[
2\rho + \frac{2\delta - \frac{1}{2} - \varrho }{2(2 + \mu + \varrho)} + \varkappa 
< 2\delta + 2\rho - \varepsilon + \varrho.
\]

Therefore, for the fixed constant $\mathcal{C}_*''' $, there must exist a time $T_2$, \( t_* < T_2 < T_1 \) such that
\begin{equation} \label{Lemma:Growth3:E20}
	\int_{t_*}^{T_2} \mathrm{d}t\, \chi_{\mathscr{S}_\mathfrak{n}}(t) \left[ \int_{\mathscr{O}_{i}^{h_{\mathfrak{n}}, R_{\mathfrak{n}}}} \mathrm{d}\omega\, F(t, \omega) \right]^2  
	= \frac{\mathcal{C}_*'''  R_{\mathfrak{n}}^{2\delta + 2\rho - \varepsilon}}{ \Theta(R_{\mathfrak{n}} / 2)}.
\end{equation}

Similarly to \eqref{Lemma:Growth2:E24}, we bound for \( t \in (t_*, T_1) \):
\begin{equation} \label{Lemma:Growth3:E21}
	\begin{aligned}
		\partial_t \left( \int_{\mathbb{R}^{+}} \mathrm{d}\omega\, G \Psi_t \right)
		\ge\ & - C_0 R_{\mathfrak{n}}^{\theta} \Psi_t(0) \mathscr{M}^2
		- C_0 R_{\mathfrak{n}}^{\theta} \Psi_t(0) \mathscr{M}^3 \\
		& + 2^{-4\mu} c_{22} C_{\text{disper}}^{2\delta} R_{\mathfrak{n}}^{-2\delta}
		\left[ \int_{\mathbb{R}_+} \mathrm{d}\omega\, G(\omega) \Psi_t(\omega) \right]
		\left[ \int_{\mathbb{R}_+} \mathrm{d}\omega\, G(\omega)
		\chi_{\left\{ \omega \in \mathscr{O}_{\sigma(t)}^{h_{\mathfrak{n}}, R_{\mathfrak{n}}} \right\}}
		\chi_{\mathscr{S}_\mathfrak{n}}(t) \right]^2.
	\end{aligned}
\end{equation}

By the same argument as in \eqref{Lemma:Growth2:E30}, and combining \eqref{Lemma:Growth3:E20}--\eqref{Lemma:Growth3:E21}, we deduce a contradiction as \( \mathfrak{n} \to \infty \). Hence, \eqref{Lemma:Growth3:E1} holds, and \eqref{Lemma:Growth3:1} follows from \eqref{Lemma:Growth1:2}.

\end{proof}

\section{Condensation growth}\label{Sec:CondensateGrowth}
\begin{proposition}[Immediate Condensation]
	\label{Lemma:Growth4}  	We assume Assumption A and Assumption B. 
 Suppose
\begin{equation}\label{Lemma:Growth4:1} 
0<	c_{\mathrm{ini}} <  \min\left\{\frac{2\delta - \frac{1}{2} - \varrho }{10(2 + \mu + \varrho)},\,\theta,\,{2\delta - \varrho},\,\frac23\delta,\,\frac{\frac{2\delta - \frac{1}{2} - \varrho }{5(2 + \mu + \varrho)}+\varrho}{2}\right\}
\end{equation}

	then $$	T_0 =0.$$
\end{proposition}

\begin{proof}
	We choose
	\begin{equation}\label{Lemma:Growth4:E1} 
		\begin{aligned}
			& 0 < c_{\mathrm{ini}} < \rho < \min\left\{\frac{2\delta - \frac{1}{2} - \varrho }{10(2 + \mu + \varrho)},\ {2\delta - \varrho},\, \theta,\, \tfrac{2}{3}\delta,\ \tfrac{\frac{2\delta - \frac{1}{2} - \varrho }{5(2 + \mu + \varrho)} + \varrho}{2} \right\}, \\
			& \max\{0,\ 2\rho - \varrho\} < \varepsilon < \min\left\{ 2\delta - \rho - \varrho,\ \frac{2\delta - \frac{1}{2} - \varrho }{5(2 + \mu + \varrho)} \right\}.
		\end{aligned}
	\end{equation}
	
Hence 	\begin{equation}\label{Lemma:Growth4:E2} \varepsilon \ < \  2\delta - \rho -\varrho  \ < \  2\delta -\varrho - c_{\mathrm{ini}}.	\end{equation}
Following \eqref{Lemma:Growth3:E5}-\eqref{Lemma:Growth3:E8}, we use $\phi(\omega) = \exp\left[\frac{4}{R_\mathfrak m}\left({R_{\mathfrak m}} - \omega\right)_+\right] - 1$ as a test function and find
	\begin{equation*} 
		\partial_t \left( \int_{\mathbb{R}^{+}} \mathrm{d}\omega\, \phi(\omega) G(t, \omega) \right)
		= \mathfrak{G}_1 + \mathfrak{G}_2 + \mathfrak{G}_3,
	\end{equation*}
	where
	\begin{equation*} 
		\begin{aligned}
			\mathfrak{G}_1 :=\; & 2c_{12} \iint_{\omega_1 > \omega_2} \mathrm{d}\omega_1 \mathrm{d}\omega_2\, \mathfrak{A}(\omega_1) \mathfrak{A}(\omega_2)\, f(\omega_1) f(\omega_2) \\
			& \quad \times \Big[ \mathfrak{A}(\omega_1 + \omega_2) \big( \phi(\omega_1 + \omega_2) - \phi(\omega_1) - \phi(\omega_2) \big) \\
			& \qquad\quad - \mathfrak{A}(\omega_1 - \omega_2) \big( \phi(\omega_1) - \phi(\omega_1 - \omega_2) - \phi(\omega_2) \big) \Big] \\
			& + c_{12} \iint_{\omega_1 = \omega_2} \mathrm{d}\omega_1 \mathrm{d}\omega_2\, \mathfrak{A}(\omega_1)^2 f(\omega_1)^2\, \mathfrak{A}(2\omega_1)\, \big[ \phi(2\omega_1) - 2\phi(\omega_1) \big],
		\end{aligned}
	\end{equation*}
	
	\begin{equation*} \label{Lemma:Growth4:E3}
		\begin{aligned}
			\mathfrak{G}_2 :=\; & c_{22} \iiint_{\mathbb{R}_+^3} \mathrm{d}\omega_1\,\mathrm{d}\omega_2\,\mathrm{d}\omega\, f(\omega_1) f(\omega_2) f(\omega)\, \mathbf{1}_{\omega+\omega_1-\omega_2 \ge 0} \\
			& \times \left[ \frac{\max\{|\omega_1 - \omega_2|,\ |\omega_2 - \omega|\}}{2(\omega + \omega_1)} \right]^\mu \cdot \Big\{ 
			[-\phi(\omega_{\text{Max}}) - \phi(\omega_{\text{Min}}) + \phi(\omega_{\text{Mid}}) \\
			& \quad + \phi(\omega_{\text{Max}} + \omega_{\text{Min}} - \omega_{\text{Mid}})] \prod_{x \in \{\omega_{\text{Max}}, \omega_{\text{Min}}, \omega_{\text{Mid}}, \omega_{\text{Max}} + \omega_{\text{Min}} - \omega_{\text{Mid}}\}}\Theta(x)\, |k_{\text{Min}}| \\
			& + [-\phi(\omega_{\text{Max}}) - \phi(\omega_{\text{Mid}}) + \phi(\omega_{\text{Min}}) + \phi(\omega_{\text{Max}} + \omega_{\text{Mid}} - \omega_{\text{Min}})] \\
			& \quad \times \prod_{x \in \{\omega_{\text{Max}}, \omega_{\text{Mid}}, \omega_{\text{Min}}, \omega_{\text{Max}} + \omega_{\text{Mid}} - \omega_{\text{Min}}\}}\Theta(x)\, |k_{\text{Min}}| \\
			& + [-\phi(\omega_{\text{Min}}) - \phi(\omega_{\text{Mid}}) + \phi(\omega_{\text{Max}}) + \phi(\omega_{\text{Min}} + \omega_{\text{Mid}} - \omega_{\text{Max}})] \\
			& \quad \times \mathbf{1}_{\omega_{\text{Min}} + \omega_{\text{Mid}} - \omega_{\text{Max}} \ge 0} \prod_{x \in \{\omega_{\text{Max}}, \omega_{\text{Mid}}, \omega_{\text{Min}}, \omega_{\text{Min}} + \omega_{\text{Mid}} - \omega_{\text{Max}}\}}\Theta(x) \\
			& \quad \times \min \left\{ |k(\omega_{\text{Max}})|, |k(\omega_{\text{Min}})|, |k(\omega_{\text{Mid}})|, |k(\omega_{\text{Min}} + \omega_{\text{Mid}} - \omega_{\text{Max}})| \right\}
			\Big\}
		\end{aligned}
	\end{equation*}
	
	and
	\begin{equation*} \label{Lemma:Growth4:E4}
		\begin{aligned}
			\mathfrak{G}_3 :=\; & 3c_{31} \iiint_{\omega_1 > \omega_2 + \omega_3} \mathrm{d}\omega_1\, \mathrm{d}\omega_2\, \mathrm{d}\omega_3\, \mathfrak{A}(\omega_1) \mathfrak{A}(\omega_2) \mathfrak{A}(\omega_3)\, f(\omega_1) f(\omega_2) f(\omega_3) \\
			& \times \Big[ \mathfrak{A}(\omega_1 + \omega_2 + \omega_3) \big( \phi(\omega_1 + \omega_2 + \omega_3) - \phi(\omega_1) - \phi(\omega_2) - \phi(\omega_3) \big) \\
			& \quad - \mathfrak{A}(\omega_1 - \omega_2 - \omega_3) \big( \phi(\omega_1) - \phi(\omega_1 - \omega_2 - \omega_3) - \phi(\omega_2) - \phi(\omega_3) \big) \Big] \\
			& + c_{31} \iint_{\omega_1 = \omega_2 + \omega_3} \mathrm{d}\omega_1\, \mathrm{d}\omega_2\, \mathfrak{A}(\omega_1) \mathfrak{A}(\omega_2) \mathfrak{A}(\omega_3)\, f(\omega_1) f(\omega_2) f(\omega_3) \\
			& \quad \times \mathfrak{A}(3\omega_1) \big[ \phi(3\omega_1) - 3\phi(\omega_1) \big] \\
			& + c_{31} \iiint_{\mathbb{R}_+^3 \setminus \left( \{\omega_1 > \omega_2 + \omega_3\} \cup \{\omega_2 > \omega_1 + \omega_3\} \cup \{\omega_3 > \omega_1 + \omega_2\} \right)} \mathrm{d}\omega_1\, \mathrm{d}\omega_2\, \mathrm{d}\omega_3 \\
			& \quad \times \mathfrak{A}(\omega_1 + \omega_2 + \omega_3)\, \mathfrak{A}(\omega_1) \mathfrak{A}(\omega_2) \mathfrak{A}(\omega_3)\, f(\omega_1) f(\omega_2) f(\omega_3) \\
			& \quad \times \big[ \phi(\omega_1 + \omega_2 + \omega_3) - \phi(\omega_1) - \phi(\omega_2) - \phi(\omega_3) \big].
		\end{aligned}
	\end{equation*}
Similar to \eqref{Lemma:Growth3:E13}-\eqref{Lemma:Growth3:E14}, we estimate  
\begin{equation*} \label{Lemma:Growth4:E5}
	\begin{aligned}
		& \int_{\mathbb{R}^{+}} \mathrm{d}\omega\, \left[\exp\left( \frac{1}{R_{\mathfrak m}} \left(4R_{\mathfrak m} - \omega\right)_+ \right) - 1\right] G(t, \omega) \\
		\ge\; & \int_{\mathbb{R}^{+}} \mathrm{d}\omega\, \left[\exp\left( \frac{1}{R_{\mathfrak m}} \left(4R_{\mathfrak m} - \omega\right)_+ \right) - 1\right] G(0, \omega) \ -\ C_0 \phi(0) \left[ R_{\mathfrak m}^{\theta} \left( \mathscr{E} + \mathscr{M} \right)^2 + R_{\mathfrak m}^{\theta} \left( \mathscr{E} + \mathscr{M} \right)^3 \right] t \\
		\ge\; & \int_{[0, R_{\mathfrak m}]} \mathrm{d}\omega\, \left[e^3-1\right] G(0, \omega) \ - \  C_0 \phi(0) \left[ R_{\mathfrak m}^{\theta} \left( \mathscr{E} + \mathscr{M} \right)^2 + R_{\mathfrak m}^{\theta} \left( \mathscr{E} + \mathscr{M} \right)^3 \right] t \\
		\ge\; &  \left[e^3-1\right] C_{\mathrm{ini}} R_{\mathfrak m}^{c_{\mathrm{ini}}} \ - \  C_0 \phi(0) \left[ R_\mathfrak m^{\theta} \left( \mathscr{E} + \mathscr{M} \right)^2 + R_\mathfrak m^{\theta} \left( \mathscr{E} + \mathscr{M} \right)^3 \right] t,
	\end{aligned}
\end{equation*}

which implies
\begin{equation*} \label{Lemma:Growth4:E6}
	\begin{aligned}
		\int_{[0, 4R_{\mathfrak m}]} \mathrm{d}\omega\, G(t, \omega)
		\ge\; & \int_{[0, 4R_{\mathfrak m}]} \mathrm{d}\omega\, \frac{1}{e^4-1} \left[\exp\left( \frac{1}{R_{\mathfrak m}} \left(4R_{\mathfrak m} - \omega\right)_+ \right) - 1\right]  \\
		\ge\; &  \frac{e^3-1}{e^4-1}  C_{\mathrm{ini}} R_{\mathfrak m}^{c_{\mathrm{ini}}} \  - \ \frac{1}{e^4-1}  C_0 \phi(0) \left[ R_\mathfrak m^{\theta} \left( \mathscr{E} + \mathscr{M} \right)^2 + R_{\mathfrak m}^{\theta} \left( \mathscr{E} + \mathscr{M} \right)^3 \right] t \\
		\ge\; & \frac{1}{2e} C_{\mathrm{ini}} R_{\mathfrak m}^{c_{\mathrm{ini}}} \ge \frac{1}{2e}  C_{\mathrm{ini}}4^{-\rho} R_{\mathfrak m-2}^{\rho},
	\end{aligned}
\end{equation*}

for
\begin{equation*} \label{Lemma:Growth4:E7}
	\begin{aligned}
		0 \le t \le \frac{( e^4-2e+1) C_{\mathrm{ini}} R_{\mathfrak m}^{c_{\mathrm{ini}}}}{ 2e(e^4-1) C_0  \left[ R_\mathfrak m^{\theta} (\mathscr{E} + \mathscr{M})^2 + R_\mathfrak m^{\theta} (\mathscr{E} + \mathscr{M})^3 \right]}.
	\end{aligned}
\end{equation*}
Here, $\mathfrak{C}_0$ denotes a universal constant that is independent of both $\mathfrak m$ and $G$.

Using the definition \eqref{Sec:Growthlemmas:2}, we obtain, for $\mathcal{C}_* = \frac{1}{2e} C_{\mathrm{ini}} 4^{-\rho}$,
\begin{equation*} \label{Lemma:Growth4:E8}
	\mathcal{A}_{\mathfrak m-2}^T = [0, T],
\end{equation*}
for all
\[
T \in \left[0, \min\left\{T_0,  \frac{( e^4-2e+1) C_{\mathrm{ini}} R_{\mathfrak m}^{c_{\mathrm{ini}}}}{ 2e(e^4-1) C_0  \left[ R_\mathfrak m^{\theta} (\mathscr{E} + \mathscr{M})^2 + R_\mathfrak m^{\theta} (\mathscr{E} + \mathscr{M})^3 \right]}\right\}\right).
\]

By Proposition \ref{Lemma:Growth3}, we bound
\begin{equation*} \label{Lemma:Growth4:E9}
	\left\vert \mathcal{A}_{\mathfrak m-2}^T \right\vert \leq C_\mathcal{A} \left( R_
	{\mathfrak{m}-2} \right)^{\min\left\{ \tfrac{1}{4}(2\delta  - \tfrac{1}{2} - \varrho),\, 2\delta - \varepsilon - \varrho \right\}},
\end{equation*}
for all
\[
T \in \left[0, \min\left\{T_0,\ \frac{ C_{\mathrm{ini}} R_{\mathfrak m}^{c_{\mathrm{ini}}}}{2e {C}_0   \left[ R_\mathfrak m^{\theta} (\mathscr{E} + \mathscr{M})^2 + R_\mathfrak m^{\theta} (\mathscr{E} + \mathscr{M})^3 \right]} \right\} \right).
\]

Therefore,
\begin{equation*}
	\begin{aligned}
		\min\left\{T_0,  \frac{( e^4-2e+1) C_{\mathrm{ini}} R_{\mathfrak m}^{c_{\mathrm{ini}}}}{ 2e(e^4-1) C_0  \left[ R_\mathfrak m^{\theta} (\mathscr{E} + \mathscr{M})^2 + R_\mathfrak m^{\theta} (\mathscr{E} + \mathscr{M})^3 \right]}\right\}
		\leq\ C_\mathcal{A} \left( R_{\mathfrak m-2} \right)^{\min\left\{ \tfrac{1}{4}(2\delta  - \tfrac{1}{2} - \varrho),\, 2\delta - \varepsilon - \varrho \right\}}.
	\end{aligned}
\end{equation*}

By \eqref{Lemma:Growth4:E1}-\eqref{Lemma:Growth4:E2} 
\[
c_{\mathrm{ini}} < \min\left\{ \tfrac{1}{4}(2\delta  - \tfrac{1}{2} - \varrho),\, 2\delta - \varepsilon - \varrho \right\},
\]
we deduce that, for sufficiently large $\mathfrak m$,
\begin{equation*}
	\begin{aligned}
		T_0 \leq C_\mathcal{A} \left( R_{\mathfrak m-2} \right)^{\min\left\{ \tfrac{1}{4}(2\delta  - \tfrac{1}{2} - \varrho),\, 2\delta - \varepsilon - \varrho \right\}} \to 0, 
	\end{aligned}
\end{equation*}
as $\mathfrak m\to\infty$.

\end{proof}

\begin{proposition}[Finite-Time Condensation] 
	\label{Lemma:Growth5} 	We assume Assumption A and Assumption B. 
  Let  
	\begin{equation}\label{Lemma:Growth5:2} 
		0<	c_{\mathrm{ini}} <  \theta,
	\end{equation}	
	
	then
	$$	T_0 <\infty.$$

\end{proposition}

\begin{proof}
Let \( \rho^o, \varepsilon^o \) be sufficiently small constants satisfying
\begin{equation}\label{Lemma:Growth5:3} 
	0 < \rho^o < \min\left\{ \frac{2\delta - \tfrac{1}{2} - \varrho}{10(2 + \mu + \varrho)},\ 2\delta - \varrho,\ \theta,\ \tfrac{2}{3}\delta,\ \tfrac{\frac{2\delta - \tfrac{1}{2} - \varrho}{5(2 + \mu + \varrho)} + \varrho}{2} \right\}.
\end{equation}

We choose a sufficiently large natural number \( \mathcal{N}_0 \) and a constant \( \mathcal{C}_*^o > 0 \) such that
\begin{equation}  \label{Lemma:Growth5:1} 
	\int_{0}^{R_{\mathcal{N}_0}} \mathrm{d}\omega\, G(0, \omega) \ge \mathcal{C}_*^o R_{\mathcal{N}_0}^{\rho^o} = C_{\mathrm{ini}}\, R_{\mathcal{N}_0}^{c_{\mathrm{ini}}}.
\end{equation}

We now set
\[
\max\{0,\ 2\rho^o - \varrho\} < \varepsilon < \min\left\{ 2\delta - \varrho - \rho^o,\ \frac{2\delta - \tfrac{1}{2} - \varrho}{5(2 + \mu + \varrho)} \right\}.
\]

Setting $\mathcal{C}_*=\mathcal{C}_*^o$, $\rho   =\rho^o$, similar to the proof of Proposition \ref{Lemma:Growth4}, we obtain
\begin{equation*} 
	\mathcal{A}_{\mathcal{N}_0}^T = [0, T],
\end{equation*}
for all
\[
T \in \left[0, \min\left\{ T_0,\ \frac{\mathcal{C}_*( e^4-2e+1) R_{\mathcal{N}_0}^\rho}{2e(e^4-1){C}_0 \left[ R_{\mathcal{N}_0}^{\theta} (\mathscr{E} + \mathscr{M})^2 + R_{\mathcal{N}_0}^{\theta} (\mathscr{E} + \mathscr{M})^3 \right]} \right\} \right).
\]

By Proposition \ref{Lemma:Growth3}, we bound
\begin{equation*} \label{Lemma:Growth4:E9}
	\left\vert \mathcal{A}_{\mathcal{N}_0}^T \right\vert \leq C_\mathcal{A} \left( R_{\mathcal{N}_0} \right)^{\min\left\{ \tfrac{1}{4}(2\delta  - \tfrac{1}{2} - \varrho),\ 2\delta - \varepsilon - \varrho \right\}},
\end{equation*}
for all
\[
T \in \left[0, \min\left\{ T_0,\ \frac{\mathcal{C}_*( e^4-2e+1) R_{\mathcal{N}_0}^\rho}{2e(e^4-1){C}_0 \left[ R_{\mathcal{N}_0}^{\theta} (\mathscr{E} + \mathscr{M})^2 + R_{\mathcal{N}_0}^{\theta} (\mathscr{E} + \mathscr{M})^3 \right]} \right\} \right).
\]

Therefore,
\begin{equation*}
	\begin{aligned}
	\min\left\{ T_0,\ \frac{\mathcal{C}_*( e^4-2e+1) R_{\mathcal{N}_0}^\rho}{2e(e^4-1){C}_0 \left[ R_{\mathcal{N}_0}^{\theta} (\mathscr{E} + \mathscr{M})^2 + R_{\mathcal{N}_0}^{\theta} (\mathscr{E} + \mathscr{M})^3 \right]} \right\}
		\leq C_\mathcal{A} \left( R_{\mathcal{N}_0} \right)^{\min\left\{ \tfrac{1}{4}(2\delta - \tfrac{1}{2} - \varrho),\ 2\delta - \varepsilon - \varrho \right\}}.
	\end{aligned}
\end{equation*}

Since
\[
\rho < \min\left\{ \tfrac{1}{4}(2\delta  - \tfrac{1}{2} - \varrho),\ 2\delta - \varepsilon - \varrho \right\},
\]
we deduce that, for sufficiently large $\mathcal{N}_0$,
\begin{equation*}
	\begin{aligned}
		T_0 \leq C_\mathcal{A} \left( R_{\mathcal{N}_0} \right)^{\min\left\{ \tfrac{1}{4}(2\delta  - \tfrac{1}{2} - \varrho),\ 2\delta - \varepsilon - \varrho \right\}} < \infty.
	\end{aligned}
\end{equation*}

\end{proof}
\section{Proof of the main Theorem \ref{Theorem1}}\label{Sec:Proof}
The global existence result follows from Proposition \ref{Lemma:Global}. Item (i) of the main theorem follows from Proposition \ref{Lemma:Growth4} and Item (ii) follows from Proposition \ref{Lemma:Growth5}.

\bibliographystyle{plain}

\bibliography{WaveTurbulence}

\end{document}